\documentclass[a4paper, english, 12pt]{article}
\usepackage[utf8]{inputenc}
\usepackage[T1]{fontenc}
\usepackage[final]{microtype}
\usepackage{babel}
\usepackage[bottom=2.5cm, left=1.8cm, right=1.8cm, top=2.5cm]{geometry}

\usepackage{fancyhdr} 
\usepackage[hidelinks]{hyperref}

\usepackage{parskip} %no indent

\usepackage{graphicx} 
\usepackage[table,xcdraw,svgnames]{xcolor}
\usepackage{enumitem}
\usepackage{subcaption}
\usepackage[style=american]{csquotes}
\usepackage{booktabs}

\usepackage{amsmath, amssymb} 
\usepackage{dsfont, amsfonts, bbold, bm}
\usepackage{multirow}
\usepackage{paracol}
\usepackage{float}

\usepackage{algorithm, algpseudocode}

\usepackage{amsthm}
\theoremstyle{plain}
\newtheorem{theorem}{\\\bf Theorem}

\newtheorem{theoremA}{Theorem}[section]
\newtheorem{lemmaA}[theoremA]{Lemma}
\newtheorem{corollary}[theoremA]{Corollary}
\newtheorem*{theorem1}{Theorem 1}
\newtheorem*{theorem2}{Theorem 2}

\theoremstyle{definition}
\newtheorem{assumption}{\\\bf Assumption}

\newcommand{\addappendix}[1]{%
  \section*{#1}% start the appendix
  \addcontentsline{toc}{section}{#1}% add it to the toc
  \counterwithin*{figure}{section}% optional, if you want to reset the figure counter
  \stepcounter{section}% reset counters related to section
  \renewcommand{\thesection}{\Alph{section}}% we want A
  \renewcommand{\thefigure}{\thesection.\arabic{figure}}% optional
  \renewcommand\theequation{\thesection.\arabic{equation}}
}
\newcommand{\appref}[1]{Appendix~\ref{#1}}
\newcommand{\subsubsubsection}[1]{{\bf #1}}

\usepackage[
    articlein = false,
    backend = biber,
    bibstyle = ext-authoryear,
    sorting = nyt,     %sorts entries by name, year, title
    maxbibnames = 10, %if a list holds more than [maxnames] names, 
    minbibnames = 10,  %it is automatically truncated to [minnames] names
    date = year, %show only year
    giveninits, %only render the given name initials and not the full name
    citestyle = ext-authoryear-comp,
    maxcitenames = 2,
    mincitenames = 1,
    uniquename = false, %eg: when the surname is not unique, writes Sun (2010) instead of X. Sun (2010)
    uniquelist = false %avoid adding second author when first is repeated
]{biblatex}
\defbibfilter{appendixOnlyFilter}{
segment=1
and not segment=0
}

\DeclareFieldFormat{pages}{#1}

\DeclareFieldFormat[article]{title}{#1\addperiod}

\DeclareCiteCommand*{\citeyear}
    {}
    {\bibhyperref{\printfield{year}}\bibhyperref{\printfield{extradate}}}
    {\multicitedelim}
    {}
\DeclareDelimFormat{nameyeardelim}{\addcomma\space}
\renewcommand*\cite[1]{\citeauthor{#1} (\citeyear*{#1})}
\let\citep\parencite

\title{Semiparametric estimation of GLMs with interval-censored covariates via an augmented Turnbull estimator}

\author{Andrea Toloba$^*$, Klaus Langohr, and Guadalupe Gómez Melis \\\normalsize
Statistics and Operations Research Department and Institute for Research \\\normalsize
and Innovation in Health (IRIS),
Universitat Politècnica de Catalunya\textperiodcentered BarcelonaTech, Spain. \\\normalsize
$^*$andrea.toloba@upc.edu
}

\begin{document}

\maketitle

\textbf{\large Abstract}\\ 
Interval-censored covariates are frequently encountered in biomedical studies, particularly in time-to-event data or when measurements are subject to detection or quantification limits. Yet, the estimation of regression models with interval-censored covariates remains methodologically underdeveloped. In this article, we address the estimation of generalized linear models when one covariate is subject to interval censoring. We propose a likelihood-based approach, GELc, that builds upon an augmented version of Turnbull's nonparametric estimator for interval-censored data. We prove that the GELc estimator is consistent and asymptotically normal under mild regularity conditions, with available standard errors. Simulation studies demonstrate favorable finite-sample performance of the estimator and satisfactory coverage of the confidence intervals. Finally, we illustrate the method using two real-world applications: the AIDS Clinical Trials Group Study 359 and an observational nutrition study on circulating carotenoids. The proposed methodology is available as an R package at \href{https://github.com/atoloba/ICenCov}{\texttt{github.com/atoloba/ICenCov}}.


\newpage
% ----------------------------------------------------------------------------Intro
\section{Introduction}
\label{s:intro}

Interval censoring occurs when the value of a variable is only known to lie within an interval instead of being observed exactly. 
It is common in time-to-event data when the event of interest can only be monitored at regular time points because, for example, assessment is conducted at scheduled medical examinations or through laboratory testing. In the AIDS Clinical Trials Group Study 359 (ACTG359), patients receiving the antiretroviral therapy with indinavir were visited periodically to monitor their HIV viral load levels. When a patient's viral load exceeded 500 copies/mL, indinavir treatment was considered to have failed.
Time to indinavir failure was, therefore, interval-censored between the last examination below 500 copies/mL and the first one above the threshold \citep{gulick2000,gomez2003}.

Methods for the analysis of interval-censored response variables have been extensively studied in survival analysis.
A comprehensive overview of classical approaches can be found in \cite{sun2006} and, more recently, in \cite{bogaerts2017} and \cite{gomez2026}.
These include Turnbull’s nonparametric maximum likelihood estimator (NPMLE) for the survival function, which is a generalization of the Kaplan–Meier estimator for right-censored data, along with the accelerated failure time (AFT) model, and various adaptations of the Cox proportional hazards model for interval-censored outcomes. 
Research on interval-censored data remains highly active, see \cite{sun2022} for a compilation of recent innovations.

While methodology for regression models with interval-censored covariates is scarce, covariates subject to interval censoring are not uncommon in biomedical studies.
For instance, when an exposure onset is observable only between assessments and exposure time is a covariate of interest. In the ACTG359 study, concerns were raised about how the waiting time from indinavir failure to starting a new antiretroviral treatment could impact the prognosis of the second treatment. In this case, the interval-censored waiting time acts as an exposure for the second treatment.
\cite{gomez2003} and \cite{sizelove2024} propose two different semiparametric approaches to estimate a linear model $Y = \alpha+\bm\beta'\bm X + \gamma Z + \epsilon$, $\epsilon\sim N(0,\sigma^2)$, when the exposure time $Z$ is interval-censored.
\cite{gomez2003} assume $Z$ has a discrete distribution with finite support and hence the likelihood function is simplified to depend on a finite set of parameters. In contrast, \cite{sizelove2024} propose to estimate the hazard function of the time-to-event covariate using latent Poisson random variables. The approach proposed by \cite{gomez2003} has led to several methodological contributions \citep{topp2004,langohr2014,gomez2022} and has subsequently been applied in various practical settings \citep{langohr2004,marhuenda2022,morrison2021}.
It is also worth noting that regression models with right-censored covariates have been receiving increasing attention; see \cite{lotspeich2023} for a review of existing methods.

Beyond time-to-event data, interval censoring also arises in analytical chemistry when an analyte of interest is quantified using chemical methods that involve measurement constraints.
As a second motivating example, we present the study of \cite{marhuenda2022}, which investigated the role of circulating carotenoids in cardiovascular health.
Caro\-tenoids are a group of plant-produced phytochemicals (e.g., $\alpha$-carotene, lycopene, astaxanthin), abundant in fruits and vegetables, with antioxidant activity in humans.
In this study, the quantity of each carotenoid, say $Q_c$, was measured using liquid chromatography coupled with UV–VIS detection and mass spectrometry, an analytical technique that entails a limit of detection (${\rm LoD}_c$) and a limit of quantification (${\rm LoQ}_c$) specific for each analyte \citep{Hrvolova2016}. As a result, observations of $Q_c$ are exact only if $Q_c \geq {\rm LoQ}_c$, and otherwise lie in either $[0,{\rm LoD}_c)$ or $[{\rm LoD}_c,{\rm LoQ}_c)$ for some individuals. 
If we denote by $\mathcal{C}$ the group of carotenoids, the quantity of circulating carotenoids is $Z = \sum_{c\in\mathcal{C}} Q_c$. 
Generalized linear models (GLMs) accommodate the distinct response types of diabetes, blood pressure, and other cardiovascular outcomes, with circulating carotenoids $Z$ as a covariate.
An observation of $Z$ results interval-censored if any of the participant's $Q_c$ is below ${\rm LoD}_c$ or ${\rm LoQ}_c$. For example, if $Q_{a_i}\in[0,{\rm LoD}_a)$ for the $a$th carotenoid of the $i$th participant, then $Z_i$ is interval-censored in $\big[ \sum_{c\neq a} Q_{c_i} + 0,\; \sum_{c\neq a} Q_{c_i} + {\rm LoD}_a \big)$.
Substituting such observations by an agreed-upon value has been reported to produce biased results \citep{helsel2006}. The analysis in \cite{marhuenda2022} treats circulating carotenoids as a discrete covariate and applies the approach of \cite{gomez2003} with the GLM extension in \cite{gomez2022}.

In this article, we propose a likelihood-based approach to handle mixed interval-censored covariates {\textemdash that is, observations can be either exact or interval-censored\textemdash} in generalized linear models, and that is not restricted to discrete support.  
Specifically, we develop an augmented version of Turnbull’s NPMLE for the censored covariate overriding the need to assume a parametric form, and thereby preventing bias in the GLM regression parameters that could arise from misspecifying the covariate distribution.
We denote the new estimator as GELc, in reference to Gómez-Espinal-Lagakos \citep{gomez2003}, since our approach generalizes their algorithm to accommodate continuous censored variables. 
We prove that the GELc estimator is consistent and asymptotically normal under some regularity conditions, with the covariance matrix estimated from the observed Fisher information.
Simulation studies are conducted to evaluate finite-sample performance.

The article is structured as follows. Section \ref{s:model} introduces the notation and assumptions underlying GELc. In Section \ref{s:GELc}, the estimator is presented along with pseudocode for the estimation algorithm. Section \ref{s:augmented} describes the augmented version of Turnbull's estimator. Section \ref{s:asym} discusses the asymptotic properties of the GELc estimator, while the simulation results are reported in Section \ref{s:simstud}. In Section \ref{s:appl}, two real applications are exposed in detail. The paper concludes with a discussion in Section \ref{s:discus} of the estimator’s limitations and directions for future research.

% ---------------------------------------------------------------------------- Notation
\section{Notation and model specification}
\label{s:model}

Let $Y$ denote an outcome variable, which is assumed fully observed and may be either discrete or continuous. The fully observed covariates are represented by the $p$-dimensional vector $\bm X$, while $Z$ is reserved for the interval-censored covariate. Consider a generalized linear model (GLM) that relates the mean response $\mu = E[Y|\bm{X},Z]$ to the linear predictor $\eta = \alpha+\bm{\beta'X}+\gamma Z$ via the link function $g$:
\begin{equation} \label{eq:glm}
    E[Y|\bm{X},Z] = g^{-1}(\alpha+\bm{\beta'X}+\gamma Z).
\end{equation}
In a GLM, the distribution of $Y|\bm{X},Z$ is assumed to belong to the exponential dispersion family, with location parameter $\mu=g^{-1}(\eta)$ and dispersion parameter $\phi$ \citep{dunn2018}. The general form of the density is
\[ f_{Y|\bm{X},Z}(y|\bm x,z;\bm\theta) = h(y;\phi) \exp\Big[\big\{y\cdot\psi(\mu) - a\big(\psi(\mu)\big) \big\}/\phi \Big] \]
where $\bm\theta=(\alpha,\bm\beta',\gamma,\phi)$ is the parameter vector and $h$, $a$, and $\psi$ are known functions specific to the selected distribution. Unlike the homoscedasticity assumption in linear regression, in GLMs we assume ${\rm Var}(Y | \bm X, Z) = \phi V(\mu)$, where $V(\mu)$ is a distribution-specific function, commonly called the variance function, that depends on the mean $\mu = E[Y|\bm{X},Z]$.

Regarding the interval-censored covariate $Z$, we denote by $\lfloor Z_L, Z_R \rfloor$ the random censoring interval, where $Z_L$ and $Z_R$ represent the left and right observed variables. This notation includes open $(z_l,z_r)$, closed $[z_l,z_r]$, and semi-open $[z_l,z_r)$ observed intervals. We work within the framework of partly or mixed interval-censored data, where the censoring process that generates $\lfloor Z_L, Z_R \rfloor$ admits exact observations as $[z,z]$. This setting generalizes case II interval censoring, in which the censoring process is defined by inspection times \citep{gomez2026}.
Let $\Omega$ denote the support of $Z$, and $W(z) = Pr(Z \le z)$ the marginal distribution function of $Z$ for $z\in\Omega$.

\bigskip
In what follows, the assumptions underlying the GELc estimator are stated along with a brief discussion of their interpretation and practical implications.

\begin{assumption} \label{A1}
    The random censoring interval $\lfloor Z_L, Z_R \rfloor$ and $Y$ are conditionally independent given $Z$, that is $Y \perp\!\!\!\perp \lfloor Z_L, Z_R \rfloor \;|\; Z$. 
\end{assumption}
That is to say, the observed interval carries no additional information about $Y$ beyond what is already contained in $Z$. This is necessary to omit the dependence of the model~\eqref{eq:glm} on $Z_L$ and $Z_R$.
In the study of carotenoids, this is trivially satisfied because the outcomes (e.g., diabetes, blood pressure) are unrelated to the analytical technique used to measure the circulating carotenoids.
When dealing with periodic visits, however, one must confirm that the visit schedule is not influenced by the response; if patients with an unfavorable prognosis for $Y$ tend to be monitored more frequently for $Z$, the assumption would not hold.

\begin{assumption} \label{A2}
    The censoring process that generates $\lfloor Z_L, Z_R \rfloor$ is non-informative on the specific value of $Z$ within $\lfloor Z_L, Z_R \rfloor$.
\end{assumption}
In other words, the observed interval informs that $Z$ lies within it, but nothing else. In the study of carotenoids, the censoring process is triggered by the limits of detection and quantification, which are fixed thresholds with no additional information. In contrast, in the ACTG359 study, this condition would be violated, for example, if after an observation close to the 500 copies/mL threshold the next visit were scheduled earlier than originally planned. See \cite{oller2004} for a thorough explanation of the noninformative condition.

\begin{assumption} \label{A3}
    The random variable $Z$ is independent of the covariates $\bm X$ included in the model.
\end{assumption}
This condition would generally be violated if $\bm X$ includes potential confounders of the association between $Z$ and $Y$ \citep{hernan2020}. We impose it here as a simplifying assumption to facilitate the theoretical derivations, and describe in the discussion a possible extension to relax the assumption.

\bigskip
The full likelihood function given the observed data $\{ y_i,\bm x_i,z_{l_i},z_{r_i} \}_{i=1}^n$ is
\begin{equation*}
    L_{\rm full}\big(\bm\theta \big| \{ y_i,\bm x_i,z_{l_i},z_{r_i} \}_{i=1}^n \big)
    = 
    \prod_{i=1}^n
    Pr\big(Y\in dy_i, Z_L\in dz_{l_i}, Z_R\in dz_{r_i}, 
    \bm X\in d\bm{x_i}, Z_i\in\lfloor z_{l_i},z_{r_i} \rfloor\big),
\end{equation*}
which is the probability that $Y,Z_L,Z_R$ and $\bm X$ take values in infinitesimal neighbourhoods around the observations $(y_i, z_{l_i}, z_{r_i}, \bm{x}_i)$, and that the unobserved $Z_i$ lies within $\lfloor z_{l_i},z_{r_i} \rfloor$.
Under Assumptions~\ref{A1} to~\ref{A3}, we show in~\appref{AppA1} that $L_{\rm full}$ is proportional to
\begin{equation} \label{eq:lik}
    L\big(\bm\theta, W \big| \{ y_i,\bm x_i,z_{l_i},z_{r_i} \}_{i=1}^n \big) = 
    \prod_{i=1}^n \int_{z_{l_i}}^{z_{r_i}} 
    f_{Y|\bm{X},Z}(y_i|\bm x_i,z;\bm\theta) \,
    dW(z).
\end{equation}
We aim to jointly estimate $\bm\theta = (\alpha,\bm\beta',\gamma,\phi)'$ and the nuisance distribution function $W$ by maximum likelihood, treating $W$ as an infinite-dimensional parameter.
% 'infinite-dimensional parameter' és un concepte que existeix i té propietats asimptòtiques pròpies. Es fa servir pe amb spline-based estimators

\section{The GELc estimator}
\label{s:GELc}

Define $\mathcal{L}=\{z_{l_i}\}_{i=1}^n$ and $\mathcal{R}=\{z_{r_i}\}_{i=1}^n$ as the ordered sets of unique left and right endpoints, respectively.
These points partition the support $\Omega$ of $Z$ into disjoint intervals $\{I_j\}_{j=1}^m$ such that:
\begin{itemize}
    \item[(i)]
    The intervals do not overlap and together cover the entire support, that is $\Omega = \bigcup_{j=1}^{m} I_j$;
    \item[(ii)] 
    Every observed interval $\lfloor z_{l_i},z_{r_i} \rfloor$ can be expressed as a union of elements of $\{I_j\}_{j=1}^m$; that is, for all $i=1,\dots,n$,
    % \begin{equation} \label{eq:partitiondef}
    %   \lfloor z_{l_i},z_{r_i} \rfloor = \bigcup_{j:\, I_j\cap \lfloor z_{l_i},z_{r_i} \rfloor\neq\emptyset} I_j.  
    % \end{equation}
    \begin{equation} \label{eq:partitiondef}
      \lfloor z_{l_i},z_{r_i} \rfloor = \bigcup_{j=1}^m \kappa^i_j \; I_j, \quad \text{where } \kappa^i_j = \mathbb{1}\{I_j\subseteq \lfloor z_{l_i},z_{r_i} \rfloor\} .  
    \end{equation}
\end{itemize}
For example, for semi-open observed intervals $\{[z_{l_i},z_{r_i})\}_{i=1}^n$, the ordered combined set $\mathcal{L}\cup\mathcal{R} = \{z_{(1)}<z_{(2)}<\dots<z_{(m)}\}$ induces intervals of the form $I_j=[z_{(j)},z_{(j+1)})$ for $j=1,\dots,m \leq 2n$.
We refer to $\{I_j\}_{j=1}^m$ as the augmented Turnbull intervals, whose number $m = m(n)$ grows with sample size.

% The augmented Turnbull intervals align the support $\Omega$ of the unknown distribution $W$ with the structure of the observed data, enabling likelihood-based estimation. 

Define $w_j = Pr(Z\in I_j)$ and consider $\bm w = (w_1,\dots,w_m)'$ the vector of unknown probabilities. Let $C_{ij}(\bm\theta)$ denote the probability density of $Y = y_i$ given the observed data averaged over $Z\in I_j$, that is
\begin{equation*}
    C_{ij}(\bm\theta) = \Big( \int_{I_j} f_{Y|\bm{X},Z}(y_i|\bm x_i,u;\bm\theta)\;du \Big)/|I_j|,
\end{equation*}
where $|I_j|$ is the length of the $j$th augmented Turnbull interval. Given a value of $\bm\theta$, the vector $\bm w$ is estimated by solving
\begin{equation} \label{eq:GELc:wj}
    \hat w_j = \frac{1}{n} \sum_{i=1}^n \kappa^i_j \, 
    \frac{ 
    \hat w_j \cdot C_{ij}(\bm\theta)
    }{
    \sum_{k=1}^m \kappa^i_k \hat w_k \cdot C_{ik}(\bm\theta)
    }, \quad j=1,\dots,m
\end{equation}
 where $\kappa^i_j$ is defined in \eqref{eq:partitiondef}. %$\bm\theta$-dependent self-consistent equations
The idea is similar to that of the empirical distribution function:
$w_j = Pr(Z\in I_j)$ is estimated as the proportion of observations in $I_j$. Specifically, an exact (non-censored) observation $z_i$ contributes only to the unique $w_j$ such that $z_i \in I_j$, which is $I_j=[z_i,z_i]$. In contrast, an observation $\lfloor z_{l_i},z_{r_i} \rfloor$ spreads its contribution over all $w_j$ for which $Z_i \in I_j$ is possible, namely those $I_j$ such that $I_j\subseteq \lfloor z_{l_i},z_{r_i} \rfloor$. The information that $y_i$ and $\bm x_i$ provide about $Z_i$ is contained in $C_{ij}(\bm\theta)$ and aids the allocation of mass across $\bm w$.

\begin{assumption} \label{A4}
    The nonparametric estimator $\widehat W_n$ of the distribution function for $Z$ is assumed to be 
    \textit{piecewise-uniform},
    meaning that it increases uniformly within each augmented Turnbull interval with constant slope $\hat w_j/|I_j|$, where $|I_j|$ denotes the length of the interval.
\end{assumption}

Under~\autoref{A4} and given $\bm\theta$, $\widehat W_n$ is fully characterized by the vector $\bm w$:
\begin{equation*}
    \widehat{W}_n(z;\bm\theta) = \sum_{k=1}^{j-1} \hat w_k + (z-a_j)\frac{\hat w_j}{|I_j|}\quad \text{for } z \in I_j=\lfloor a_j, b_j\rfloor,
\end{equation*}
and the loglikelihood from \eqref{eq:lik} can be expressed as
\begin{equation*} 
    \ell\big(\bm\theta, \bm{w} \big| \{ y_i,\bm x_i,z_{l_i},z_{r_i} \}_{i=1}^n \big) = 
    \sum_{i=1}^n \log \bigg( \sum_{j=1}^{m(n)} \kappa^i_j\, w_j\, C_{ij}(\bm\theta) \bigg).
\end{equation*}
The estimating equations for the GELc estimator $\bm{\hat\theta}_n$ are of the form
\begin{equation} \label{eq:score}
    \frac{1}{n}\sum_{i=1}^n S\big(\bm\theta, W; y_i,\bm x_i,z_{l_i},z_{r_i}\big) = 0,
\end{equation}
where
\begin{equation*}
    S\big(\bm\theta, W; y,\bm x,z_{l},z_{r}\big) =\frac{
    \int_{z_{l}}^{z_{r}}
        \mathcal{S}(\bm\theta; y,\bm x,z)\,
        f_{Y|\bm{X},Z}(y|\bm x,z;\bm\theta) \, dW(z)
    }{
    \int_{z_{l}}^{z_{r}} 
    f_{Y|\bm{X},Z}(y|\bm x,u;\bm\theta) \, dW(u)
    }
\end{equation*}
and $\mathcal{S}(\bm\theta; y,\bm x,z) = \frac{d}{d\bm\theta}\log f_{Y|\bm{X},Z}(y|\bm x,z;\bm\theta)$ (see~\appref{AppA1}).
The estimation algorithm alternates between (i) estimating $W$ via the self-consistent equations \eqref{eq:GELc:wj} given the current ${\bm{\hat\theta}_n}$, and (ii) updating ${\bm{\hat\theta}_n}$ by solving the estimating equations in \eqref{eq:score} with the newly updated $\widehat W_n$; see Algorithm~\ref{alg}. GELc has been implemented in R and is available at \href{https://github.com/atoloba/ICenCov}{\texttt{github.com/atoloba/ICenCov}}; see~\appref{AppA2} for details.

\begin{algorithm}[!ht]
\caption{GELc estimation algorithm}  \label{alg}
\begin{algorithmic}
\State \textbf{Input:} Initial estimates $\bm\theta_{\rm old}$, $\bm w_{\rm old}$ and tolerances $\epsilon_p > 0$, $\epsilon_\ell > 0$
\State \textbf{Output:} Final estimates $\bm\theta_{\rm max}$, $\bm w$, $\ell$, and $\widehat{\mathrm{Var}}(\bm\theta_{\rm max})$
\Repeat
    \Repeat \Comment{$\bm w$-estimation stage}
        \For{$j = 1,\dots,m$}
            \State $w_j \gets \displaystyle
            \frac{1}{n}\sum_{i=1}^n \kappa^i_j\,
                \frac{w_j\, C_{ij}(\bm\theta_{\rm old}) }
                { \sum_{k=1}^m \kappa^i_k\, w_k\, C_{ik}(\bm\theta_{\rm old}) }$
        \EndFor
        \State $\delta_w \gets \|\bm w - \bm w_{\rm old}\|_2/\|\bm w_{\rm old}\|_2$
    \Until{$\delta_w < \epsilon_p$}

    \State $\bm\theta_{\rm max} \gets \displaystyle\arg\max_{\bm\theta\in\Theta}\,\ell(\bm\theta, \bm w)$   \Comment{$\bm\theta$-estimation stage}
    \State $\ell \gets \ell(\bm\theta_{\rm max}, \bm w)$
    \State $\delta_\theta \gets \|\bm\theta_{\rm max} - \bm\theta_{\rm old}\|_2/\|\bm\theta_{\rm old}\|_2$
    \State $\delta_\ell \gets \|\ell - \ell_{\rm old}\|_2/\|\ell_{\rm old}\|_2$
    \State \textbf{update} $\bm w_{\rm old}$, $\bm\theta_{\rm old}$, and $\ell_{\rm old}$
    
\Until{$\delta_\theta+\delta_w < \epsilon_p$ or $\delta_\ell < \epsilon_\ell$}
\end{algorithmic}
\end{algorithm}

\section{\texorpdfstring{Augmented Turnbull estimator for the distribution of $Z$}{Augmented Turnbull estimator for the distribution of \textit{Z}}}
\label{s:augmented}

Turnbull's estimator, denoted by $\widehat W^{{\rm Tb}}$, is a step function that increases only on a set of disjoint intervals $\{[p_j,q_j]\}_{j=1}^M$. Turnbull intervals are constructed so that each left endpoint $p_j$ belongs to the set of observed left endpoints $\mathcal{L}=\{z_{l_i}\}_{i=1}^n$ and is immediately followed by a right endpoint $q_j\in\mathcal{R}=\{z_{r_i}\}_{i=1}^n$ \citep{turnbull1976}. The probabilities $w^{{\rm Tb}}_j = Pr(Z\in [q_j,p_j])$ are determined by the self-consistent equations
\begin{align} \label{eq:TBeq}
    &\hat w^{{\rm Tb}}_j = \frac{1}{n} \sum_{i=1}^n \mathbb{E}_{\widehat W^{{\rm Tb}}}\big[\mathbb{1}\{Z\in[q_j,p_j]\} \,|\, \lfloor z_{l_i},z_{r_i} \rfloor \big] \\
    &= \frac{1}{n} \sum_{i=1}^n \alpha^i_j \frac{\hat w^{{\rm Tb}}_j}{\sum_{k=1}^M \alpha^i_k \hat w^{{\rm Tb}}_k} \quad \text{for } j=1,\dots,M,  \nonumber
\end{align}
where $\alpha^i_j = \mathbb{1}\{[q_j,p_j]\subseteq\lfloor z_{l_i},z_{r_i} \rfloor\}$ \citep{groeneboom1992}. \cite{turnbull1976} proved that any distribution function $H$ that satisfies $H(p_j)- H(q_j)=\hat w^{\rm Tb}_j$ for all $j = 1, \dots, M$, and assigns no mass outside the Turnbull’s intervals, is a nonparametric maximum likelihood estimator (NPMLE) of $W$. Turnbull’s $\widehat W^{{\rm Tb}}$, defined over $\{[q_j,p_j]\}_{j=1}^M$, maximizes the likelihood of the observed intervals
\[
    L^{{\rm Tb}} \big(W | \{\lfloor z_{l_i},z_{r_i} \rfloor\}_{i=1}^n \big) = \prod_{i=1}^n \big\{ W(z_{r_i})-W(z_{l_i}) \big\},
\]
but not necessarily \eqref{eq:lik}.
The augmented Turnbull estimator $\widehat W_n$ is defined along the intervals $\{I_j\}_{j=1}^{m(n)}$ as constructed in Section \ref{s:GELc} and, analogously to \eqref{eq:TBeq},
\begin{equation*}
    \hat w_j = \frac{1}{n} \sum_{i=1}^n \mathbb{E}_{\widehat W_n}\big[\mathbb{1}\{Z\in I_j\} \,|\, \lfloor z_{l_i},z_{r_i} \rfloor,y_i,\bm x_i \big].
\end{equation*}
Under Assumptions~\ref{A1} to~\ref{A3}, this yields the $\bm\theta$-dependent self-consistent equations
\begin{equation} \label{eq:mle}
    \hat w_j = \frac{1}{n} \sum_{i=1}^n \kappa^i_j \, 
    \frac{\int_{I_j} f_{Y|\bm{X},Z}(y_i |\bm x_i, z;\bm\theta) d\widehat W_n(z;\bm\theta) }{\sum_{k=1}^m \kappa^i_k \int_{I_k} f_{Y|\bm{X},Z}(y_i |\bm x_i, u;\bm\theta) d\widehat W_n(u;\bm\theta)},
\end{equation}
where we are omitting for clarity the dependence of $\hat w_j$ to $\bm\theta$.

\begin{theorem} \label{thm:mle}
    For fixed $\bm\theta$, any distribution function $H$ satisfying
    \begin{equation*}
        H(b_j)-H(a_j) = \hat w_j, \quad \text{for all } j=1,\dots,m
    \end{equation*}
    where $I_j=\lfloor a_j,b_j \rfloor$ are the augmented Turnbull intervals and $\hat w_j$ solves~\eqref{eq:mle} is a nonparametric maximum likelihood estimator (NPMLE) of $W$ for the likelihood in \eqref{eq:lik}. That is, the NPMLE is completely characterized by the $\bm\theta$-dependent self-consistent equations~\eqref{eq:mle}.
\end{theorem}
In other words, the likelihood in \eqref{eq:lik} does not determine the shape of $\widehat W_n$ within the augmented Turnbull intervals. \autoref{A4}, which imposes a piecewise-uniform form for $\widehat{W}_n$, provides a convenient choice, although other piecewise specifications consistent with \eqref{eq:mle} could equally be adopted.
\autoref{A4} leads to the equations proposed in~\eqref{eq:GELc:wj}.
The proof of~\autoref{thm:mle} appears in~\appref{AppA3}.

\begin{theorem} \label{thm:TBconsistent}
Assume that the random censoring interval $\lfloor Z_L,Z_R \rfloor$ can take only finitely many observed intervals. If for all $z \in\Omega$ the probability of observing an exact value $z$ converges to one as $n\to\infty$, then the augmented Turnbull estimator $\widehat W_n$ converges uniformly to the true distribution $W_0$, that is,
\[
\sup_{\bm\theta\in\Theta} \sup_{z\in\Omega} |\widehat W_n(z;\bm\theta)-W_0(z)| \xrightarrow[n\to\infty]{p}0
\]
for $\Theta$ compact parameter space.
\end{theorem}
The proof is provided in~\appref{AppA4}. 
The assumption that $\lfloor Z_L,Z_R \rfloor$ ranges over a finite set of intervals is reasonable: in follow-up studies, visit times are recorded on a discrete scale (e.g., days or months) over a finite study period; and measurements are often obtained with finite decimal precision. Hence, even if the true $Z$ varies continuously, the observable endpoints $Z_L$ and $Z_R$ lie on a finite set.
By contrast, the condition on the support $\Omega$ fails if there exist a region that has positive probability for $Z$ but is unreachable for $Z_L$ and $Z_R$. This is the case, for example, for single carotenoids $Q_c$ under the limit of detection: the true $Q_c$ can lie in $[0,{\rm LoD}_c)$, but it cannot occur that $0<Z_L \leq Z_R < {\rm LoD}_c$.

% ---------------------------------------------------------------------------- Asymptotic properties
\section{Asymptotic properties of the GELc estimator}
\label{s:asym}

Let $\bm\theta_0$ denote the true parameter value. The estimating equations \eqref{eq:score} are solved iteratively, with $W$ replaced by the augmented Turnbull NPMLE $\widehat W_n(z;\bm\theta)$ at each update of $\bm\theta$.
Let $\Theta$ and $\mathcal{W}$ denote the parameter and function spaces of $\bm\theta$ and $W$, respectively. 

\begin{theorem} \label{thm:consistent}
    Under conditions \ref{thm:consis:iid}--\ref{thm:consis:w}, the GELc estimator is consistent, in the sense that $\| \bm{\hat\theta}_n - \bm\theta_0 \| \xrightarrow[]{p} 0$ as $n\to\infty$.
    \begin{enumerate}[label=(C\arabic*), leftmargin=*]
    
    \item \label{thm:consis:iid}
    The observations $\{(y_i,z_{l_i},z_{r_i},\bm x_i)\}_{i=1}^n$ are independent and identically distributed (i.i.d.) with true distribution satisfying model~\eqref{eq:glm} evaluated at $\bm\theta_0\in\Theta$; and $Z_1,\dots, Z_n$ are i.i.d. with true distribution $W_0\in\mathcal{W}$.

    \item \label{thm:consis:collinearity}
    No perfect multicollinearity among $(\bm X, Z)$ is assumed.
    
    \item \label{thm:consis:w}
    $\widehat W_n$ converges uniformly to $W_0$, that is,
    \[
    \sup_{\bm\theta\in\Theta} \sup_{z\in\Omega} |\widehat W_n(z;\bm\theta)-W_0(z)| \xrightarrow[n\to\infty]{p} 0.
    \]
    \end{enumerate}
    Additional standard conditions concerning the compactness of $\Theta$, entropy of $\mathcal W$, and the continuity of $S(\bm\theta,W;y,z_{l},z_{r},\bm x)$ are stated in~\appref{AppA5}.
\end{theorem}
The proof, provided in~\appref{AppA5}, builds on the arguments of \cite{chen2003}. Consistency follows from the identifiability of $\bm\theta_0$, and the uniform convergence of the estimating function $n^{-1}\sum_{i=1}^n S\big(\bm\theta, \widehat W_n; y_i,\bm x_i,z_{l_i},z_{r_i}\big)$ to $\mathbb{E}\big[S\big(\bm\theta, W_0; y,\bm x,z_{l},z_{r}\big)\big]$. Conditions \ref{thm:consis:iid} and \ref{thm:consis:collinearity} are standard for identifiability. Uniform convergence of $\widehat W_n$ is essential for the consistency of $\bm{\hat\theta}_n$, and it must hold in the entire support of $Z$. We have provided sufficient conditions in~\autoref{thm:TBconsistent}.

\begin{theorem} \label{thm:asymnorm}
    Suppose that the conditions of~\autoref{thm:consistent} are satisfied. 
    If $\frac{d}{d \bm\theta} \mathcal{S}(\bm\theta; y,\bm x,z)$ has full column rank and $\widehat W_n$ converges to $W_0$ at a rate faster than $n^{-1/4}$, that is,
    \[
    \sup_{\bm\theta\in\Theta} \sup_{z\in\Omega} |\widehat W_n(z;\bm\theta)-W_0(z)| =o_p(n^{-1/4}),
    \]
    then, under additional smoothness and differentiability conditions stated in~\appref{AppA5}, $\sqrt{n}(\bm{\hat\theta}_n-\bm\theta_0)$ is asymptotically normal with with covariance matrix $\mathcal I^{-1}$, that is,    
    \[ \sqrt{n}(\bm{\hat\theta}_n-\bm\theta_0) \xrightarrow[n\to\infty]{d} N(0,\mathcal I^{-1}), \]
    where $\mathcal I = \mathbb{E}\big[ S(\bm\theta_0,W_0;\bm o)S(\bm\theta_0,W_0;\bm o)' \big]$ is the Fisher information matrix.
\end{theorem}
The proof is derived from the first-order Taylor expansion of the estimating equations \eqref{eq:score} around $(\bm\theta_0,W_0)$ evaluated at $(\bm{\hat\theta}_n,\widehat W_n)$, which yields the classical asymptotic linearization
\begin{equation} \label{eq:linearization}
    \sqrt{n}(\bm{\hat\theta}_n-\bm\theta_0) \approx
    -J_n^{-1} \sqrt{n}\, A_n 
\end{equation}
where $J_n = n^{-1} \sum_{i=1}^n \frac{d}{d \bm\theta} S(\bm\theta, W_0; y_i,\bm x_i,z_{l_i},z_{r_i})\rvert_{\bm\theta_0}$. Unlike in the standard linearization, the leading term $A_n$ includes an additional differential component due to the estimation of $W$, which converges to zero faster than $n^{-1/2}$ if $\widehat W_n$ converges uniformly to $W_0$ at a rate faster than $n^{-1/4}$. 
The convergence rate of $\widehat W_n$ heavily depends on the censoring process and, by analogy with Turnbull’s estimator, is expected to be faster than $n^{-1/4}$ \citep[][Section 3.6.2]{sun2006}. 
Asymptotic normality then follows from Slutsky’s theorem. Finally, the full column rank condition on $\frac{d}{d \bm\theta} \mathcal{S}(\bm\theta; y,\bm x,z)$ guarantees that $\mathcal I$ is invertible. The complete proof is provided in~\appref{AppA5}. Fisher information matrix is consistently estimated as $\widehat{\mathcal I} = - \frac{d^2}{d\bm\theta^2} \ell(\bm\theta,W)$ evaluated at $(\bm{\hat\theta}_n, \widehat W_n)$, the MLE.

\section{Simulation Study}
\label{s:simstud}

We conducted a simulation study to evaluate how well the asymptotic properties hold for moderate sample sizes. Performance of $\bm{\hat\theta}_n$ was assessed in terms of bias and empirical standard error, and coverage probability of 95\% confidence intervals was also evaluated.

\subsection{Design of scenarios}

% scenarios
We considered sample sizes $n = 100$, $300$, and $500$.
The true $Z_i$ values were simulated from an exponential distribution with mean $12$ and then censored using the process described below.
Since the performance may also be influenced by the width of the observed intervals, the censoring process was replicated with three increasing mean widths: $3$, $6$, and $9$.
To assess sensitivity to the outcome distribution, we considered two options from the GLM family: a Gamma model with $\log E[Y|Z] = \alpha+\gamma Z$; and a Logistic regression model for a binary outcome, ${\rm logit} E[Y|Z] = \alpha+\gamma Z$.
For the same purpose, we set the covariate effect $\gamma$ and the dispersion parameter $\phi$ to several values. These choices were inspired by the two illustrations presented in Section~\ref{s:appl}. The $72$ scenarios are summarized in Table~\ref{t:scenarios}.
\clearpage

\begin{table}[!ht]
\caption{Summary of the simulation scenarios with $n=100$, $300$, and $500$.}
\label{t:scenarios}
\begin{center}
\begin{tabular}{lll}
\toprule
\multicolumn{3}{l}{\textbf{Censoring process}} \\[1ex]
 & Fully observed data: $[z_i, z_i]$ for $i=1,\dots,n$ &  \\ 
 \noalign{\vskip 4pt}\cline{2-3}\noalign{\vskip 4pt}
\multirow{3}{*}{} & \multirow{3}{9cm}{Intervals $[z_{l_i}, z_{r_i})$ generated with mean $\mu$ and standard deviation $\sigma = 0.75\mu$} & 
$\mu = 9,\ \sigma = 6.75$ \\
&&$\mu = 6,\ \sigma = 4.50$ \\
&&$\mu = 3,\ \sigma = 2.25$ \\
\hline
\multicolumn{3}{l}{\textbf{Outcome $Y_i$}} \\
\multirow{4}{*}{} & \multirow{4}{9cm}{Gamma with mean $\mu_i = \exp(10 + \gamma z_i)$ and variance $\phi\mu_i^2$} & 
$\gamma = 0.02,\ \phi = 0.02$ \\
&&$\gamma = 0.02,\ \phi = 1$ \\
&&$\gamma = -0.05,\ \phi = 0.02$ \\
&&$\gamma = -0.05,\ \phi = 1$ \\
\noalign{\vskip 4pt}\cline{2-3}\noalign{\vskip 4pt}
\multirow{2}{*}{} & \multirow{2}{9cm}{Bernoulli with mean $p_i = \mathrm{logit}^{-1}(\gamma z_i)$ and variance $p_i(1 - p_i)$} & 
$\gamma = 0.1$ \\
&&$\gamma = -0.1$ \\
\bottomrule
\end{tabular}
\end{center}
\end{table}

% censoring process, noninformative
The censoring process is as follows. For each simulated $z_i$, an initial gap $\xi_{i1}\sim U(0,\mu)$ and a sequence of independent and identically distributed $N(\mu,\sigma)$ gap values $\{\xi_{ir}\}_{r=2}^{R_i}$ are generated. These gaps are summed accumulatively, starting with $\tau_{0} = 0$ and setting $\tau_{r} = \sum_{s=1}^r \xi_{is}$ for $r=1,2,\dots$, until the condition $z_i< \tau_{R_i}$ is met. The observed interval is then defined as $[z_{l_i}, z_{r_i})$ with $z_{l_i}=\tau_{R_i-1}$ and $z_{r_i}=\tau_{R_i}$. With this generating mechanism the noninformative condition is guaranteed \citep{gomez2009}. The mean width $\mu$ is set to three increasing values, and the standard deviation is chosen so that the coefficient of variation is fixed at $0.75$.
In Figure~\ref{fig:simstudy:data}, the simulated observed intervals are displayed.

\begin{figure}
\begin{center}
\centerline{\includegraphics[width=\linewidth]{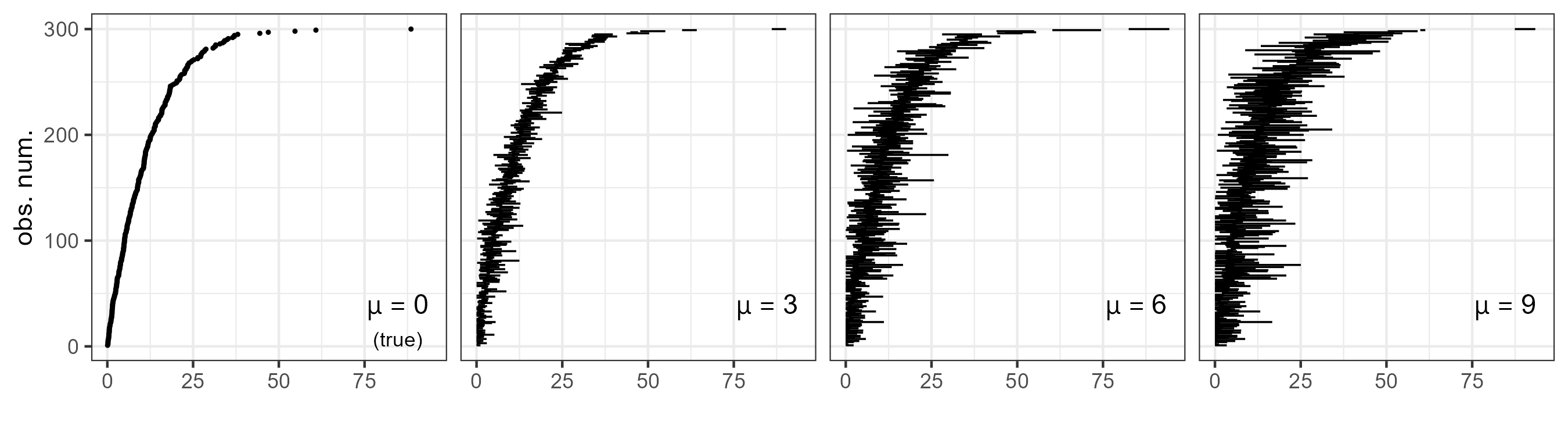}}
\end{center}
\caption{Simulated observed intervals for $n=300$ under the four censoring scenarios: the true $Z_i$ values ($\mu=0$) and the three scenarios of censoring process with increasing mean widths. Observations are ordered by the true $Z_i$ values to facilitate comparison of how censoring interval widths vary across scenarios.
\label{fig:simstudy:data}}
\end{figure}

% data generating mechanism, repetitions, and reproducibility
A total of $500$ datasets were generated per scenario.
To reduce the MCSE, data were generated according to the following procedure:
\begin{enumerate}
    \item Simulate the values $\{z_i\}_{i=1}^n$ for $n=500$ (which in practice are not observed);
    \item Simulate the outcomes $\{y_i\}_{i=1}^n$ for each distribution scenario;
    \item Simulate the observed intervals $\{z_{l_i},z_{r_i}\}_{i=1}^n$ for each censoring scenario;
    \item Construct full datasets by combining the variables simulated at steps 2 and 3;
    \item Select the first $n = 100$, $300$, or $500$ observations, according to the scenario.
\end{enumerate}
In short, the same underlying $\{z_i\}_{i=1}^n$ is used across all scenarios, and the simulated variables $\{y_i, z_{l_i},z_{r_i}\}_{i=1}^n$ are reused where appropriate. This guarantees that differences in results across scenarios are attributable solely to the factors of interest (outcome distribution, censoring process, and sample size), rather than to variability in the data generation \citep{morris2019}. The procedure is repeated $500$ times, and the state of the random-number generator at the start of each repetition is stored to ensure reproducibility.

\subsection{Results}

As shown in~\appref{AppB2}, relative bias is negligible for the regression parameters $(\alpha,\gamma)$ at sample sizes $n=300$ and $500$, and it becomes noticeable at $n=100$ for some scenarios. Regarding the width of the observed intervals, relative bias increases slightly with wider intervals, though the trend is modest and does not raise major concerns about estimator robustness.
In contrast, the dispersion parameter $\phi$ exhibits systematic downward bias across all scenarios. This aligns with known behavior of the MLE for the dispersion parameter in GLMs \citep{dunn2018}.
Empirical standard errors decrease as the sample size increases, revealing that the precision of the estimator improves with sample size. The root mean squared error is nearly identical to the empirical standard error, indicating that the mean squared error is influenced primarily by estimator variability rather than bias. Figure~\ref{fig:simstudy:result} presents the relative bias and empirical standard error for varying sample sizes and mean widths of the observed intervals in the gamma regression scenario with parameters $\alpha_0 = 10$, $\gamma_0 = -0.05$, and $\phi_0 = 1$. Monte Carlo standard errors (MCSEs), computed following \cite{morris2019}, were small and are reported in~\appref{AppB23}. The formulas are given in~\appref{AppB1}.

\begin{figure}[!ht]
\begin{center}
\centerline{\includegraphics[width=\linewidth]{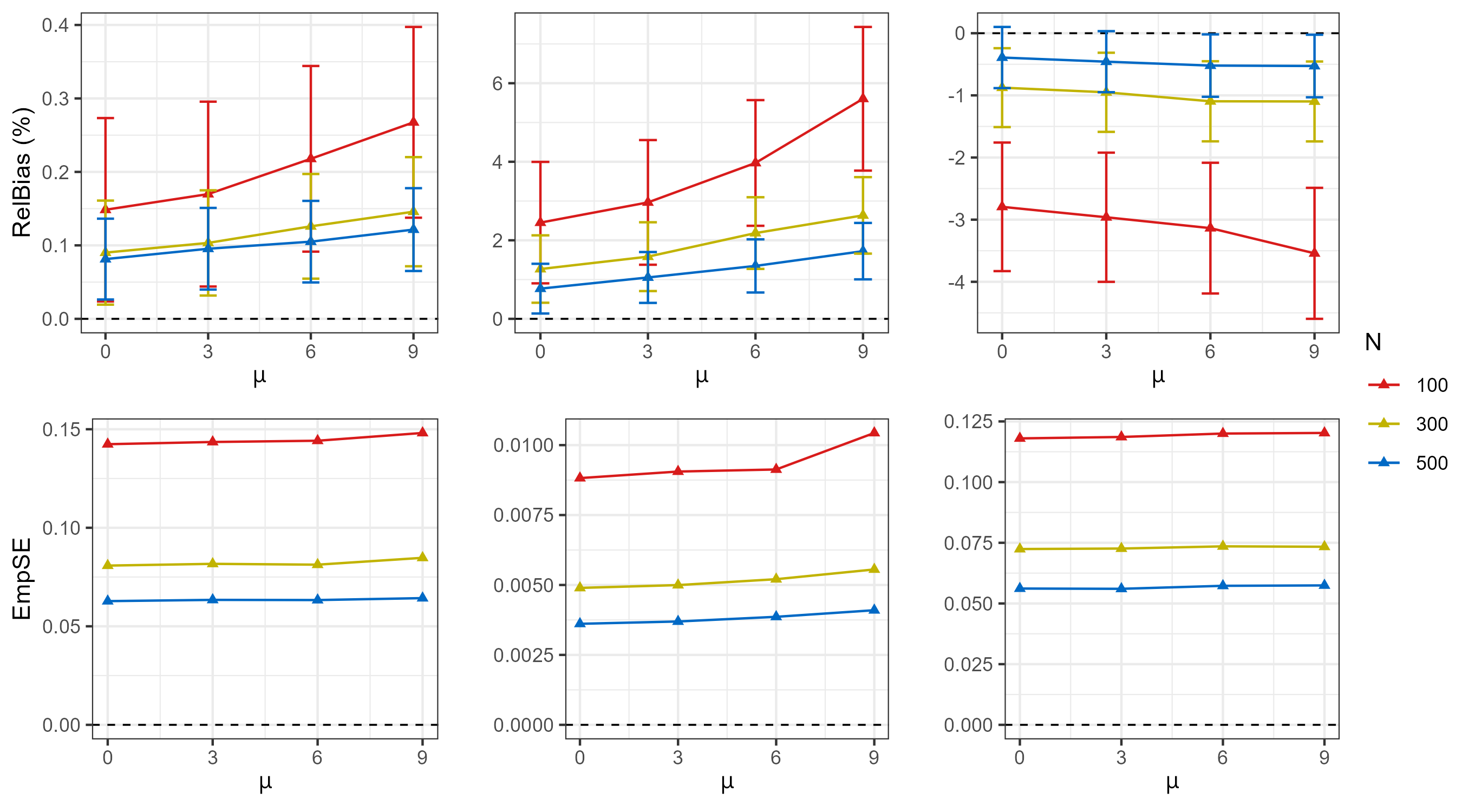}}
\end{center}
\caption{Results for the gamma regression with true parameters $\alpha=10$, $\gamma = -0.05$ and $\phi = 1$, showing relative bias (with Monte Carlo error bars) and empirical standard error. From left to right: $\alpha$, $\gamma$, and $\phi$ parameters.
\label{fig:simstudy:result}}
\end{figure}

\medskip

Coverage probability of the 95\% confidence intervals is reported in~\appref{AppB3}. For the regression parameters $(\alpha, \gamma)$, coverage is generally close to the target level of $0.95$. Slight undercoverage is observed in some scenarios with $n=100$ and wider observed intervals, but this deviation diminishes as sample size increases.
A likely explanation is the extra variability introduced by estimating the covariate distribution $W$, as exposed in \eqref{eq:linearization}. The additional term in the asymptotic law of the estimator is of order $n^{-1/2}$, so it affects finite-sample behavior but shrinks with sample size.
 % The improvement in coverage with increasing $n$ is consistent with this interpretation.

\appref{AppB4} reports the mean and standard deviation of computation time, grouped by sample size and mean width of the observed intervals. Instances of non-convergence were rare and negligible. As expected, the mean computation time increases with both the width of the observed intervals and the sample size. For example, for the mean width of 9, the mean computation time per model is 8 minutes for $n=100$, 1 hour for $n=300$, and 3 hours for $n=500$. This increase is attributable to the greater computational burden associated with a higher number of augmented Turnbull intervals $m$, which affects both the self-consistent equations and the subsequent likelihood maximization.

% ---------------------------------------------------------------------------- Applications
\section{Illustrative examples}
\label{s:appl}

\subsection{Study of carotenoids}
Carotenoids are a family of dietary metabolites found in fruits and vegetables, with well-documented antioxidant properties.
As explained in the introduction, the study of \cite{marhuenda2022} aimed to investigate the association between circulating plasma carotenoids and a range of anthropometric, clinical, and biochemical parameters related to cardiovascular risk .
The data corresponds to the baseline visit of a subgroup ($n=230$) of participants in the PREDIMED-Plus study~\citep{predimedplus}, a Spanish primary prevention trial of cardiovascular disease in older adults (aged 55 to 75) at high cardiovascular risk.

As illustration, we consider diabetes ($Y$) as the outcome of interest, where $Y$ is a binary variable equal to $1$ for diabetic individuals. We explore its association with total carotenoid concentration ($Z$) using a logistic regression model adjusted for sex ($X_1$) and Body Mass Index (BMI; $X_2$).
\begin{equation*}
    \log\Big( \frac{E[Y]}{1-E[Y]} \Big) = \alpha+ \beta_1 X_1 + \beta_2 X_2 +\gamma Z
\end{equation*}

Of the 230 participants, 63 (27\%) were diabetic, 124 (54\%) were men, and the mean BMI is 33 kg/m$^2$ (SD: 3.5). 
Total carotenoid concentration is interval-censored, with observed left endpoints ($z_{l_i}$) starting from $0$ and right endpoints ($z_{r_i}$) reaching up to 41.4 $\mu$mol/L. The observed intervals have a mean width of 0.44 (SD: 0.41), and eight observations (3\%) are uncensored.
Parameter estimates obtained using the GELc estimator, along with estimated standard errors, are presented in Table~\ref{t:table1}. Parameter interpretation follows as in standard logistic regression: adjusting for sex and BMI, each one-unit increase in carotenoid concentration is associated with an odds ratio of exp(–0.08) = 0.92 (95\%-CI: 0.86 to 0.98) for having diabetes.

Simulation results allow us to evaluate the impact of bias on the odds ratio. In the scenarios for logistic regression with $n=100$ and interval width of 9, relative bias for $\gamma$ reached up to 14\%. For the estimation obtained this means that $(-0.08-\gamma_0)/\gamma_0 = 0.14$, which gives $\gamma_0=-0.07$. The corresponding bias-adjusted odds ratio is then 0.93, which is very close to the estimated one. On the other hand, the confidence intervals demonstrate good coverage. We can therefore conclude that higher carotenoid levels are associated with lower odds of diabetes.

\begin{table}[!ht]
\caption{Estimates obtained by the GELc estimation approach.}
\centering
\label{t:table1}
\begin{tabular}{lccccc}
  \toprule
  $\hat\alpha$ {\small (se)} & $\hat\beta_1$ {\small (se)} & $\hat\beta_2$ {\small (se)}& $\hat\gamma$ {\small (se)}  \\
    \hline
    $-1.24$ {\small ($1.45$)}        
        & $-0.53$ {\small ($0.32$)} 
        & $0.03$ {\small ($0.04$)}
        & $-0.08$ {\small ($0.03$)}  \\
    \bottomrule
\end{tabular}
\end{table}

\subsection{AIDS Clinical Trials Group Study 359}

The data ($n=81$) is from a randomized clinical trial designed to compare six different antiretroviral treatment regimens in HIV-infected individuals who had previously failed the indinavir therapy, an inhibitor of the HIV \citep{gulick2000}. 
One research question of interest was whether delays in initiating the new treatment could influence viral load levels at baseline, and so potentially affect treatment prognosis.
Indinavir failure was defined as a rebound in viral load ($>500$ copies/mL) after an initial period of virologic suppression, with viral load monitored at regular visits. Therefore, the time between indinavir failure and study start is interval-censored.
At study start, the median viral load was 12,000 copies/mL (IQR: 5,000-35,000). Participants were 40 years old in mean (SD: 9), ranging from 24 to 64.
Observed left endpoints for time between indinavir failure and study are as early as 1 week, and right endpoints as late as 90 weeks. The observed intervals have a mean width of 12 weeks (SD: 9).

This study served as illustration in \cite{gomez2003} and \cite{langohr2014}. The authors fitted a linear model between the $\log_{10}$-transformed viral load at study start ($Y$) and the time since indinavir failure ($Z$), adjusted by age ($X$). Also, \cite{morrison2021} modelled the distribution of HIV viral load as a function of time since HIV seroconversion.
Here, for illustration purposes, we consider a Gamma GLM with a log link function:
\[ \log (E[Y]) = \alpha+\beta X+\gamma Z, \]
where $Z$ is interval-censored and $\gamma$ is the parameter of interest. 
Equivalently, a gamma distribution is assumed for viral load ($Y$), that is, the conditional density of $Y$ is given by
\begin{equation*}
    f_{Y| X, Z}(y ; \mu_i, \phi) = \exp\bigg\{ \frac{y (-1/\mu_i) - \log\mu_i}{\phi} + 
  \log(y/\phi)/\phi - \log(y\Gamma(1/\phi)) \bigg\},
\end{equation*}
with covariate-depending mean $\mu_i = \exp(\alpha+\beta x_i+\gamma Z_i)$ and variance $\sigma_i^2 = \phi \mu_i^2$.

In Table~\ref{t:table2} we provide the parameter estimates obtained using the GELc estimator, along with the estimated standard errors. Parameter interpretation is the same as in a standard gamma regression model: adjusting for age, RNA copies/mL increase by exp(3$\times$0.02) - 1 = 5\% (95\%-CI: 1 to 9\%) with every 3-week increase in time between indinavir failure and study start. 
Comparing with the simulation scenario of Gamma regression with $\gamma=0.02$, $\phi=1$, $N=100$, and $\mu=9$, we deduce the relative bias for $\gamma$ is approximately $3\%$; that is $\gamma_0\approx 0.02/1.03 = 0.02$.

\begin{table}[!ht]
\caption{Estimates obtained by the GELc estimation approach.}
\centering
\label{t:table2}
\begin{tabular}{lccccc}
  \toprule
  % Model & 
  $\hat\alpha$ {\small (se)} & $\hat\beta$ {\small (se)}& $\hat\gamma$ {\small (se)}  & $\hat\phi$ \\
    \hline
  % Gamma GLM &
    $10.2$ {\small ($0.64$)}        
        & $-0.01$ {\small ($0.014$)} 
        & $0.02$ {\small ($0.007$)} 
       &  $1.165$ \\
    \bottomrule
\end{tabular}
\end{table}

Although reproducing previous analyses is not the focus here, it is interesting to compare our results with those in in \cite{gomez2003} and \cite{langohr2014}, where a linear model was used for $\log_{10}$-transformed viral load. To make our estimates comparable, we can apply the identity $\log(E[Y]) = \log_{10}(E[Y]) \cdot \log 10$, which is equivalent to dividing our regression coefficients by $\log 10$. After this adjustment, the resulting estimates are similar in magnitude. Regarding error variability, the linear model yielded $\widehat {\rm Var}(\log_{10}(Y))=0.2732$, whereas the Gamma regression model implies $\widehat {\rm Var}(Y) = \hat\phi \hat{\mu}^2$ with $\hat\phi=1.165$. Using the delta method to translate the Gamma result to the $\log_{10}$ scale,
\[
\widehat {\rm Var}_{\rm G}(\log_{10}(Y)) \approx \Big(\frac{1}{y \log 10}\Big)^2\, \widehat {\rm Var}(Y) \approx \hat\phi/\log(10)^2 = 0.2197.
\]
Since $0.2197<0.2732$, this suggests that a Gamma GLM for RNA may fit the ACTG359 data better than an LM for $\log_{10}$RNA. Nevertheless, any model preference and distributional assumption should be validated with appropriate diagnostic tools.

% ---------------------------------------------------------------------------- Discussion
\section{Discussion}
\label{s:discus}

In this paper, we proposed the GELc (Gomez-Espinal-Lagakos for continuous support) estimator to fit GLMs when a covariate is interval-censored. 
We extended the estimator introduced by \cite{gomez2003} to accommodate continuous covariates, and we showed that it is consistent and asymptotically normal under standard regularity conditions.
% We also investigated its finite-sample properties through simulation studies.
Two methodological considerations originally observed by \cite{gomez2003} are particularly relevant in the context of censored covariates and were retained in our extension. 
First, the estimation of the censored covariate’s distribution must account for its association with the response variable.
Ignoring this dependence would compromise the consistency 
of the regression parameter $\gamma$, which is the central goal of the analysis. %covariate’s effect
Second, assuming a parametric form for the covariate's distribution may result in model misspecification, which could introduce substantial bias in the regression parameters of the GLM.

The \textit{augmented Turnbull’s estimator} we propose is the nonparametric maximum likelihood estimator for the distribution of the interval-censored covariate. It is defined over a partition of the covariate’s support that generalizes Turnbull’s construction by accounting for all observed endpoints, ensuring that each subject’s interval contributes information precisely where it is observed. 
Right-censored observations produce an interval of the form $\lfloor k, \infty)$ in the constructed partition, over which the estimator is currently assumed to increase uniformly. This reflects a standard limitation of nonparametric estimators, including the Kaplan–Meier estimator for right-censored data, which provide no information about the distribution beyond the largest finite right endpoint.
Recent work by \cite{lee2025} and \cite{vazquez2024} have explored the estimation of regression parameters for right-censored covariates using inverse probability weighting methods, offering promising directions for future developments.

In addition, the current formulation assumes that the censored covariate $Z$ is independent of the additional covariates $\bm X$. This simplifying assumption was necessary to develop the GELc estimator and to establish its asymptotic properties. It could be relaxed, however, by specifying a semiparametric accelerated failure time model (AFTM) for $Z|\bm X$. One of the particularities of AFTMs is that the model expression is equivalent to $Z=e^{\bm{x'\vartheta}} E$, with $E=\exp(\epsilon)$ now independent of $\bm X$. With this reparametrization of the estimation problem, the augmented Turnbull estimator is used on the distribution of $E$ while $\vartheta$ is estimated jointly with the regression parameters.
Because $E=e^{-\bm{x'\vartheta}} Z$ is a one-to-one transformation, the model is identifiable and the asymptotic results established in this paper are applicable.

Besides GLMs, the GELc estimator can be extended to accommodate interval-censored covariates in other regression settings,
provided the conditional distribution of the response is fully specified. 
This includes not only parametric models with fully observed outcomes, but also the accelerated failure time model (AFTM) and other parametric survival models for right- and interval-censored responses. These extensions are particularly relevant in survival analysis, where time-to-event data may arise in the covariate as well as in the response.

Despite its flexibility, the proposed methodology has some practical limitations. 
The augmented Turnbull estimator depends on the regression parameters through the conditional density function. This increases the computational burden and slows down the estimation process. It would therefore be valuable to develop an EMICM-type algorithm, building on the approach of \cite{anderson2017b} but adapted to the augmented Turnbull estimator, as this could offer a more efficient alternative to solving the self-consistent equations directly.

Besides this, there are several interesting lines of work to develop. One is the extension of the GELc estimator to the Cox proportional hazards model and other semiparametric survival models. 
In contrast to parametric survival models, estimation is based on the partial likelihood and the response distribution is left unspecified.
Another relevant direction would be to enable the method to simultaneously handle multiple censored covariates.
Finally, because the censoring of covariates also affects the construction of residuals, it is essential to develop diagnostic tools to appropriately verify the model assumptions.

\section*{Supplementary material}

The GELc estimator is implemented in R and available in GitHub at 
\href{https://github.com/atoloba/ICenCov}{\texttt{github.com/atoloba/ICenCov}}; the repository also contains reproducible materials for the illustrations presented in Section~\ref{s:appl}. The R code for the simulation study can be found at \href{https://github.com/atoloba/Simulations_icglm_2025}{\texttt{github.com/atoloba/Simulations\_icglm\_2025}}.

\section*{Funding}

This work was supported by the Ministerio de Ciencia, Innovación y Universidades (Spain) under grant PID2023-148033OB-C21, and by the Departament d’Empresa i Coneixement de la Generalitat de Catalunya under grant 2021 SGR 01421 (GRBIO). A. Toloba was additionally supported by the Departament de Recerca i Universitats de la Generalitat de Catalunya through the 2022 FISDU 00277 fellowship.

\printbibliography[segment=0]
\clearpage

\appendix
\newrefsegment

\addappendix{Appendix A}

\subsection{Likelihood function and estimating equations}
\label{AppA1}

This section provides the derivation of the simplified likelihood expression given in Equation~\eqref{eq:lik}, as well as the corresponding estimating equations in~\eqref{eq:score}.

As in introduced in Section~\ref{s:model}, the full likelihood function is given by
\begin{equation*}
    L_{\rm full}(\bm\theta|\{ y_i,\bm x_i,z_{l_i},z_{r_i} \}_{i=1}^n) 
    = \prod_{i=1}^n Pr\big(Y\in dy_i, Z_L\in dz_{l_i}, Z_R\in dz_{r_i}, 
    \bm X\in d\bm{x_i}, Z_i\in\lfloor z_{l_i},z_{r_i} \rfloor\big)
\end{equation*}
which can be rewritten as
\begin{align*}
    &L_{\rm full}(\bm\theta|\{ y_i,\bm x_i,z_{l_i},z_{r_i} \}_{i=1}^n) 
    =  \prod_{i=1}^n \int_{z_{l_i}}^{z_{r_i}} 
    Pr\big(Y\in dy_i, Z_L\in dz_{l_i}, Z_R\in dz_{r_i},
    \bm X\in d\bm{x_i}, Z_i\in dz\big) dz \\
    &= \prod_{i=1}^n \Big\{ \int_{z_{l_i}}^{z_{r_i}}
    Pr\big(Y\in dy_i | Z_L=z_{l_i}, Z_R=z_{r_i},\bm X = \bm{x_i}, Z_i=z\big) \cdot \\
    &\qquad Pr(Z_L\in dz_{l_i}, Z_R\in dz_{r_i} | Z_i=z, \bm X=\bm{x_i})\cdot
    Pr(Z_i\in dz | \bm X=\bm{x_i})\cdot
    Pr(\bm X\in d\bm{x_i})\, dz \Big\}.
\end{align*}
Under~\autoref{A1}, the first term of the integrand corresponds to the conditional density $f_{Y|\bm{X},Z}(y_i|\bm x_i,z;\bm\theta)$.
By the non-informative censoring~\autoref{A2}, the second term is constant over the interval $\lfloor z_{l_i}, z_{r_i} \rfloor$. Consequently, $Pr(Z_L\in dz_{l_i}, Z_R\in dz_{r_i} | Z_i=z, \bm X=\bm{x_i})\equiv K(z_{l_i}, z_{r_i}, \bm x_i)$ doesn't depend on $z$ and can be factorized out of the integral.
That is,
\begin{align*}
    L_{\rm full}(\bm\theta|\{ y_i,\bm x_i,z_{l_i},z_{r_i} \}_{i=1}^n) 
    =& \Big( \prod_{i=1}^n \int_{z_{l_i}}^{z_{r_i}} f_{Y|\bm{X},Z}(y_i|\bm x_i,z;\bm\theta)\, Pr(Z\in dz \mid \bm X=\bm{x}_i) \Big)  \\
    &\cdot\Big( \prod_{i=1}^n K(z_{l_i}, z_{r_i}, \bm x_i)\, Pr(\bm X\in d\bm{x}_i) \Big).
\end{align*}
Assuming the parameter spaces of $(Y,Z)|\bm X$ and $(\bm X, Z_L, Z_R)$ are disjoint, the second term can be treated as a constant with respect to $\bm\theta$ and thus omitted from the likelihood for estimation purposes. Under~\autoref{A3} and since $Z_i \sim_{\text{iid}} Z$, it holds that $Pr(Z \in dz | \bm X = \bm{x}_i) = Pr(Z \in dz) = dW(z)$. This yields the simplified likelihood function:
\[
L\big(\bm\theta, W |\{ y_i,\bm x_i,z_{l_i},z_{r_i} \}_{i=1}^n \big) = \prod_{i=1}^n \int_{z_{l_i}}^{z_{r_i}} f_{Y|\bm{X},Z}(y_i|\bm x_i,z;\bm\theta)\, dW(z).
\]

\bigskip
Consider the loglikelihood function conditional on $W$, that is,
\[
\ell\big(\bm\theta; W \mid \{ y_i,\bm x_i,z_{l_i},z_{r_i} \}_{i=1}^n\big) = \sum_{i=1}^n \log\left( \int_{z_{l_i}}^{z_{r_i}} f_{Y|\bm{X},Z}(y_i|\bm x_i,z;\bm\theta)\, dW(z) \right).
\]
The score vector is given by
\[
\frac{d}{d\bm\theta} \ell(\bm\theta; W) = 
\sum_{i=1}^n \left( \int_{z_{l_i}}^{z_{r_i}} f_{Y|\bm{X},Z}(y_i|\bm x_i,z;\bm\theta)\, dW(z) \right)^{-1} 
\int_{z_{l_i}}^{z_{r_i}} \frac{d}{d\bm\theta} f_{Y|\bm{X},Z}(y_i|\bm x_i,z;\bm\theta)\, dW(z),
\]
where
\[
\frac{d}{d\bm\theta} f_{Y|\bm{X},Z}(y|\bm x,z;\bm\theta) = 
\underbrace{\frac{d}{d\bm\theta}\log f_{Y|\bm{X},Z}(y|\bm x,z;\bm\theta)}_{\displaystyle\mathcal{S}(\bm\theta; y,\bm x,z)}
\cdot f_{Y|\bm{X},Z}(y|\bm x,z;\bm\theta),
\]
and $\mathcal{S}(\bm\theta; y,\bm x,z)$ is the score vector for given $z$. Therefore, the score equations for $\bm\theta$ are 
\begin{equation} \label{eq:sc}
    \sum_{i=1}^n S\big(\bm\theta, W; y_i,\bm x_i,z_{l_i},z_{r_i}\big) = 0
\end{equation}
with 
\begin{align*}
    &S\big(\bm\theta, W; y_i,\bm x_i,z_{l_i},z_{r_i}\big) = 
    \frac{d}{d\bm\theta} \ell(\bm\theta; W)= \\[1ex]
    &\int_{z_{l_i}}^{z_{r_i}}
        \mathcal{S}(\bm\theta; y_i,\bm x_i,z)
        \frac{
        f_{Y|\bm{X},Z}(y_i|\bm x_i,z;\bm\theta) 
        }{
        \int_{z_{l_i}}^{z_{r_i}} 
        f_{Y|\bm{X},Z}(y_i|\bm x_i,u;\bm\theta) \, dW(u)
        }\, dW(z).
\end{align*}

\medskip

The estimating equations~\eqref{eq:score} result from multiplying \eqref{eq:sc} by $1/n$.

\subsection{Implementation of GELc}
\label{AppA2}

% Choose initial estimates $\bm\theta^{(0)}$ and $\bm w^{(0)}$. For example, place the masses $\bm w^{(0)}$ uniformly, that is $w_j^{(0)}=|I_j|/|\Omega|$, and use as $\bm\theta^{(0)}$ the estimated value obtained from fitting the model (\ref{eq:glm}) with simple imputation. 

The GELc estimation algorithm alternates between updating the distribution weights $\bm w$ and the parameter vector $\bm\theta$ until joint convergence is achieved.
We use a combined convergence criterion. The
condition for parameter convergence $\delta_\theta+\delta_w < \epsilon_p$ indicates that the relative changes in both parameter vectors are sufficiently small. In parallel, the loglikelihood stabilization condition $\delta_\ell < \epsilon_\ell$ monitors the change in the loglikelihood value, and captures situations in which the loglikelihood reaches a flat region even if the first criterion has not yet been met.

The $\bm w$-estimation stage involves element-wise matrix operations, with a per-iteration computational cost of $O(nm)$, where $n$ is the sample size and $m$ the number of augmented Turnbull intervals. 
For the $\bm{\theta}$-estimation stage, we maximize the loglikelihood using the Limited-memory Broyden–Fletcher–Goldfarb–Shanno algorithm (L-BFGS-B) implemented in the \texttt{optim} function, which is that of \cite{byrd1995}. This is accessed via the \texttt{mle2} function from the \texttt{bbmle} package \citep{bbmle}, which provides a flexible interface for likelihood-based inference in R. The gradient of the loglikelihood function $\ell$ is provided analytically for the regression parameters $\alpha$, $\bm{\beta}$, and $\gamma$, which accelerates convergence. Evaluation of $\ell$ requires approximating at most $n \times m$ integrals, which is done using the \texttt{integrate} function; if numerical instability occurs, we use \texttt{cubintegrate} from the \texttt{cubature} package instead \citep{cubature}. The algorithm is implemented in R and available at \href{https://github.com/atoloba/ICenCov}{\texttt{github.com/atoloba/ICenCov}}.
% Alternating between the two estimation stages increases computational complexity.

% wider observed intervals tend to require more iterations to reach convergence.

\subsubsection*{Analytical gradient of $\ell(\bm\theta|\bm w)$}

Denote by $\bm\beta=(\beta_0,\beta_1,\dots,\beta_{p+1})'$ the regression parameters corresponding to intercept, the censored-covariate's effect, and the coefficients for the $p$-dimensional fully observed covariate vector $\bm X$, respectively. 
Because the conditional density function $f_{Y|\bm X,Z}$ belongs to the exponential dispersion family, the loglikelihood given $W$ can be written as
\begin{align*}
    \ell(\bm\beta; \phi | W) &= \sum_{i=1}^n \log \int_{z_{l_i}}^{z_{r_i}} f_{Y|\bm X,Z}(y_i |\bm x_i, s;\bm\theta)\, dW(s) \\
    &= \sum_{i=1}^n \log \int_{z_{l_i}}^{z_{r_i}} \exp\{ \underbrace{\{y_i\psi_i(s) - a(\psi_i(s))\}/\phi + c(y_i,\phi)}_{\displaystyle \tilde l_i(\bm\beta, s;\phi)} \}  \, dW(s),
\end{align*}
where $\psi_i(s) = \psi(g^{-1}(\beta_0+\beta_1 s + \beta_2 x_{i1}+\dots+ \beta_{p+1} x_{ip}))$.
Since $\frac{d}{d\bm\beta}\tilde l_i(\bm\beta, z;\phi)$ is the score function for given $z$,
\begin{align*}
    \frac{\partial}{\partial\beta_j} l(\bm\beta; \phi | W) &= \sum_{i=1}^n \frac{1}{\int_{z_{l_i}}^{z_{r_i}} f(y_i |\bm x_i,s;\bm\theta)\, dW(s)} \int_{z_{l_i}}^{z_{r_i}}  \frac{\partial}{\partial\beta_j} f(y_i |\bm x_i,s;\bm\theta) \, dW(s) \\
    &= \sum_{i=1}^n \int_{z_{l_i}}^{z_{r_i}} \frac{\partial}{\partial\beta_j} \tilde l_i(\bm\beta, s;\phi)\,  \frac{f(y_i |\bm x_i,s;\bm\theta) \, dW(s)}{\int_{z_{l_i}}^{z_{r_i}} f(y_i |\bm x_i,s;\bm\theta)\, dW(s)} ds \\
    &= \sum_{i=1}^n \frac{\int_{z_{l_i}}^{z_{r_i}} \frac{y_i-\mu_i}{\phi V(\mu_i)}\frac{d\mu}{d\eta} \tilde x_{ij} f(y_i |\bm x_i,s;\bm\theta) \, dW(s)}{\int_{z_{l_i}}^{z_{r_i}} f(y_i |\bm x_i,s;\bm\theta)\, dW(s)},
\end{align*}
where $\bm{\tilde x}_{i}=(1,s,\bm x_{i})$. In terms of $\bm w$:
\begin{equation*}
    \frac{\partial}{\partial\beta_j} l(\bm\beta; \phi | \bm w) =\sum_{i=1}^n \frac{\displaystyle    
    \sum_{J=1}^m \kappa^i_J \frac{w_J}{|I_J|} \int_{I_J} \frac{y_i-\mu_i}{\phi V(\mu_i)}\frac{d\mu}{d\eta} \tilde x_{ij} f(y_i |\bm{x}_i, s;\bm\theta) ds
    }{\displaystyle
    \sum_{k=1}^m \kappa^i_k \frac{w_k}{|I_k|} \int_{I_k} f(y_i |\bm{x}_i, s;\bm\theta)ds}.
\end{equation*}

The gradient function is provided to \texttt{optim} via the \texttt{mle2} function.

\subsection{Proof of Theorem~\ref{thm:mle}}
\label{AppA3}

Recall the equations in \eqref{eq:mle}:
\begin{equation*}
    \hat w_j = \frac{1}{n} \sum_{i=1}^n \kappa^i_j \, 
    \frac{\int_{I_j} f_{Y|\bm{X},Z}(y_i |\bm x_i, z;\bm\theta) d\widehat W_n(z;\bm\theta) }{\sum_{k=1}^m \kappa^i_k \int_{I_k} f_{Y|\bm{X},Z}(y_i |\bm x_i, u;\bm\theta) d\widehat W_n(u;\bm\theta)}, \qquad j=1,\dots,m
\end{equation*}
where the dependence of $\hat w_j$ to $\bm\theta$ is omitted for clarity.

\begin{theorem1}
    For fixed $\bm\theta$, any distribution function $H$ satisfying
    \begin{equation*}
        H(b_j)-H(a_j) = \hat w_j, \quad \text{for all } j=1,\dots,m
    \end{equation*}
    where $I_j=\lfloor a_j,b_j \rfloor$ are the augmented Turnbull intervals and $\hat w_j$ solves~\eqref{eq:mle} is a nonparametric maximum likelihood estimator (NPMLE) of $W$ for the likelihood in \eqref{eq:lik}. That is, the NPMLE is completely characterized by the $\bm\theta$-dependent self-consistent equations~\eqref{eq:mle}.
\end{theorem1}

\begin{proof}
    
Assume $\bm\theta$ known and fixed, consider the loglikelihood function of $W$
\begin{equation*}
    \ell\big(W(\cdot); \bm\theta\big| \{ y_i,\bm x_i,z_{l_i},z_{r_i} \}_{i=1}^n\big) = \sum_{i=1}^n \log\Big(\int_{z_{l_i}}^{z_{r_i}} f_{Y|\bm{X},Z}(y_i|\bm x_i,z;\bm\theta)\,dW(z) \Big),
\end{equation*}
and denote the ith likelihood contribution as:
\begin{equation*}
    L_i\big(W(\cdot); \bm\theta\big) = \int_{z_{l_i}}^{z_{r_i}} f_{Y|\bm{X},Z}(y_i|\bm x_i,z;\bm\theta)\,dW(z).
\end{equation*}
For a fixed $z_0\in\Omega$, define $A_{z_0}(s):=Pr(Z\leq s | Z\leq z_0) - Pr(Z\leq s)$. It follows that 
\[
A_{z_0}(s) = \left\{\begin{array}{ll}
    0 & {\rm if}\, W(z_0)=0 \\
    \frac{W(s \wedge z_0)}{W(z_0)} - W(s) & {\rm if}\, W(z_0)>0
\end{array}\right.
\]
where $a\wedge b = \min\{a,b\}$, because
\[ A_{z_0}(s) = \frac{Pr(Z\leq s, Z\leq z_0)}{Pr(Z\leq z_0)} - Pr(Z\leq s) = 
 \frac{Pr(Z\leq \min\{s,z_0\})}{Pr(Z\leq z_0)} - Pr(Z\leq s).\]

\medskip

We aim to differentiate the functional $\ell\big(W(\cdot); \bm\theta\big) = \sum_{i=1}^n \log L_i\big(W(\cdot); \bm\theta\big)$ at $W(s)=Pr(Z\leq s)$ in the direction $A_{z_0}(s)=Pr(Z\leq s | Z\leq z_0) - Pr(Z\leq s)$.  
The MLE will be that function $\widehat W_n$ that makes the differential equal to zero for all $z_0\in\Omega$, provided that the second differential is negative. The differential of $\ell\big(W(\cdot); \bm\theta\big)$ in the direction $A_{z_0}(s)$ is
\begin{equation}\label{eq:lim1}
    \lim_{h\to 0} \frac{1}{h} \big( \ell(W+hA_{z_0};\bm\theta) - \ell(W;\bm\theta\textbf{})\big).
\end{equation}

The following lemma and corollary are stated beforehand, as the development of \eqref{eq:lim1} relies on them.
For notational convenience, we write $L_i(W)$ instead of $L_i\big(W(\cdot); \bm\theta\big)$.

\begin{lemmaA} \label{lA1}
    Consider $W_1, W_2$ two possible distribution functions for $Z$, and let $a,b$ be two arbitrary constants. Then, $L_i(aW_1+bW_2) = a L_i(W_1) + b L_i(W_2)$.
\end{lemmaA}
\begin{proof} 
\[
    \begin{array}{rl}
        L_i(aW_1+bW_2) =&  \int_{z_{l_i}}^{z_{r_i}} f_{Y|\bm X, Z}(y_i |\bm x_i, s;\bm\theta) \, d(aW_1+bW_2)(s)\\
        =&\int_{z_{l_i}}^{z_{r_i}} f_{Y|\bm X, Z}(y_i |\bm x_i, s;\bm\theta) \, \big( a\cdot dW_1(s) + b\cdot dW_2(s) \big) 
        = a L_i(W_1) + b L_i(W_2)
    \end{array}
\]
\end{proof}
\begin{corollary} \label{cA1}
    Define $W_{z_0}(s):=W(s \wedge z_0)$, a function of $s$ that is constant for $s>z_0$. It follows from the lemma that $L_i(A_{z_0}) = \frac{1}{W(z_0)} L_i(W_{z_0}) - L_i(W)$ for $W(z_0)>0$.
\end{corollary}

\bigskip

Therefore, \eqref{eq:lim1} develops as
\begin{align}
    &\lim_{h\to 0} \frac{1}{h} \big( \ell(W+hA_{z_0};\bm\theta) - \ell(W;\bm\theta\textbf{})\big) = \sum_{i=1}^n \lim_{h\to 0} \frac{1}{h} \log \frac{L_i(W+hA_{z_0})}{L_i(W)} \nonumber\\
    & \underset{\it lemma\,\ref{lA1}}{=}  \sum_{i=1}^n \lim_{h\to 0} \frac{1}{h} \log\frac{L_i(W) + hL_i(A_{z_0})}{L_i(W)}\nonumber\\
    & \underset{*}{=} \sum_{i=1}^n \log \lim_{1/h\to \infty}  \bigg( 1+\frac{1}{1/h}\frac{L_i(A_{z_0})}{L_i(W)} \bigg)^{1/h} = \sum_{i=1}^n \frac{L_i(A_{z_0})}{L_i(W)} \nonumber\\
    &\underset{\it coro.\,\ref{cA1}}{=} \sum_{i=1}^n \bigg( \frac{1}{W(z_0)} \frac{L_i(W_{z_0})}{L_i(W)} - 1 \bigg) = \bigg( \frac{1}{W(z_0)} \sum_{i=1}^n\frac{L_i(W_{z_0})}{L_i(W)} \bigg) - n, \label{eq:funcdif}
\end{align}
where in * we have used the fact that the logarithm can be interchanged with the limit due to its continuity and that $e^k = \lim_{n\to\infty} (1+\frac{k}{n})^{n}$ by definition.

\medskip

The MLE is the function $\widehat W_n$ that makes \eqref{eq:funcdif} equal to zero for all $z_0\in\Omega$. That is, it is the solution of
\begin{equation} \label{eq:WexprL}
     W(z) = \frac{1}{n} \sum_{i=1}^n \frac{L_i(W_{z})}{L_i(W)} \quad\forall z\in\Omega.
\end{equation}

This last part of the proof is to obtain the equations~\eqref{eq:mle} from \eqref{eq:WexprL}. The function $W_z(s)$ develops into
\[
W_{z}(s) = W(s \wedge z) = \mathbb{1}\{s< z\} W(s) + \mathbb{1}\{s\geq z\} W(z),
\]
\[
\frac{d}{ds}W_{z}(s) = \mathbb{1}\{s< z\} dW(s) + \mathbb{1}\{s\geq z\} \cdot 0,
\]
which means that the contributions
\[
\frac{L_i(W_{z})}{L_i(W)} = \frac{\int_{z_{l_i}}^{z_{r_i}} f_{Y|\bm X, Z}(y_i |\bm x_i, s;\bm\theta)\, dW_{z}(s)}{\int_{z_{l_i}}^{z_{r_i}} f_{Y|\bm X, Z}(y_i |\bm x_i, u;\bm\theta)\, dW(u)}
\]
depend on the position of $z$ with respect to the observed interval $\lfloor z_{l_i},z_{r_i} \rfloor$. That is:

\begin{itemize}
    \item If $z< z_{l_i}$, then $dW_{z}(s) = 0$ for all $s\in\lfloor z_{l_i},z_{r_i} \rfloor$. Therefore $L_i(W_z)/L_i(W)=0$.
    \item If $z\geq z_{r_i}$, then $dW_{z}(s) = dW(s)$ for all $s\in\lfloor z_{l_i},z_{r_i} \rfloor$ and $L_i(W_z)/L_i(W)=1$.
    \item Otherwise, $z\in \lfloor z_{l_i},z_{r_i} \rfloor$ and
    \begin{equation*}
        \frac{L_i(W_{z})}{L_i(W)} = \frac{\int_{z_{l_i}}^{z} f_{Y|\bm X, Z}(y_i |\bm x_i, s;\bm\theta)\, dW(s)}{\int_{z_{l_i}}^{z_{r_i}} f_{Y|\bm X, Z}(y_i |\bm x_i, u;\bm\theta)\, dW(u)}.
    \end{equation*}
\end{itemize}

\textit{Note: If the left endpoint is open, ie $\lfloor z_{l_i},z_{r_i} \rfloor=(z_{l_i},z_{r_i} \rfloor$, then the first case becomes $z\leq z_{l_i}$. It has been omitted for clarity.}

\textit{Note: It is irrelevant if the right endpoint is open or closed. If closed, $z=z_{r_i}$ in the third case gives $\frac{L_i(W_{z})}{L_i(W)} = 1$, the same result than in the second case.}

\medskip

Wrapping up, equation \eqref{eq:WexprL} is equal to 
\begin{align} \label{eq:WexprL2}
     W(z) = \frac{1}{n} \sum_{i=1}^n \bigg\{ 
     \mathbb{1}\{z\in (-\infty, z_{l_i}\rfloor \} \cdot 0 + 
     \mathbb{1}\{z\in \lfloor z_{r_i},\infty) \} \cdot 1  \nonumber \\
     + \mathbb{1}\{z\in \lfloor z_{l_i},z_{r_i} \rfloor \}
     \frac{\int_{z_{l_i}}^{z} f_{Y|\bm X, Z}(y_i |\bm x_i, s;\bm\theta)\, dW(s)}
     {\int_{z_{l_i}}^{z_{r_i}} f_{Y|\bm X, Z}(y_i |\bm x_i, u;\bm\theta)\, dW(u)}
     \bigg\}.
\end{align}
Moving $z$ along $\Omega$, any solution $\widehat W_n$ of \eqref{eq:WexprL2} can be expressed as
\begin{equation*}
    \widehat W_n(z) = \sum_{k=1}^{j-1} \hat w_k + \frac{1}{n} \sum_{i=1}^n\kappa^i_j
    \frac{\int_{a_j}^{z} f_{Y|\bm X, Z}(y_i |\bm x_i, s;\bm\theta)\, d\widehat W_n(s)}
     {\sum_{k=1}^m\kappa^i_k \int_{a_k}^{b_k} f_{Y|\bm X, Z}(y_i |\bm x_i, u;\bm\theta)\, d\widehat W_n(u)} , \quad \text{if } z\in I_j
\end{equation*}
where $w_j = Pr(Z\in I_j=\lfloor a_j,b_j \rfloor)$, $\kappa^i_j=\mathbb{1}\{I_j\subseteq \lfloor z_{l_i},z_{r_i} \rfloor\}$, and
\begin{equation} \label{eq:jfdoiljfsa}
    \hat w_j = \frac{1}{n} \sum_{i=1}^n\kappa^i_j
    \frac{\int_{a_j}^{b_j} f_{Y|\bm X, Z}(y_i |\bm x_i, s;\bm\theta)\, d\widehat W_n(s)}
     {\sum_{k=1}^m\kappa^i_k \int_{a_k}^{b_k} f_{Y|\bm X, Z}(y_i |\bm x_i, u;\bm\theta)\, d\widehat W_n(u)}.
\end{equation}
Observe that the specific shape of $d\widehat W_n(s)$ within each interval $I_j$ plays no role in the likelihood given the observed data and fixed $\bm\theta$. Indeed, different choices of shape will generally lead to different values of the masses $\hat w_1,\dots,\hat w_m$, but the loglikelihood value $\ell(\bm{\hat w};\bm\theta | \{ y_i,\bm x_i,z_{l_i},z_{r_i} \}_{i=1}^n)$ will remain the same. Consequently, the system of equations in \eqref{eq:jfdoiljfsa} define an entire class of functions that qualify as NPMLEs. In other words, the NPMLE is completely characterized by the $\bm\theta$-dependent self-consistent equations given by \eqref{eq:jfdoiljfsa}.

\end{proof}

\subsection{Proof of Theorem~\ref{thm:TBconsistent}}
\label{AppA4}

\begin{theorem2} 
Assume that the random censoring interval $\lfloor Z_L,Z_R \rfloor$ can take only finitely many observed intervals. If for all $z \in\Omega$ the probability of observing an exact value $z$ converges to one as $n\to\infty$, then the augmented Turnbull estimator $\widehat W_n$ converges uniformly to the true distribution $W_0$, that is,
\[
\sup_{\bm\theta\in\Theta} \sup_{z\in\Omega} |\widehat W_n(z;\bm\theta)-W_0(z)| \xrightarrow[n\to\infty]{p}0
\]
for $\Theta$ compact parameter space.
\end{theorem2}

\begin{proof}

We first prove the consistency of $\widehat W_n$ assuming that the true distribution has finite support, denoted by $W_0^*$, and later extend the results for arbitrary $W_0$. We use a rationale similar to \cite{yu2000}.

Let $\Omega^* = \{s_1,\dots, s_{m}\}$ denote the finite support of $W_0^*$. Since the probability of observing each value $s_k$ converges to one as the sample size increases, it follows from the strong law of large numbers that 
\[
Pr(Z_i\neq s_k, \forall i=1,\dots,n)\xrightarrow[n\to\infty]{}0.
\]
Therefore, in the limit, the augmented Turnbull intervals are precisely $I_j=[s_j,s_j]$, for $j=1,\dots, m$, covering the entire support $\Omega^*$.

From \eqref{eq:WexprL2}, it follows that $\widehat W_n$ satisfies the self-consistent equation
\begin{align} \label{eq:integraleq}
    H_n(z) = \int \frac{
\int_{z_{l}}^{z} f_{Y|\bm X, Z}(y |\bm x, s;\bm\theta)\, dH_n(s)
} {
\int_{z_{l}}^{z_{r}} f_{Y|\bm X, Z}(y |\bm x, u;\bm\theta)\, dH_n(u) 
} \mathbb{1}\{z\in\lfloor z_l,z_r \rfloor\} \, dP_n(y,\bm x, z_l,z_r)
\nonumber\\
+ \frac{1}{n}\sum_{i=1}^n  \mathbb{1}\{z\in\lfloor z_r,\infty)\},
\end{align}
where $P_n$ denotes the empirical distribution function of the observed random variable $\bm O = (Y,\bm X, Z_L,Z_R)$.

Since each $H_n$ is a distribution function, the sequence $\{H_n\}_{n\geq1}$ is bounded and monotone. Therefore, by the Helley-Bray selection theorem, there exist a subsequence $\{H_{n_k}\}_{k\geq1}$ and a function $H$ such that $H_{n_k}(z)\to H(z)$ pointwise for all $z\in\Omega$ \citep{rudin1976}. Moreover, since the empirical cdf $P_n$ converges to the true cdf $P$ with probability one, the limit of \eqref{eq:integraleq} as $k\to\infty$ is
\[
H(z) = \int \frac{
\int_{z_{l}}^{z} f_{Y|\bm X, Z}(y |\bm x, s;\bm\theta)\, dH(s)
} {
\int_{z_{l}}^{z_{r}} f_{Y|\bm X, Z}(y |\bm x, u;\bm\theta)\, dH(u) 
} \mathbb{1}\{z\in\lfloor z_l,z_r \rfloor\} \, dP(y,\bm x, z_l,z_r)
\,+ Pr\big(z\in\lfloor Z_R,\infty)\big).
\]
which simplifies to
\[
H(z) = \int Pr(Z\leq z|y,x,zl,zr)\; dP(y,x,zl,zr) = Pr(Z\leq z) = W_0^*(z).
\]
Given that the support of $W_0^*$ is finite, pointwise convergence implies uniform convergence, so
we conclude that $\widehat W_n\to W_0^*$ uniformly on $\Omega^*$ for fixed $\bm\theta$.
Finally, uniform convergence in $\bm\theta\in\Theta$ follows from the compactness of $\Theta$ and the fact that $\widehat W_n(z;\bm\theta)$ is uniformly continuous in $\bm\theta$ \citep{vaart1998}.

\bigskip

Consider $W_0$ arbitrary.
Since the support of $\lfloor Z_L,Z_R \rfloor$ is finite, the sequence of augmented Turnbull intervals $\big\{\{I_j\}_{j=1}^{m(n)}\big\}_{n\geq 1}$ stabilizes as $n\to\infty$. Consequently, for any $\epsilon>0$, we can construct a discrete distribution $W_0^*$ such that $\sup_{z\in\Omega^*} |W_0^*-W_0|<\epsilon$, and which also satisfies the self-consistent equation \eqref{eq:integraleq}. Since we have already established that $\sup_{\theta\in\Theta} \|\widehat W_n(\cdot;\bm\theta)-W_0^*(\cdot)\|\to 0$, the conclusion follows from the triangle inequality:
\[
\sup_{\theta\in\Theta} \|\widehat W(\cdot;\bm\theta)-W_0(\cdot)\|\to 0.
\]

\end{proof}

\subsection{Asymptotic properties of the GELc estimator}%
\label{AppA5}

Let $\bm O=(Y,Z_L,Z_R,\bm X)$ be the random variable for the observed data. Denote by $\bm o=(y, z_{l}, z_{r}, \bm{x})$ an arbitrary realization of $\bm O$ and by
$\{\bm{o}_i = (y_i, z_{l_i}, z_{r_i}, \bm{x}_i)\}_{i=1}^n$ the sample observed data. 

% Let $\bm\theta_0$ denote the true parameter for $f_{Y|\bm X, Z}(y |\bm x, z;\bm\theta)$
% and $W_0(\cdot)$ the true underlying distribution for $Z$. The pair $(\bm{\hat\theta}_n, \widehat W_n(\cdot;\bm{\hat\theta}_n))$ form the GELc estimator for $\bm\theta_0$ and $W_0$, 
% where we have added the subscript $n$ to clarify the asymptotic behavior as $n\to \infty$. Let $\Theta$ denote the parameter space for $\bm\theta$ and $\mathcal{W}$ the function space for $W$. Assume $\bm\theta_0\in\Theta$ and $W_0(\cdot)\in\mathcal{W}$. 

The estimating equations for $\bm{\hat\theta}_n$ in \eqref{eq:score} correspond to the empirical mean equation
\begin{equation*}
    g_n\big(\bm\theta,\widehat W_n(\cdot;\bm\theta)\big) = n^{-1}\sum_{i=1}^n S\big(\bm\theta, \widehat W_n(\cdot;\bm\theta); \bm o_i\big)=0.
\end{equation*}
Define the associated population expectation
\begin{equation*}
    G\big(\bm\theta,\widehat W_n(\cdot;\bm\theta)\big) = \mathbb{E}_{\bm O}\big[ S\big(\bm\theta, \widehat W_n(\cdot;\bm\theta); \bm o\big) \big],
\end{equation*}
where
\begin{equation*}
    S(\bm\theta, W;\bm o) =
    \int_{z_{l}}^{z_{r}}
        \mathcal{S}(\bm\theta; y,\bm x,z)
        \frac{
        f_{Y|\bm{X},Z}(y|\bm x,z;\bm\theta) 
        }{
        \int_{z_{l}}^{z_{r}} 
        f_{Y|\bm{X},Z}(y|\bm x,u;\bm\theta) \,
        dW(u)
        }\, dW(z)
\end{equation*}
with $\mathcal{S}(\bm\theta; y,\bm x,z) = \frac{d}{d\bm\theta}\log f_{Y|\bm{X},Z}(y|\bm x,z;\bm\theta)$ the score vector for given $z$.

Let $\Theta$ denote the parameter space of $\bm\theta$ and $\mathcal{W}$ the function space of $W$. Assume that $\bm\theta_0 \in \Theta$ and $W_0(\cdot) \in \mathcal{W}$, where the subscript $0$ indicates the true (but unknown) parameter and distribution of $Z$.

\subsubsection{\texorpdfstring{Consistency of $\hat\theta_n$}{Consistency of the GELc estimator}}

At each step of the iterative process to solve $g_n\big(\bm\theta,W(\cdot)\big)=0$, the nuisance distribution $W$ is replaced by the NPMLE $\widehat W_n(\cdot;\bm\theta)$. This plug-in step affects the asymptotic behavior of the empirical mean $g_n$ as $n\to\infty$. To facilitate comprehension, we first show the consistency of $\bm{\hat\theta}_n$ assuming $W_0$ is known, and then extend the proof for $W$ estimated by the augmented Turnbull estimator $\widehat W_n(\cdot;\bm\theta)$.

\subsubsubsection{\texorpdfstring{Consistency assuming $W_0$ is known}{Consistency assuming W0 is known}}

In addition to \ref{thm:consis:iid} and \ref{thm:consis:collinearity}, the required conditions for the consistency of $\bm{\hat\theta}_n$ assuming that $W_0$ is known are:
\begin{enumerate}[label=($\rm\tilde C$\arabic*), leftmargin=*, start=3]
\item \label{c:w0:compactness}
The parameter space $\Theta$ is compact
\item \label{c:w0:continuity}
The score function $S(\bm{\theta}, W_0; \bm{o})$ is continuous at $\bm\theta\in\Theta$ with probability one.
\item \label{c:w0:score-bound}
There exists a constant $c > 0$ such that $\|\mathcal{S}(\bm\theta; y, \bm x, z)\| \leq c$ for all $\bm\theta \in \Theta$, with probability one. This holds if $(\bm X, Z)$ have bounded support and $Y$ does not exhibit heavy tails.
\end{enumerate}

Let $\tilde g_n(\bm\theta)=g_n\big(\bm\theta, W_0\big)$ and $\tilde G(\bm\theta)=G\big(\bm\theta, W_0\big)$ denote the sample and population means of $S(\bm\theta, W;\bm o)$ if $W=W_0$ were known. Then the estimating equations $\tilde g_n(\bm\theta) = 0$ match the structure of the Generalized Method of Moments (GMM). 
Following Theorem 2.1 in \cite{newey1994b}, the consistency of $\bm{\hat\theta}_n$ requires showing that (i) the true parameter $\bm\theta_0$ is identifiable through the population equations $\tilde G(\bm\theta)=0$; and (ii) the sample mean $\tilde g_n(\bm\theta)$ converges in probability to the population mean $\tilde G(\bm\theta)$ uniformly for all $\bm\theta\in\Theta$.
Therefore, given that $S(\bm\theta, W_0;\bm o)$ is continuous in $\bm\theta$~\ref{c:w0:continuity}, and $\Theta$ is compact~\ref{c:w0:compactness}, $\bm{\hat\theta}_n$ converges in probability to $\bm\theta_0$.

\begin{lemmaA} \label{lemma:W0:iden}
    The true parameter $\bm\theta_0$ is identifiable, in the sense that there is a unique $\bm\theta=\bm\theta_0$ satisfying $\tilde G(\bm\theta)=0$.
\end{lemmaA}

\begin{proof}
The vector-valued function $S(\bm\theta, W_0;\bm o)$ corresponds to the score equations of the likelihood function given in equation~(\ref{eq:lik}):
\[ L(\bm\theta| W_0, \{y_i,\bm x_i,z_{l_i},z_{r_i}\}_{i=1}^n) = \prod_{i=1}^n \int_{z_{l_i}}^{z_{r_i}} f_{Y|\bm{X},Z}(y_i|\bm x_i,z;\bm\theta)\, dW_0(z). \]
It is well known that the expected value of the score equations, when evaluated at the true parameter, equals zero. Therefore, it is immediate to see that $\tilde G(\bm\theta_0)=0$.

Now, assume there exists another solution $\bm{\theta}_1\neq \bm{\theta}_0$ such that $\tilde G(\bm{\theta}_1) = \tilde G(\bm{\theta}_0) =0$.
This is only possible if $\bm{\theta}_1$ is also a maximum of the likelihood function. We will prove that different parameter values $\bm{\theta}_1\neq \bm{\theta}_0$ induce different likelihood functions. Suppose the likelihoods are identical, that is
\[ \int_{z_{l}}^{z_{r}} f_{Y|\bm{X},Z}(y|\bm x,z;\bm\theta_1)\, dW_0(z) = \int_{z_{l}}^{z_{r}} f_{Y|\bm{X},Z}(y|\bm x,z;\bm\theta_0)\, dW_0(z) \]
Therefore, for any pair $z_{l},z_{r}\in\Omega$,
\[ \int_{z_{l}}^{z_{r}} \big[f_{Y|\bm{X},Z}(y|\bm x,z;\bm\theta_1) - f_{Y|\bm{X},Z}(y|\bm x,z;\bm\theta_0)\big] dW_0(z)=0. \]
Taking $z_l=z_r$ we get that $f_{Y|\bm{X},Z}(y|\bm x,z;\bm\theta_1) = f_{Y|\bm{X},Z}(y|\bm x,z;\bm\theta_0)$ for all $z\in\Omega$,
which means that both probability distributions are identical. This contradicts~\ref{thm:consis:iid}, hence $\bm\theta_1=\bm\theta_0$.

\end{proof}

\begin{lemmaA}
    The sample-level mean $\tilde g_n(\bm\theta)$ converges uniformly in probability to the population-level mean $\tilde G(\bm\theta)$, that is $\sup_{\bm{\theta} \in \Theta} \| \tilde g_n(\bm{\theta}) - \tilde G(\bm{\theta}) \| \xrightarrow[n\to\infty]{p} 0$.
\end{lemmaA}
\begin{proof}
    The Uniform Law of Large Numbers (ULLN) \citep[][Lemma 2.4]{newey1994b} states that, under compactness of $\Theta$~\ref{c:w0:compactness} and continuity of $S(\bm{\theta}, W_0; \bm{o})$ at each $\bm{\theta}\in\Theta$ with probability one~\ref{c:w0:continuity}, if there exists a bound function $d(\bm o)$ such that $\big\| S(\bm{\theta}, W_0; \bm{o}) \big\| \leq d(\bm o)$ $\forall \bm\theta\in\Theta$ and $\mathbb{E}[d(\bm o)]<\infty$ then $G(\bm{\theta})$ is continuous and
    \[\sup_{\bm{\theta} \in \Theta} \| \tilde g_n(\bm{\theta}) - \tilde G(\bm{\theta}) \| \xrightarrow{p} 0 \quad \text{as } n \to \infty.\]
    From~\ref{c:w0:score-bound}, there exists a constant $c>0$ such that $\|\mathcal{S}(\bm\theta;y,\bm x, z)\| \leq c$ for all $\bm\theta\in\Theta$, with probability one.
    % Condition~\ref{c:w0:score-bound} guarantees that the score function of the non-censored problem is almost surely bounded. 
    The function $S(\bm{\theta}, W_0; \bm{o})$ is a definite integral of $\mathcal{S}(\bm\theta;y,\bm x, z)$ multiplied by a positive term. Therefore, it's also bounded.
    
\end{proof}

\subsubsubsection{Proof of Theorem \ref{thm:consistent}}

Under conditions \ref{c:iid}--\ref{c:compactness}, the GELc estimator is consistent, in the sense that $\| \bm{\hat\theta}_n - \bm\theta_0 \| \xrightarrow[]{p} 0$ as $n\to\infty$.

\begin{enumerate}[label=(C\arabic*), leftmargin=*]

 \item \label{c:iid}
    The observations $\{(y_i,z_{l_i},z_{r_i},\bm x_i)\}_{i=1}^n$ are independent and identically distributed (i.i.d.), with the true data-generating mechanism governed by the parameter $\bm{\theta}_0\in\Theta$, so that the model is correctly specified; and the non-observed $Z_1,\dots, Z_n$ are i.i.d. with true distribution $W_0\in\mathcal{W}$.

    \item \label{c:collinearity}
   The covariates $(\bm X, Z)$ cannot exhibit  perfect multicollinearity. That is, the design matrix $(\bm 1, \bm x_{\cdot 1},\dots, \bm x_{\cdot p}, \bm s)$, where $\bm x_{\cdot j}=(x_{1j},\dots,x_{nj})'$ for $j=1,\dots,p$ and $\bm s=(s_1,\dots, s_n)$ with $s_i\in\lfloor z_{l_i},z_{r_i}\rfloor$,   
    has full column rank.

    \item \label{c:consistent-w}
    $\widehat W_n$  has to converge unifomly to $W_0$ in the entire support, that is,
    \[
    \sup_{\bm\theta\in\Theta} \sup_{z\in\Omega} |\widehat W_n(z;\bm\theta)-W_0(z)| \xrightarrow[n\to\infty]{p} 0.
    \]

    \item \label{c:continuity}
    The vector-valued function $S(\bm{\theta}, W; \bm{o})$ is (i) continuous in $W$ near $W_0$ uniformly for all $\theta\in\Theta$; (ii) continuous with respect to $\bm\theta\in\Theta$ at $W=W_0$ with probability one. 

    \item \label{c:holder}
    Define $\left\| W_1 - W_2 \right\|_{\infty} = 
    \sup_{\bm\theta\in\Theta} \sup_{z\in\Omega} |W_1(z;\bm\theta)-W_2(z;\bm\theta)|$.
    Each component of $S(\bm{\theta}, W; \bm{o})$ is Hölder continuous in the sense of \cite{chen2003}, meaning that, for all $j=1,\dots,p+3$,
    \[
    \left| S_j(\bm{\theta}_1, W_1; \bm{o}) - S_j(\bm{\theta}_2, W_2; \bm{o}) \right| 
    \leq b_j(z) \left\{ 
    \left\| \bm\theta_1 - \bm\theta_2 \right\|^{s_{1j}} 
    + \left\| W_1 - W_2 \right\|_{\infty}^{s_j} 
    \right\}
    \]
    for some constants $s_{1j}, s_j \in (0, 1]$ and a measurable function $b_j(\cdot)$ with $\mathbb{E}[b_j(Z)]^r < \infty$ for some $r \geq 2$.

    \item \label{c:compactness}
    The parameter space $\Theta$ is compact, and the function space $\mathcal{W}$ satisfies, for each $j=1,\dots,p+3$, the entropy bound
    \[
    \int_0^\infty \sqrt{\log N\bigl(\varepsilon^{1/s_j}, \mathcal{W}, \|\cdot\|_{\infty}\bigr)}\, d\varepsilon < \infty,
    \]
    where 
    $N(\varepsilon,\mathcal W, ||\cdot||_{\infty})$ denotes the covering number of $\mathcal W$, that is, the minimal number of balls of radius $\varepsilon$ in the sup-norm $ \left\| W_1 - W_2 \right\|_{\infty}$ required to cover $\mathcal{W}$. 
\end{enumerate}

\medskip
To prove the consistency of $\bm{\hat\theta}_n$ we follow 
Theorem 1 in \cite{chen2003}, which provides conditions for consistency of estimators that approximately solve the estimating equations $g_n\big(\bm\theta, \widehat W(\cdot,\bm\theta)\big) = 0$, provided the true parameter $\bm\theta_0$ is identifiable by $g_n\big(\bm\theta, W_0\big) = 0$ (proved in Lemma \ref{lemma:W0:iden}).
The leading ideas are the following.
\begin{enumerate}
    \item The value of the estimating function at the chosen estimate $\bm{\hat\theta}_n$ is asymptotically as small as possible, up to a negligible error, that is,
    \[ \| g_n\big(\bm{\hat\theta}_n,\widehat W_n(\cdot;\bm{\hat\theta}_n)\big) \| \leq 
    \inf_{\bm\theta\in\Theta} \| g_n\big(\bm\theta,\widehat W_n(\cdot;\bm\theta)\big) \| + o_p(1). \]
    This condition states that $\bm{\hat\theta}_n$ approximately solves the estimating equations, in the sense that it minimizes  $g_n\big(\bm\theta,\widehat W_n(\cdot;\bm\theta)\big)$ up to a term that converges to zero in probability as $n\to\infty$. We ensure this point by computing $\bm{\hat\theta}_n$ via numerical optimization with a small tolerance, so that the algorithmic error is asymptotically negligible.
    
    \item The population moment equation $G\big(\bm\theta,W_0\big)=0$ is uniquely solved at $\bm\theta_0$ and is bounded away from zero elsewhere. Formally, $G(\bm\theta_0,W_0)=0$ and, for all $\delta>0$, there exists $\epsilon(\delta)>0$ such that $\inf_{\|\bm\theta-\bm\theta_0\| >\delta} G(\bm\theta,W_0) \geq \epsilon(\delta)>0$. This condition holds in our setting by Lemma \ref{lemma:W0:iden} (identifiability) together with the continuity assumption \ref{c:continuity}.

    \item With respect to the chosen metric of the type $\| W_1 - W_2 \|_{\mathcal W} = \sup_{\bm\theta\in\Theta} \| W_1(\cdot;\bm\theta) - W_2(\cdot;\bm\theta) \|$, the estimator $\widehat W_n(\cdot;\bm\theta)$ satisfies
    \[\| \widehat W_n - W_0 \|_{\mathcal W} \xrightarrow[n\to\infty]{p} 0.  \]
    In our case, condition \ref{c:consistent-w} is formulated equivalently with the metric
    \[
    \left\| W_1 - W_2 \right\|_{\infty} = 
    \sup_{\bm\theta\in\Theta} \sup_{z\in\Omega} |W_1(z;\bm\theta)-W_2(z;\bm\theta)|.
    \]
    Moreover, for the augmented Turnbull estimator, condition \ref{c:consistent-w} holds by Theorem~\ref{thm:TBconsistent}, so this is satisfied in our setting.

    \item Finally, uniform convergence of $g_n$ to $G$ must hold in a neighborhood of $W_0$. Specifically, for all positive sequences converging to zero $\delta_n=o_p(1)$,
    \begin{equation*}
        \sup_{\substack{\|W - W_0\|_{\mathcal W} \leq \delta_n \\ \bm\theta \in \Theta}}
        \| g_n(\bm\theta,W) - G(\bm\theta,W) \| \xrightarrow[n\to\infty]{p} 0.
    \end{equation*}
    This is implied by Lemma~\ref{lem:stochastic_equicontinuity} (via the triangle inequality), and holds under \ref{c:holder} and \ref{c:compactness}.
\end{enumerate}

\begin{lemmaA} \label{lem:stochastic_equicontinuity}
    \textbf{Stochastic equicontinuity.}
    Under regularity conditions \ref{c:holder} and \ref{c:compactness}, 
    for all positive sequences $\delta_n$ such that $\delta_n \xrightarrow[]{p} 0$,
    \[
    \sup_{\substack{\|W - W_0\|_{\infty} \leq \delta_n \\ \bm\theta \in \Theta}} \| g_n(\bm\theta,W) - G(\bm\theta,W) - g_n(\bm\theta_0,W_0) \| = o_p(n^{-1/2}).
    \]
\end{lemmaA}

\begin{proof}
    Consider the class of functions
    \[
    \mathcal{G}_\delta 
    = \bigl\{ S(\bm{\theta}, W; \bm{o}) : \bm\theta \in \Theta,\, \|W-W_0\|_{\infty} \leq \delta \bigr\},
    \]
    indexed by $(\bm\theta, W)$ and equipped with the $L_2(P)$ metric, a natural choice in empirical process theory \citep[see][]{Vaart1996}, defined as
    \[
    \bigl\| S(\bm\theta_1, W_1; \bm{o}) - S(\bm\theta_2, W_2; \bm{o}) \bigr\|_{L_2(P)}
    = 
    \Bigg(
    \mathbb{E}_{\bm O}\Big[
    \sum_{j=1}^{p} 
    \big(
    S_j(\bm\theta_1, W_1; \bm{o}) 
    -
    S_j(\bm\theta_2, W_2; \bm{o})
    \big)^2
    \Big]
    \Bigg)^{1/2}.
    \]

    As stated in Section~4 of \cite{chen2003}, verifying stochastic equicontinuity reduces to bounding the bracketing number
    \[
    N_{[]}(\varepsilon, \mathcal{G}_\delta, \|\cdot\|_{L_2(P)}).
    \]
    If each component of $S(\bm{\theta}, W; \bm{o})$ is Hölder continuous as in Condition~\ref{c:holder}, then the bracketing number of $\mathcal{G}_\delta$ can be controlled by the covering number of the parameter class
    \[
    \bigl\{ (\bm\theta, W): \bm\theta \in \Theta,\, \|W - W_0\|_{\infty} \le \delta \bigr\}.
    \]

    Since $\Theta$ is a compact subset of $\mathbb{R}^p$, its covering number is finite. Moreover, Condition~\ref{c:compactness} bounds the complexity of $\mathcal{W}$ by controlling the growth of its covering number as $\varepsilon \to 0$. Therefore, the entropy of $\mathcal{G}_\delta$ is uniformly bounded as $\delta \to 0$, and the stochastic equicontinuity condition follows.
\end{proof}

    % \[
    % \sup_{\substack{\|W - W_0\|_{\infty} \leq \delta_n \\ \|\bm\theta-\bm\theta_0\|\leq \delta_n}} \| g_n(\bm\theta,W) - G(\bm\theta,W) - g_n(\bm\theta_0,W_0) \| = o_p(n^{-1/2}).
    % \]

\subsubsection{\texorpdfstring{Asymptotic normality of $\hat\theta_n$}{Asymptotic normality of the GELc estimator}}

% definition

% We first show that if $W=W_0$ is known then $\bm{\hat\theta}_n$ is asymptotically normal.

% Later extend the result to the GELc estimator with estimated $W$, here summarise the result of how the asymptotics of $\widehat W$ affect and why at the end variance matrix $I$.

\subsubsubsection{\texorpdfstring{Asymptotic normality assuming $W_0$ is known}{Asymptotic normality assuming W0 is known}}

Assume that conditions \ref{thm:consis:iid} and \ref{thm:consis:collinearity} of Theorem 3 hold, and that conditions \ref{c:w0:compactness} to \ref{c:w0:score-bound} (which guarantee the consistency of $\bm{\hat\theta}_n$ if $W_0$ is known) are satisfied.

Under the regularity conditions \ref{c:w0:interior} to \ref{c:w0:nonsingular}, we have
\[
\sqrt{n}(\bm{\hat\theta}_n-\bm\theta_0)\xrightarrow{d} N(0,\tilde I(\bm\theta_0)^{-1})
\]
where $\tilde I(\bm\theta_0)$ is the expected information matrix at the true value.

\begin{enumerate}[label=($\rm\tilde R$\arabic*), leftmargin=*]
    
    \item \label{c:w0:interior}
    The true parameter $\bm\theta_0$ is interior to the compact parameter space $\Theta$.

    \item \label{c:w0:c1}
    The score function $S(\bm{\theta}, W_0; \bm{o})$ is continuously differentiable in a neighborhood $\Theta_\delta$ of $\bm\theta_0$ with probability approaching one; and $\mathbb{E}\left[\sup_{\theta \in \Theta_\delta} \| \frac{d}{d \bm\theta} S(\bm\theta, W_0; \bm o)\|\right] < \infty$.

    \item \label{c:w0:bounded}
    For $\bm{\theta}_0$ and $W_0$ fixed, $\mathbb{E}[\| S(\bm{\theta}_0, W_0; \bm{o})\|^2]<\infty$, so that the information matrix exists.
        
    \item \label{c:w0:nonsingular}
    The Jacobian matrix for given $z$, $\frac{d}{d\bm\theta} \mathcal{S}(\bm\theta; y,\bm x,z)$, is nonsingular at $\theta_0$, ie it has full column rank.
\end{enumerate}

\bigskip

Developing the first-order Taylor expansion of $g_n(\bm\theta)$ around $\bm\theta_0$, which is well-defined because of \ref{c:w0:interior}, it can be seen that $\sqrt{n}(\bm{\hat\theta}_n-\bm\theta_0)$ is asymptotically equivalent to
\begin{equation} \label{eq:AsEquiv0}
    \sqrt{n}(\bm{\hat\theta}_n-\bm\theta_0) \approx
    \big[- \tilde J_n(\bm\theta_0)^{-1}\big] 
    \big[\sqrt{n} \tilde g_n(\bm\theta_0) \big],
\end{equation}
where $\tilde J_n(\bm\theta_0)= \frac{d}{d \bm\theta}\tilde g_n(\bm\theta)\rvert_{\bm\theta=\bm\theta_0}$ is the Jacobian of the estimating function $\tilde g_n$ evaluated at the true parameter. 
We can write
\[
\tilde J_n(\bm\theta_0) = n^{-1} \sum_{i=1}^n \frac{d}{d \bm\theta} S(\bm\theta, W_0; \bm o_i) \rvert_{\bm\theta=\bm\theta_0}.
\]
Since $\mathbb{E}\left[\sup_{\theta \in \Theta_\delta} \| \frac{d}{d \bm\theta} S(\bm\theta, W_0; \bm o)\|\right] < \infty$~\ref{c:w0:c1}, the iid random variables $\frac{d}{d \bm\theta} S(\bm\theta, W_0; \bm o_i) \rvert_{\bm\theta=\bm\theta_0}$, $i=1,\dots,n$, have finite expectation and so, by the law of large numbers, the sample mean $\tilde J_n(\bm\theta_0)$ converges in probability to the expectation 
\[ \mathbb{E}\big[\frac{d}{d \bm\theta} S(\bm\theta, W_0; \bm o) \rvert_{\bm\theta=\bm\theta_0}\big] = \frac{d}{d \bm\theta} \mathbb{E}\big[S(\bm\theta, W_0; \bm o)\big]\rvert_{\bm\theta=\bm\theta_0} = \frac{d}{d\bm\theta}\tilde G(\bm\theta)\rvert_{\bm\theta=\bm\theta_0} = \tilde J(\bm\theta_0).\]
That is, $\tilde J_n(\bm\theta_0)\to \tilde J(\bm\theta_0)$ in probability as $n\to\infty$.

On the other hand, the sequence of random variables $S(\bm\theta, W_0; \bm o_i) \rvert_{\bm\theta=\bm\theta_0}$, $i=1,\dots,n$, are iid with mean 
$\mathbb{E}[S(\bm\theta_0, W_0; \bm o)]=G(\bm\theta_0)=0$ and variance
\[
    {\rm Var}\big(S(\bm\theta_0, W_0; \bm o)\big) = \mathbb{E}\big[ S(\bm\theta_0, W_0; \bm o) S(\bm\theta_0, W_0; \bm o)'\big] 
    - \underbrace{\mathbb{E}\big[S(\bm\theta_0, W_0; \bm o)\big]\mathbb{E}\big[S(\bm\theta_0, W_0; \bm o)\big]'}_{0}
    = \tilde I(\bm\theta_0).
\]

Since $\tilde I(\bm\theta_0)<\infty$~\ref{c:w0:bounded}, by the Central Limit Theorem,
\[
\sqrt{n} \frac{1}{n}\sum_{i=1}^n S(\bm\theta_0, W_0; \bm o_i) = \sqrt{n} \tilde g_n(\bm\theta_0) \to N\Big(0,\, \tilde I(\bm\theta_0)\Big)
\]
in distribution as $n\to\infty$.

Notice that Fisher's information matrix is equal to
\[
\tilde I(\bm\theta_0) = \mathbb{E}\big[ S(\bm\theta_0, W_0; \bm o) S(\bm\theta_0, W_0; \bm o)'\big] = -\frac{d}{d\bm\theta} \mathbb{E}\big[S(\bm\theta, W_0; \bm o)\big]\rvert_{\bm\theta=\bm\theta_0} = -\tilde J(\bm\theta_0).
\]

By Slutsky's theorem, since $\tilde J_n(\bm\theta_0)\xrightarrow[]{p} \tilde J(\bm\theta_0)$ and $\sqrt{n} \tilde g_n(\bm\theta_0) \xrightarrow[]{d} N\big(0,\tilde I(\bm\theta_0)\big)$,
\[ 
    \big[- \tilde J_n(\bm\theta_0)^{-1}\big] 
    \big[\sqrt{n} \tilde g_n(\bm\theta_0) \big]
    \xrightarrow[n\to\infty]{d} N(0, \tilde\Sigma) 
\]
with 
\[
\tilde\Sigma= \tilde J(\bm\theta_0)^{-1} \tilde I(\bm\theta_0) (\tilde J(\bm\theta_0)^{-1})' = 
\tilde I(\bm\theta_0)^{-1}.
\]

The two expressions in equation~(\ref{eq:AsEquiv0}) are equal except for a remaining term $- \sqrt{n} \tilde J_n(\bm\theta_0)^{-1} o_p(|\bm{\hat\theta}_n-\bm\theta_0|)$, which is asymptotically negligible. Indeed, first, $|\bm{\hat\theta}_n-\bm\theta_0|$ converges in probability to $0$, that is, $o_p(|\bm{\hat\theta}_n-\bm\theta_0|)=o_p(1)$ because of consistency. Secondly, $\tilde J_n(\bm\theta_0)\xrightarrow[]{p} \tilde J(\bm\theta_0)$, so $\tilde J_n(\bm\theta_0)$ is bounded in probability. Therefore, $\tilde J_n(\bm\theta_0)^{-1} o_p(|\bm{\hat\theta}_n-\bm\theta_0|)$ stills converges in probability to $0$, and the remaining term is $- \sqrt{n}o_p(1)$. The product of a $o_p(1)$ term with $\sqrt{n}$ grows strictly slower than $\sqrt{n}$. Consequently, the remaining term is asymptotically negligible compared to the dominant term in the expansion of $\sqrt{n}(\bm{\hat\theta}_n-\bm\theta_0)$.

\medskip

In the following lemma we prove that $\tilde J(\bm\theta_0)$ is invertible under regularity condition~\ref{c:w0:nonsingular}, which in turn implies the invertibility of $\tilde I(\bm\theta_0)$.

\begin{lemmaA} \label{lem:nonsingular}
    If \ref{c:w0:nonsingular}, then the Jacobian matrix $\tilde J(\bm\theta) = \frac{d}{d\bm\theta}\tilde G(\bm\theta)$ is nonsingular at $\bm\theta=\bm\theta_0$. 
\end{lemmaA}
\begin{proof}
            
    Recall that under Assumptions~\ref{A1} to~\ref{A3}, the law of $\bm O$ can be expressed as
    \begin{align*}
        &Pr\big(Y\in dy_i, \bm X\in d\bm{x_i},
        Z_L\in dz_{l_i}, Z_R\in dz_{r_i}, 
        Z_i\in\lfloor z_{l_i},z_{r_i} \rfloor\big) 
        \\
        &=\Big( \int_{z_{l}}^{z_{r}} f_{Y|\bm{X},Z}(y|\bm x,z;\bm\theta)\, dW(z) \Big) \cdot
        \Big(K(z_{l}, z_{r}, \bm x) Pr(\bm X\in d\bm{x}) \Big).
    \end{align*}

    Therefore,
    \begin{align*}
        \frac{d}{d\bm\theta}&\tilde G(\bm\theta) =\frac{d}{d\bm\theta} \mathbb{E}\big[S(\bm{\theta}, W_0; \bm{o})\big] \\[1ex]
       &=\frac{d}{d\bm\theta} \int_{{\rm Sup}(\bm O)}
         \Big\{ \int_{z_{l}}^{z_{r}}
            \mathcal{S}(\bm\theta; y,\bm x,z)
            f_{Y|\bm{X},Z}(y|\bm x,z;\bm\theta) 
            \, dW_0(z)
        \Big\}
        \Big(K(z_{l}, z_{r}, \bm x) Pr(\bm X\in d\bm{x}) \Big)\\
        & = \int_{{\rm Sup}(\bm O)}  
         \Big\{ \int_{z_{l}}^{z_{r}}
             \frac{d}{d\bm\theta}\Big[\mathcal{S}(\bm\theta; y,\bm x,z)
            f_{Y|\bm{X},Z}(y|\bm x,z;\bm\theta) \Big]
            \, dW_0(z)
        \Big\}
        \Big(K(z_{l}, z_{r}, \bm x) Pr(\bm X\in d\bm{x}) \Big) = *
    \end{align*}
    Now,
    \begin{align*}
        &\frac{d}{d\bm\theta}\Big[\mathcal{S}(\bm\theta; y,\bm x,z) f_{Y|\bm{X},Z}(y|\bm x,z;\bm\theta) \Big]\\
        & =  
        \frac{d \mathcal{S}(\bm\theta; y,\bm x,z)}{d \bm\theta} f_{Y|\bm{X},Z}(y|\bm x,z;\bm\theta) 
        + \mathcal{S}(\bm\theta; y,\bm x,z) \frac{d f_{Y|\bm{X},Z}(y|\bm x,z;\bm\theta)}{d\bm\theta}\\
        & = \frac{d \mathcal{S}(\bm\theta; y,\bm x,z)}{d \bm\theta} f_{Y|\bm{X},Z}(y|\bm x,z;\bm\theta) 
        + \mathcal{S}(\bm\theta; y,\bm x,z)\mathcal{S}(\bm\theta; y,\bm x,z)' f_{Y|\bm{X},Z}(y|\bm x,z;\bm\theta).
    \end{align*}
    Hence,
    \begin{align*}
        & * = \int_{{\rm Sup}(\bm O)} \int_{z_{l}}^{z_{r}}
         \Big[\frac{d}{d\bm\theta} \mathcal{S}(\bm\theta; y,\bm x,z) + \mathcal{S}(\bm\theta; y,\bm x,z)\mathcal{S}(\bm\theta; y,\bm x,z)' \Big] \cdot \\        
         & \quad Pr\big(Y\in dy_i, \bm X\in d\bm{x_i}, Z_L\in dz_{l_i}, Z_R\in dz_{r_i}, Z_i\in dz\big) dz
    \end{align*}
    where $\frac{d}{d\bm\theta} \mathcal{S}(\bm\theta; y,\bm x,z)$ is the Jacobian matrix for given $z$. 

    By condition \ref{c:w0:nonsingular}, the matrix $\frac{d}{d\bm\theta} \mathcal{S}(\bm\theta; y,\bm x,z)$ is non-singular at $\bm\theta_0$.    
    The Sherman–Morrison principle states that if a matrix $A$ is invertible and $u$ is a vector, then
    \[A + uu' \text{ invertible} \quad{\rm iff}\quad u'A^{-1}u \neq -1.\]
    Since $\mathcal{S}(\bm\theta; y,\bm x,z)'\cdot \frac{d}{d\bm\theta} \mathcal{S}(\bm\theta; y,\bm x,z)^{-1}\cdot\mathcal{S}(\bm\theta; y,\bm x,z) = -1$ is an unlikely degenerate case, we conclude that $\tilde J(\bm\theta)$ is non-singular at $\bm\theta_0$.
        
\end{proof}

\subsubsubsection{Proof of Theorem \ref{thm:asymnorm}}

Assume that conditions \ref{c:iid} to \ref{c:compactness} for the GELc estimator $\bm{\hat\theta}_n$ to be consistent are satisfied and $\bm\theta_0\in {\rm int}(\Theta)$. Following \cite{chen2003}, define the shrinking sets $\Theta_\delta = \{\theta \in \Theta : \|\theta - \theta_0\| \leq \delta\}$ and $\mathcal W_\delta = \{ W\in\mathcal W : \sup_{\bm\theta\in\Theta_\delta} \|W(\cdot;\bm\theta)-W_0(\cdot)\| \leq\delta\}$ for some arbitrary small $\delta > 0$.

Recall 
\begin{align*}
    &g_n\big(\bm\theta,\widehat W_n(\cdot;\bm\theta)\big) = n^{-1} \sum_{i=1}^n S(\bm\theta,\widehat W_n(\cdot;\bm\theta);\bm o_i), \\
    &G\big( \bm\theta,\widehat W_n(\cdot;\bm\theta) \big) = \mathbb{E}\big[ S(\bm\theta,\widehat W_n(\cdot;\bm\theta);\bm o)\big].
\end{align*}

\medskip

The following regularity conditions are required to establish the asymptotic normality of $\bm{\hat\theta}_n$:

\begin{enumerate}[label=(R\arabic*), leftmargin=*]
    
    \item \label{c:c1theta}
    The function $S(\bm\theta, W_0; \bm o)$ is differentiable in $\Theta_\delta$, and the derivative $\frac{d}{d\bm\theta}S(\bm\theta, W_0; \bm o)$ is continuous at $\bm\theta=\bm\theta_0$

    \item \label{c:nonsingular}
    The Jacobian matrix for given $z$, $\frac{d}{d\bm\theta} \mathcal{S}(\bm\theta; y,\bm x,z)$, is nonsingular at $\theta_0$, ie it has full column rank.
    
    \item \label{c:c1w}
    For all $\bm\theta\in\Theta_\delta$, $S(\bm\theta, W; \bm o)$ is differentiable at $W_0$ in all directions $[W-W_0]\in\mathcal W$, meaning that $\{W_0 + h(W - W_0) : h \in [0, 1]\} \subset \mathcal{W}$ and the limit 
    \[
    \lim_{h\to 0} \frac{1}{h} \Big[ 
    S\big(\bm\theta, W_0(\cdot) + h\{W(\cdot;\bm\theta)-W_0(\cdot)\}; \bm o\big) -
    S(\bm\theta, W_0; \bm o )
    \Big] 
    \]
    exists for all $W$ such that $[W-W_0]\in\mathcal W$. The limit defines the gateaux derivative of $S(\bm\theta, W; \bm o)$ at $W_0$ in the direction $[W-W_0]$, denoted by $D S(\bm\theta, W_0; \bm o)[W-W_0]$.
     
    Consider a positive sequence $\delta_n$ such that $\delta_n\xrightarrow{p} 0$. 
    For all $(\bm\theta, W)\in \Theta_{\delta_n}\times\mathcal{W}_{\delta_n}$,
    \begin{enumerate}
        \item[(i)] $\|S(\bm\theta, W; \bm o) - S(\bm\theta, W_0; \bm o) - D S(\bm\theta, W_0; \bm o)[W-W_0]\| \le b(\bm o) \|W - W_0\|_{\infty}^2$,  with $\mathbb{E}[b(\bm o)]<\infty$.
        \item[(ii)] $\|D S(\bm\theta, W_0; \bm o)[W-W_0] - D S(\bm\theta_0, W_0; \bm o)[W-W_0] \| \le o(1) \delta_n$.
    \end{enumerate}
    
    \item \label{c:what}
    $\|\widehat W_n - W_0\|_{\infty} = o_p(n^{-1/4})$; and
    $\widehat W_n\in\mathcal W$ with probability tending to one.
\end{enumerate}

\begin{lemmaA}
    \label{lem:orthogonality}
    Let $\widehat{W}_n$ be a consistent estimator of $W_0$ such that $ \sup_{\bm\theta\in\Theta}\|\widehat W_n(\cdot;\bm\theta) - W_0(\cdot)\|_{\infty} = o_p(n^{-1/4})$, and suppose the moment function $S(\bm\theta,W;\bm o)$ is gateaux differentiable in $W$ at $W_0$, and behaves smoothly around $(\bm\theta_0, W_0)$ as \ref{c:c1theta} and \ref{c:c1w}.
 Then, the moment function is asymptotically orthogonal along the direction $[\widehat{W}_n - W_0]$, that is,
    \[
        \mathbb{E}\big[ D S(\bm\theta_0, W_0; \bm o)[\widehat{W}_n - W_0] \big] = o_p(n^{-1/2}).
    \]
\end{lemmaA}
\begin{proof}
    Let $\Delta = [W-W_0]\in\mathcal W$ denote any perturbation of $W$. Since $S(\bm\theta,W;\bm o) = \frac{d}{d\bm\theta} \ell(\bm\theta,W;\bm o)$ is the score vector,
    \[ 
    \mathbb{E}\big[ D S(\bm\theta_0,W_0;\bm o)[\Delta]\big] = \mathbb{E}\bigg[ \frac{d}{d\bm\theta}
    D \ell(\bm\theta,W_0;\bm o)[\Delta]
    \bigg]\rvert_{\bm\theta=\bm\theta_0}.
    \]    
    Using the same results as in~\appref{AppA3},
    \begin{gather*}
        D \ell(\bm\theta,W_0;\bm o)[\Delta] = 
         \lim_{h\to0} \frac{1}{h} \big(\ell(\bm\theta,W_0+h\Delta;\bm o) - \ell(\bm\theta,W_0;\bm o)\big) =\\
        \lim_{h\to0} \frac{1}{h} \log\frac{L(\bm\theta,W_0;\bm o)+hL(\bm\theta,\Delta;\bm o)}{L(\bm\theta,W_0;\bm o)} =
        \frac{L(\bm\theta,\Delta;\bm o)}{L(\bm\theta,W_0;\bm o)}.
    \end{gather*}
    
    Now, replacing $\Delta = [W-W_0]$,
    \begin{gather*}
        \frac{d}{d\bm\theta}\frac{L(\bm\theta,\Delta;\bm o)}{L(\bm\theta,W_0;\bm o)} = \frac{d}{d\bm\theta}\bigg( \frac{L(\bm\theta,W;\bm o)}{L(\bm\theta,W_0;\bm o)} - 1\bigg) = \\[1ex]
        \frac{\big(\frac{d}{d\bm\theta}L(\bm\theta,W;\bm o)\big)\cdot L(\bm\theta,W_0;\bm o) -
        L(\bm\theta,W;\bm o)\cdot\big(\frac{d}{d\bm\theta}L(\bm\theta,W_0;\bm o)\big)}{L(\bm\theta,W_0;\bm o)^2} = \\[1ex]
        \frac{\frac{d}{d\bm\theta}L(\bm\theta,W;\bm o)}{L(\bm\theta,W_0;\bm o)} - 
        \frac{L(\bm\theta,W;\bm o)}{L(\bm\theta,W_0;\bm o)}\cdot
        \frac{\frac{d}{d\bm\theta}L(\bm\theta,W_0;\bm o)}{L(\bm\theta,W_0;\bm o)} =\\
        \frac{L(\bm\theta,W;\bm o)}{L(\bm\theta,W_0;\bm o)} 
        \Big( S(\bm\theta,W;\bm o) - S(\bm\theta,W_0;\bm o) \Big)
    \end{gather*}  
    since $S(\bm\theta,W;\bm o) = \frac{d}{d\bm\theta} \ell(\bm\theta,W;\bm o) = \frac{\frac{d}{d\bm\theta}L(\bm\theta,W;\bm o)}{L(\bm\theta,W;\bm o)}$.

    Recall that, under Assumptions~\ref{A1}) to~\ref{A3}, the true law of $\bm O = (Y,Z_L,Z_R,\bm X)$ can be expressed as
    \begin{gather*}
        Pr\big(Y\in dy, \bm X\in d\bm{x},
        Z_L\in dz_{l}, Z_R\in dz_{r}, 
        Z\in\lfloor z_{l},z_{r} \rfloor\big) =   \\
        \Big( \int_{z_{l}}^{z_{r}} f_{Y|\bm{X},Z}(y|\bm x,z;\bm\theta_0)\, dW_0(z) \Big) \cdot
        \Big(K(z_{l}, z_{r}, \bm x) Pr(\bm X\in d\bm{x}) \Big) =\\
        L(\bm\theta_0,W_0;\bm o) K(z_{l}, z_{r}, \bm x) Pr(\bm X\in d\bm{x}).
    \end{gather*}
    Evaluating at $\bm\theta_0$ and taking expectation,
    \begin{gather*}
        \mathbb{E}\bigg[
        \frac{d}{d\bm\theta}D \ell(\bm\theta,W_0;\bm o)[\Delta]
        \rvert_{\bm\theta=\bm\theta_0}\bigg] = \\
        \mathbb{E}\bigg[ 
        \frac{L(\bm\theta_0,W;\bm o)}{L(\bm\theta_0,W_0;\bm o)} 
        \bigg( S(\bm\theta_0,W;\bm o) - S(\bm\theta_0,W_0;\bm o) \bigg)
        \bigg] =\\
        \int_{{\rm Sup}(\bm O)} 
        \big( S(\bm\theta_0,W;\bm o) - S(\bm\theta_0,W_0;\bm o) \big) 
        L(\bm\theta_0,W;\bm o)\, K(z_{l}, z_{r}, \bm x) Pr(\bm X\in d\bm{x}) dy\, dz_l\, dz_r\,  d\bm x,
    \end{gather*}
    which is exactly zero only for $W=W_0$, and it is $o_p(n^{-1/2})$ for $W=\widehat W_n$. This rate relies on condition \ref{c:c1w}(i), so that the impact of estimating $W$ vanishes faster than $n^{-1/2}$ as long as $\widehat W_n$ converges to $W_0$ at a rate faster than $n^{-1/4}$.
\end{proof}

The first-order Taylor expansion of $g_n\big(\bm\theta,W\big)$ around $(\bm\theta_0,W_0)$ develops into
\begin{equation}\label{eq:AsEquiv1}
    \sqrt{n}(\bm{\hat\theta}_n-\bm\theta_0) \approx
    -J_n(\bm\theta_0,W_0)^{-1} \sqrt{n} \big\{ g_n(\bm\theta_0,W_0) + D g_n(\bm\theta_0,W_0)[\widehat W_n-W_0] \big\},
\end{equation}
where
$J_n(\bm\theta_0,W_0)= \frac{d}{d \bm\theta}g_n(\bm\theta,W_0)\rvert_{\bm\theta=\bm\theta_0}$ is the Jacobian of $g_n$ with respect to $\bm\theta$, and
$D g_n(\bm\theta_0,W_0)[\widehat W_n-W_0]$ is the gateaux derivative of $g_n$ at $W_0$ in the direction $[\widehat W_n-W_0]$. These derivatives are well defined under the regularity conditions  \ref{c:c1theta} and \ref{c:c1w}.

\medskip

On one hand, the Jacobian can be expressed as the sample average
\[
J_n(\bm\theta_0,W_0) = n^{-1} \sum_{i=1}^n \frac{d}{d \bm\theta} S(\bm\theta, W_0; \bm o_i) \rvert_{\bm\theta=\bm\theta_0}.
\]
To apply the law of large numbers, we need that the random variable $\frac{d}{d \bm\theta} S(\bm\theta, W_0; \bm o) \rvert_{\bm\theta=\bm\theta_0}$ has finite expectation. A sufficient condition for this is
\[
\mathbb{E} \left[ \sup_{\bm\theta \in \Theta_\delta} \left\| \frac{d}{d \bm\theta} S(\bm\theta, W_0; \bm o) \right\| \right] < \infty,
\]
which is satisfied due to the Hölder continuity assumption in \ref{c:holder} together with the smoothness condition \ref{c:c1theta}.

Therefore, by the law of large numbers, $J_n(\bm\theta_0,W_0)\to  J(\bm\theta_0,W_0)$ in probability as $n\to\infty$, where
\[
J(\bm\theta_0,W_0) = \frac{d}{d \bm\theta} G(\bm\theta,W_0)\rvert_{\bm\theta=\bm\theta_0} = \mathbb{E}\big[\frac{d}{d \bm\theta} S(\bm\theta, W_0; \bm o) \rvert_{\bm\theta=\bm\theta_0}\big].
\]

\medskip
On the other hand,
\begin{gather}
    \big\{ g_n(\bm\theta_0,W_0) + D g_n(\bm\theta_0,W_0)[\widehat W_n-W_0] \big\} = \notag\\
    \Big( n^{-1} \sum_{i=1}^n S(\bm\theta_0,W_0;\bm o_i) \Big) +
    \Big( n^{-1} \sum_{i=1}^n D S(\bm\theta_0,W_0;\bm o_i)[\widehat W_n-W_0] \Big). \label{eq:limiting}
\end{gather}
In Lemma \ref{lem:orthogonality} we showed that under regularity conditions \ref{c:c1theta} and \ref{c:c1w}, if $
\|\widehat W_n - W_0\|_{\infty} = o_p(n^{-1/4})$ \ref{c:what} then 
\[
    \mathbb{E}\big[ D S(\bm\theta_0, W_0; \bm o)[\widehat{W}_n - W_0] \big] = o_p(n^{-1/2}).
\]
In addition, under the given regularity conditions, the random variable $D S(\bm\theta_0, W_0; \bm o)[\widehat{W}_n - W_0]$ has uniformly bounded variance. As a result, the second term in \eqref{eq:limiting} is asymptotically negligible for the limiting distribution of $n^{1/2}(\bm{\hat\theta}_n-\bm\theta_0)$. 
By the Central Limit Theorem, it is straightforward that:
\[
\sqrt{n} \big\{ g_n(\bm\theta_0,W_0) + D g_n(\bm\theta_0,W_0)[\widehat W_n-W_0] \big\} \xrightarrow[]{d} N\big(0, I(\bm\theta_0,W_0)\big)
\]
where
\[
I(\bm\theta_0,W_0) = \rm{Var}\big( S(\bm\theta_0,W_0;\bm o)\big) = \mathbb{E} \big[ S(\bm\theta_0,W_0;\bm o)S(\bm\theta_0,W_0;\bm o)' \big]
\]
denotes the expected information matrix. 
The condition $\mathbb{E}[\|S(\bm\theta_0,W_0;\bm o)\|^2]<\infty$, required for the CLT to apply, is implied by the Hölder continuity condition \ref{c:holder}.

\bigskip

In Theorem 2, \cite{chen2003} give sufficient conditions under which the second-order remainder in the Taylor expansion is $o_p(1)$, so that $ \sqrt{n}(\widehat{\bm\theta}_n - \bm\theta_0) $ is asymptotically equivalent to the right-hand side of \eqref{eq:AsEquiv1}. Besides the regularity conditions stated at the beginning of the proof, it is required that 
\begin{itemize}
    \item The estimator $\bm{\hat\theta}_n$ approximately minimizes the sample moment condition within an $o_p(n^{-1/2})$ neighborhood. Formally,
    \[
    \| g_n\big(\bm{\hat\theta}_n,\widehat W_n(\cdot;\bm{\hat\theta}_n)\big) \| =
        \inf_{\bm\theta\in\Theta_\delta} \| g_n\big(\bm\theta,\widehat W_n(\cdot;\bm\theta)\big) \| +o_p(1/\sqrt{n}).
    \]
    As mentioned in the proof of Theorem 3, this point is ensured via the numerical optimization algorithm.

    \item Stochastic equicontinuity condition: for all positive sequences $\delta_n=o_p(1)$,
    \begin{equation*}
        \sup_{\substack{\bm\theta\in\Theta_{\delta_n},\\
        W\in\mathcal W_{\delta_n}}}
        \| g_n(\bm\theta,W) - G(\bm\theta,W) - g_n(\bm\theta_0,W_0) \| = o_p(n^{-1/2}).
    \end{equation*}
    This is proven in Lemma \ref{lem:stochastic_equicontinuity}.
\end{itemize}

Wrapping up,
\[
J_n(\bm\theta_0,W_0)\xrightarrow[]{p}  J(\bm\theta_0,W_0)
\]
\[
\sqrt{n} \big\{ g_n(\bm\theta_0,W_0) + D g_n(\bm\theta_0,W_0)[\widehat W_n-W_0] \big\} \xrightarrow[]{d} N\big(0, I(\bm\theta_0,W_0)\big).
\]
Thus, by Slutsky's theorem, 
\[-J_n(\bm\theta_0,W_0)^{-1} \sqrt{n} \big\{ g_n(\bm\theta_0,W_0) + D g_n(\bm\theta_0,W_0)[\widehat W_n-W_0] \big\} \xrightarrow[]{d}  N\big(0, \Sigma\big)\]
where
\[
\Sigma = J(\bm\theta_0,W_0)^{-1}\, I(\bm\theta_0,W_0)\, (J(\bm\theta_0,W_0)^{-1})',
\]
which, because of $I(\bm\theta_0,W_0) = -J(\bm\theta_0,W_0)$,  simplifies to 
\[
\Sigma = I(\bm\theta_0,W_0)^{-1}.
\]
In Lemma \ref{lem:nonsingular} we proved that, under condition \ref{c:nonsingular}, the Jacobian $J(\bm\theta_0,W_0)$ is invertible, which guarantees that the covariance matrix $\Sigma$ is well defined.

%  If you do use BiBTeX, please use the .bst file that comes with 
%  the distribution. 

\printbibliography[heading=subbibliography, filter=appendixOnlyFilter]

% =========================================================================
\clearpage
\addappendix{Appendix B}

\subsection{Metrics and MCSE formulae}
\label{AppB1}

Relative bias (RelBias; \%), empirical standard error (EmpSE), root mean squared error (RMSE), and the coverage probability of 95\% confidence intervals (CP;\%) are computed over $R=500$ replications. Relative bias is reported for all parameters, except for the intercept $\alpha$ in the logistic regression scenario, for which the true value is zero and absolute bias (Bias;\%) is reported instead.

The expressions are given for the parameter $\gamma$, analogous formulas were used for $\alpha$ and $\phi$. Let $\bar{\gamma} = \frac{1}{R} \sum_{r=1}^R \hat{\gamma}^{(r)}$, then
\begin{gather*}
    \widehat{\rm RelBias}(\hat\gamma) = \frac{1}{R} \sum_{r=1}^R (\hat\gamma^{(r)} - \gamma_{0})/\gamma_{0}, \\[2ex]
    \widehat{\rm EmpSE}(\hat{\gamma}) = \sqrt{\frac{1}{R-1} \sum_{r=1}^R \left( \hat{\gamma}^{(r)} - \bar{\gamma} \right)^2}, \\[2ex]
    \widehat{\rm RMSE}(\hat{\gamma}) = \sqrt{\frac{1}{R} \sum_{r=1}^R \left( \hat{\gamma}^{(r)} - \gamma_0 \right)^2}, \\[2ex]
    \widehat{\rm CP}(\hat{\gamma}) = \frac{1}{R} \sum_{r=1}^{R} \mathbb{1}\{
{\rm LB}_{\hat\gamma}^{(r)} \leq \gamma_{0} \leq {\rm UB}_{\hat\gamma}^{(r)} \},
\end{gather*}
where ${\rm LB}_{\hat\gamma}^{(r)}$ and ${\rm UB}_{\hat\gamma}^{(r)}$ denote the lower and upper bounds of the 95\% confidence interval in the $r$th replication.

Monte Carlo standard errors of the estimated performance measures are computed as follows:
\begin{gather*}
    \text{MCSE}(\widehat{\text{RelBias}}(\hat\gamma)) = 
\sqrt{\frac{1}{R(R-1)} \sum_{r=1}^R 
\big( \text{rb}^{(r)} - \widehat{\text{RelBias}}(\hat\gamma) \big)^2},    \\[2ex]
\text{MCSE}(\widehat{\text{EmpSE}}(\hat\gamma)) = 
\widehat{\text{EmpSE}}(\hat\gamma) / \sqrt{2(R-1)},                       \\[2ex]
\text{MCSE}(\widehat{\text{RMSE}}(\hat\gamma)) = 
\frac{1}{2\, \widehat{\text{RMSE}}(\hat\gamma)}
\sqrt{\frac{1}{R(R-1)} \sum_{r=1}^R \Big( 
\text{mse}^{(r)} - \overline{\text{mse}} \Big)^2} ,                        \\[2ex]
\text{MCSE}(\widehat{\text{CP}}(\hat\gamma)) = \sqrt{ \frac{1}{R} \widehat{\text{CP}}(\hat\gamma) \big( 1 - \widehat{\text{CP}}(\hat\gamma) \big) },
\end{gather*}
where $\text{rb}^{(r)} = (\hat{\gamma}^{(r)} - \gamma_{0})/\gamma_{0}$ and $\text{mse}^{(r)} = (\hat{\gamma}^{(r)} - \gamma_0)^2$.

\subsection{Performance of the GELc estimator}
\label{AppB2}

\subsubsection{Figures}
\label{AppB21}

Figures \ref{fig:simstudy:1} to \ref{fig:simstudy:5} illustrate the bias with respect to sample size and interval width for the estimated parameters. Relative bias (RelBias) is reported for all parameters except those with true values fixed at zero, for which absolute bias is reported instead. In Figures \ref{fig:simstudy:6} and \ref{fig:simstudy:7}, empirical standard error (EmpSE) is presented only for $\hat{\gamma}$, as the trends are identical for the remaining parameters.

\textbf{Gamma regression}
\begin{figure}[!ht] 
\begin{center}
\centerline{\includegraphics[width=.65\linewidth]{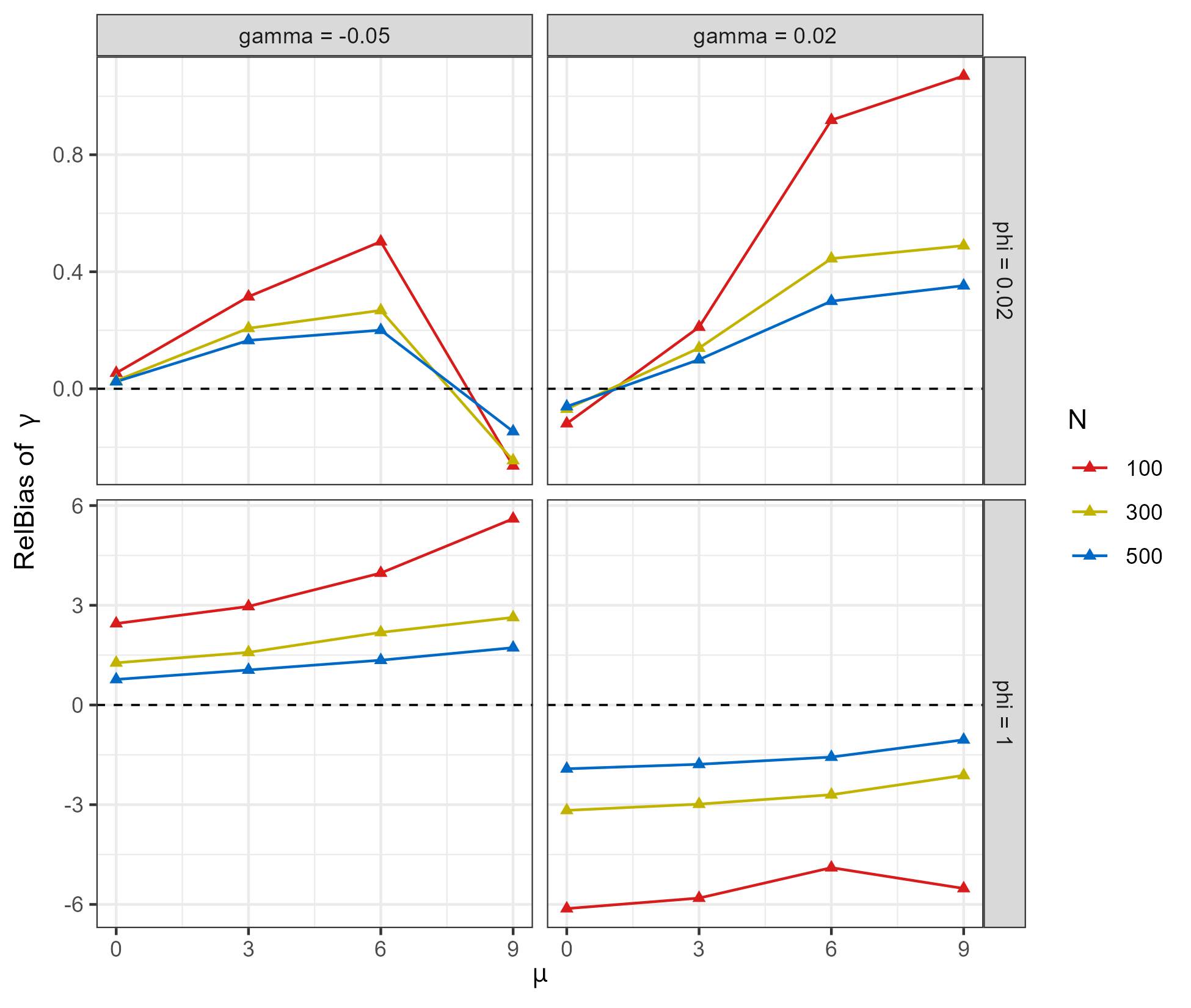}}
\vspace{-3em}
\end{center}
\caption{
Relative bias of $\gamma$ for increasing mean width of observed intervals (x-axis) and sample size, in the Gamma regression model for varying $\gamma_0$ and $\phi_0$.} \label{fig:simstudy:1}
\end{figure}

\textbf{Logistic regression}
\begin{figure}[!ht]
\begin{center}
\centerline{\includegraphics[width=.65\linewidth]{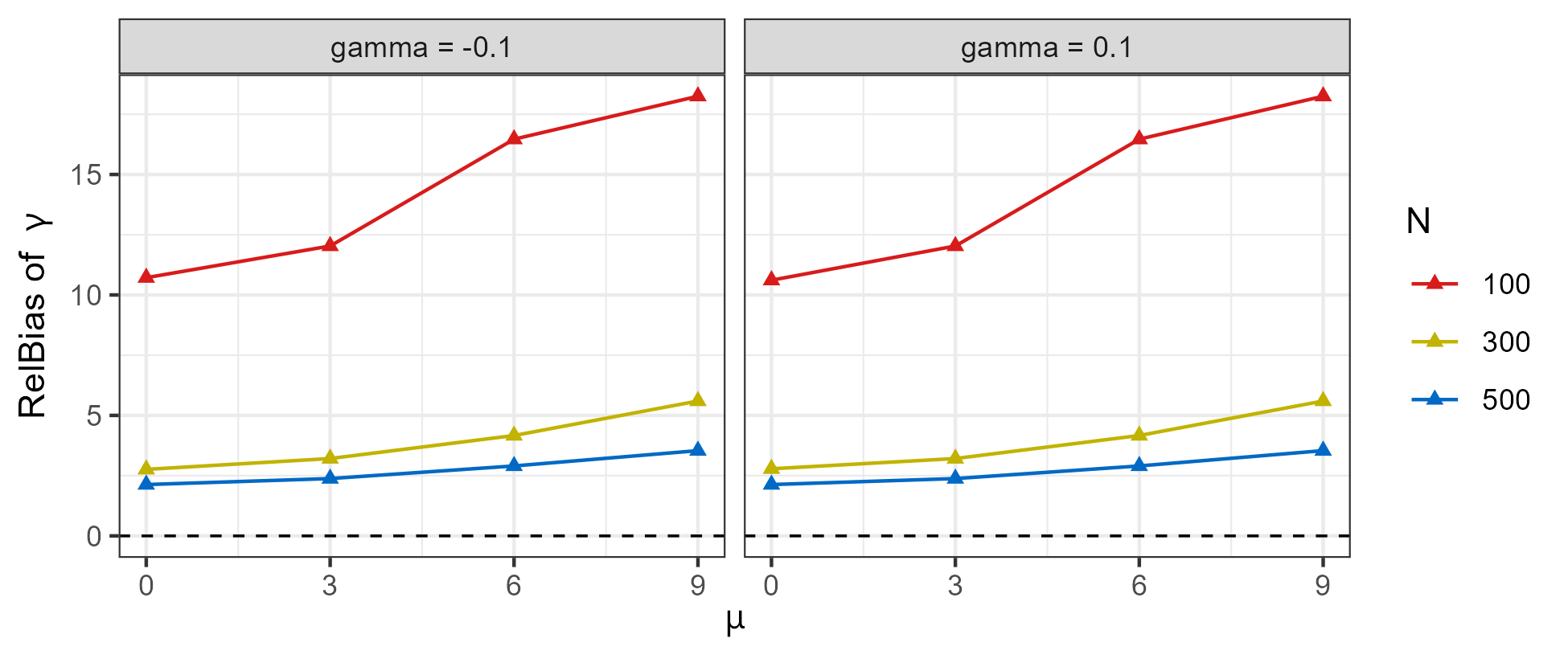}}
\vspace{-3em}
\end{center}
\caption{Relative bias of $\gamma$ for increasing mean width of observed intervals (x-axis) and sample size, in the Logistic regression model for varying $\gamma_0$.} \label{fig:simstudy:2}
\end{figure}

\clearpage

\textbf{Gamma regression}
\begin{figure}[!ht]
\begin{center}
\centerline{\includegraphics[width=.65\linewidth]{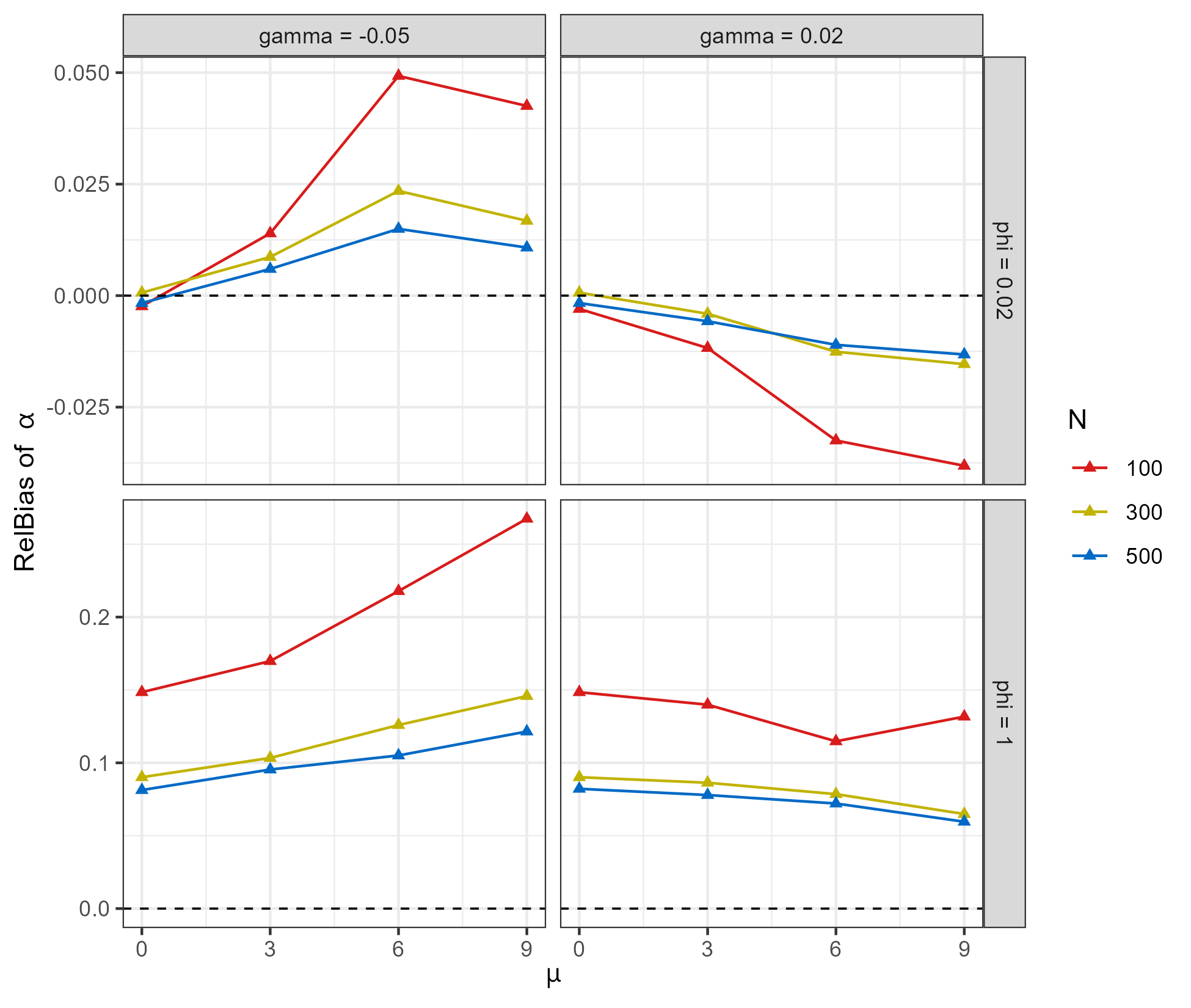}}
\end{center}
\caption{Relative bias of $\alpha$ for increasing mean width of observed intervals (x-axis) and sample size, in the Gamma regression model for $\alpha_0=10$, and varying $\gamma_0$ and $\phi_0$.} \label{fig:simstudy:3}
\end{figure}

\textbf{Logistic regression}
\begin{figure}[!ht]
\begin{center}
\centerline{\includegraphics[width=.65\linewidth]{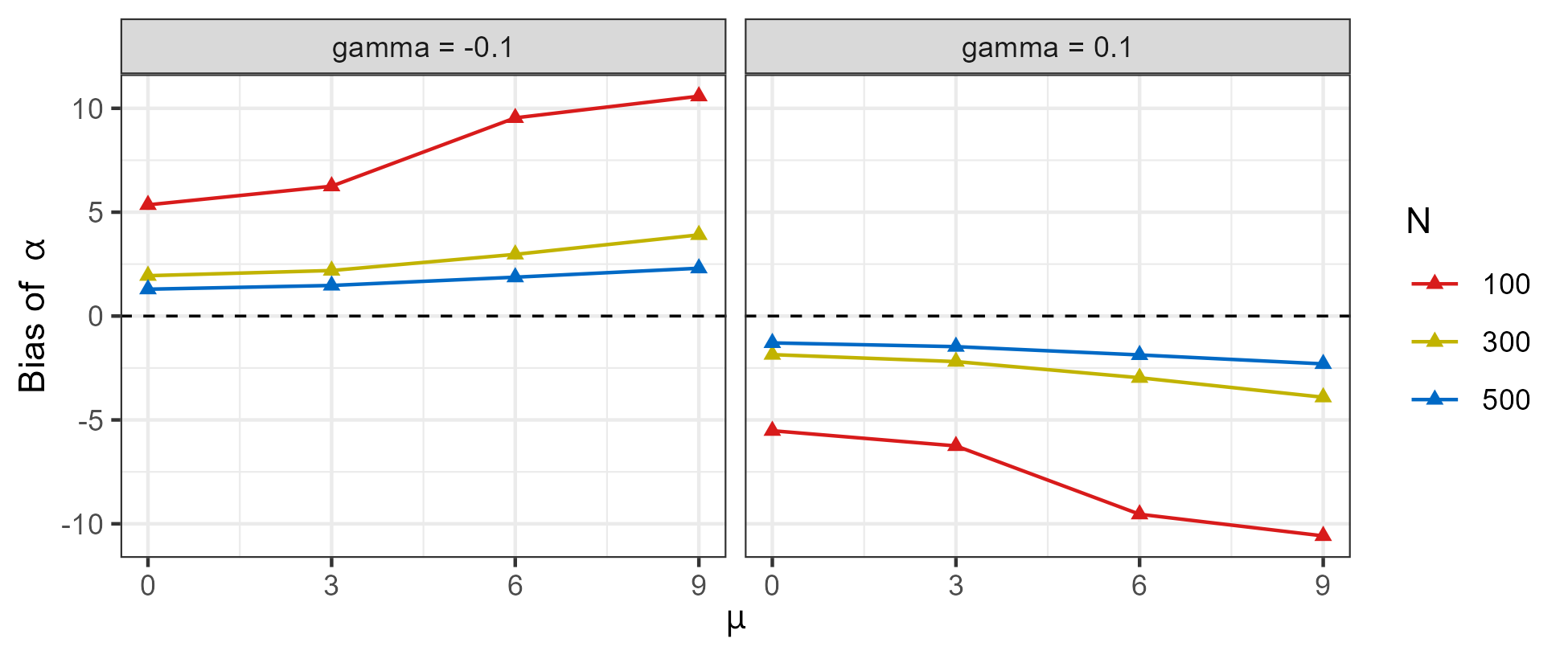}}
\end{center}
\caption{Absolute bias of $\alpha$ for increasing mean width of observed intervals (x-axis) and sample size, in the Logistic regression model for $\alpha_0=0$ and varying $\gamma_0$.} \label{fig:simstudy:4}
\end{figure}

\clearpage

\textbf{Gamma regression}
\begin{figure}[!ht]
\begin{center}
\centerline{\includegraphics[width=.7\linewidth]{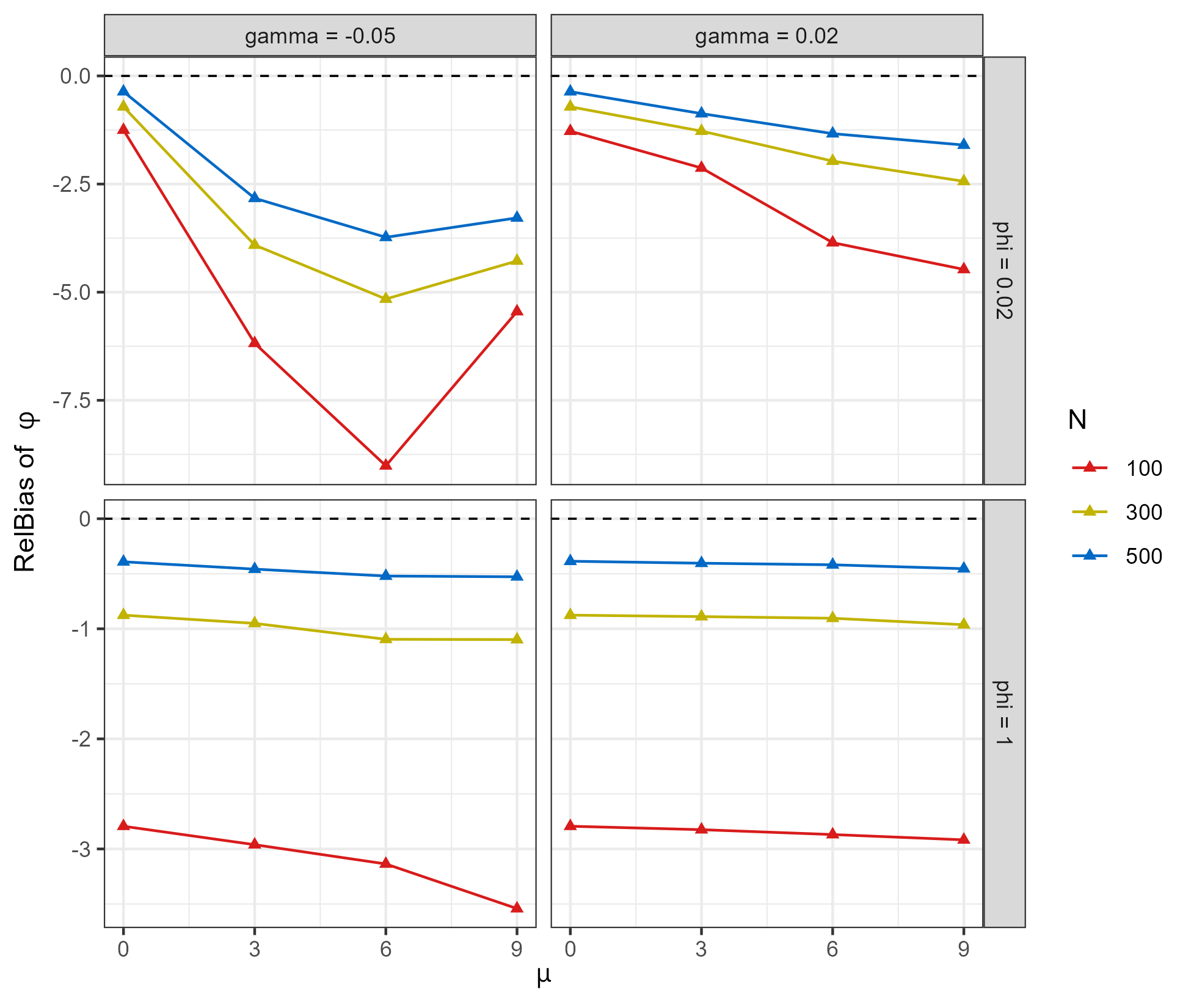}}
\end{center}
\caption{Relative bias of $\phi$ for increasing mean width of observed intervals (x-axis) and sample size, in the Gamma regression model for varying $\gamma_0$ and $\phi_0$.} \label{fig:simstudy:5}
\end{figure}

\clearpage

\textbf{Gamma regression}
\begin{figure}[!ht]
\begin{center}
\centerline{\includegraphics[width=.65\linewidth]{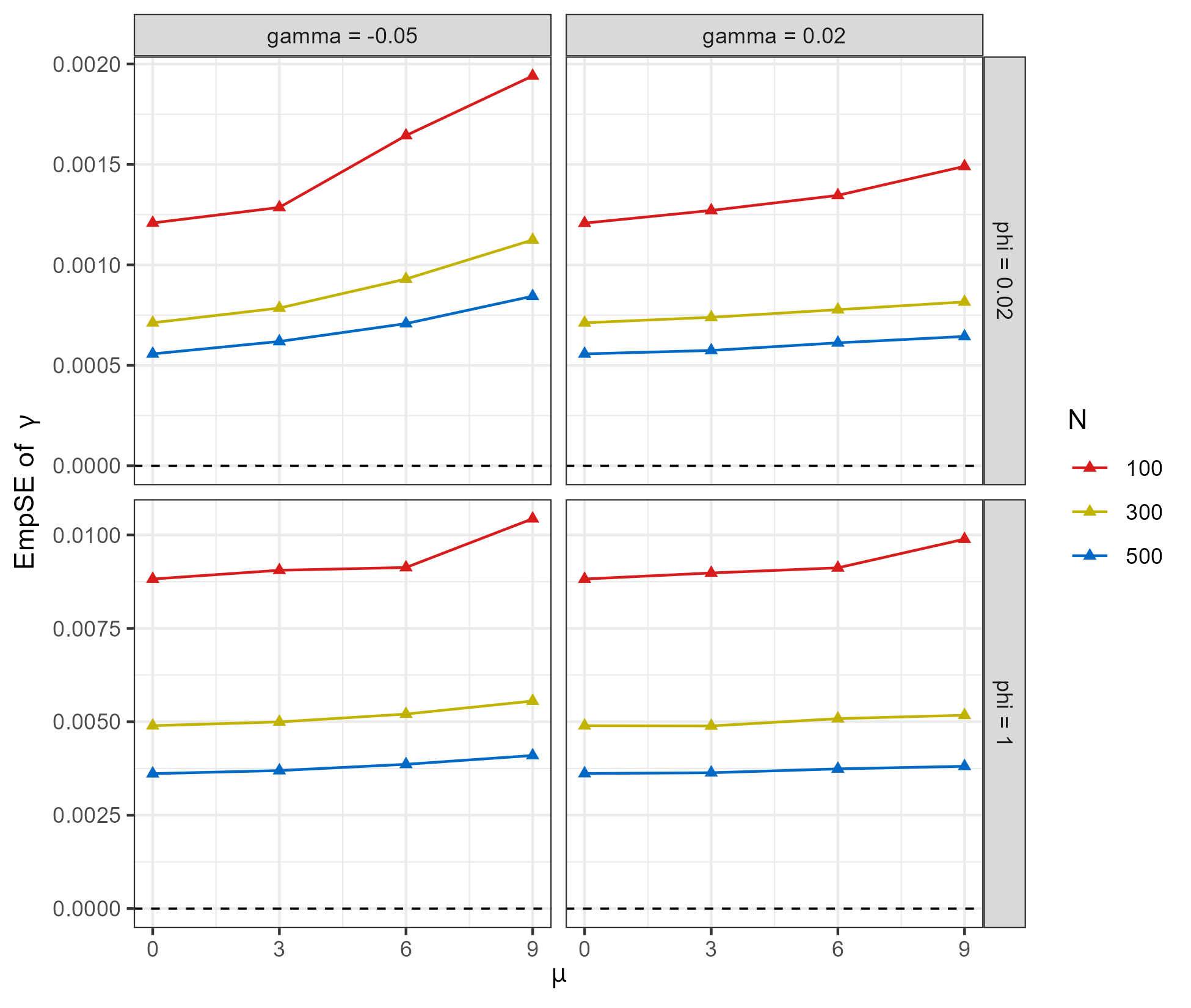}}
\end{center}
\caption{Empirical standard error of $\gamma$ for increasing mean width of observed intervals (x-axis) and sample size, in the Gamma regression model for varying $\gamma_0$ and $\phi_0$.} \label{fig:simstudy:6}
\end{figure}

\textbf{Logistic regression}
\begin{figure}[!ht]
\begin{center}
\centerline{\includegraphics[width=.65\linewidth]{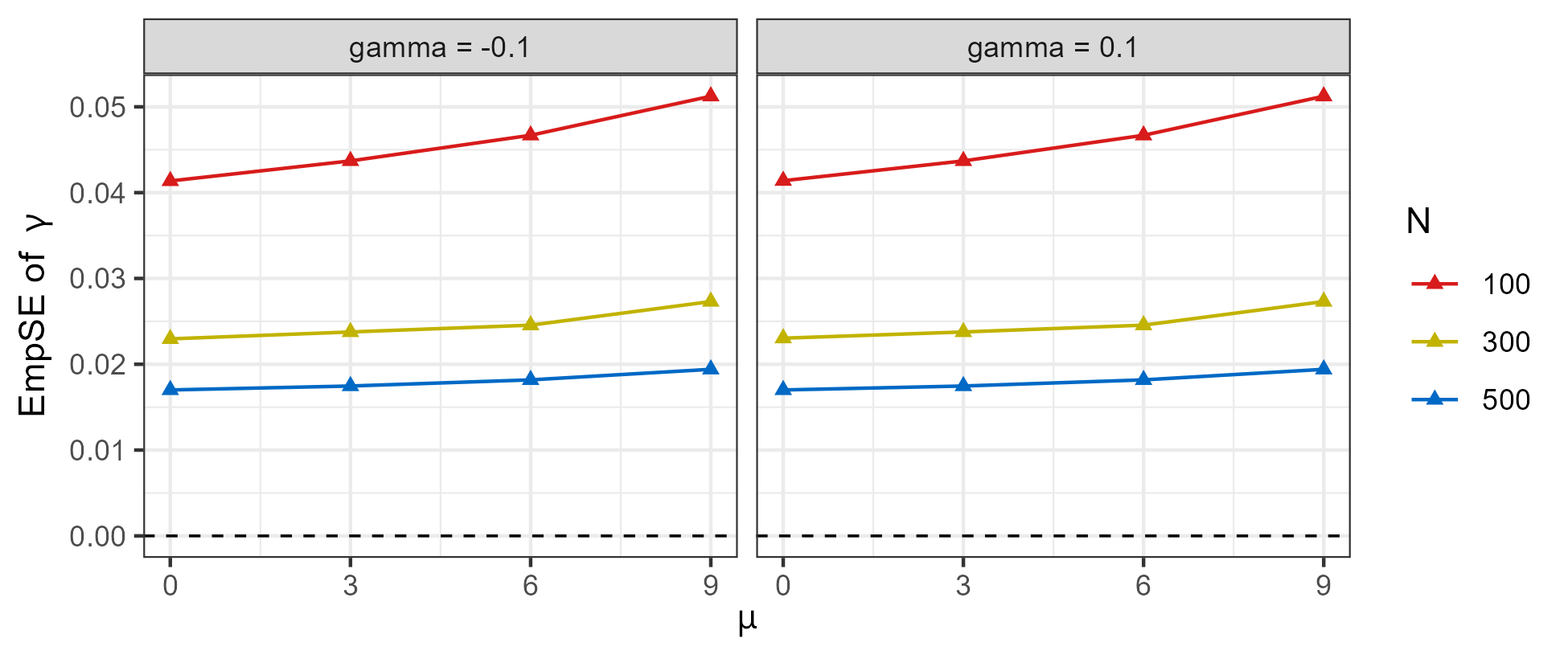}}
\end{center}
\caption{Empirical standard error of $\gamma$ for increasing mean width of observed intervals (x-axis) and sample size, in the Logistic regression model for varying $\gamma_0$.} \label{fig:simstudy:7}
\end{figure}

\clearpage
\subsubsection{Tables of performance metrics}
\label{AppB22}

\begin{table}[!ht]
\caption{Gamma regression with $\alpha=10$, $\gamma = -0.05$ and $\phi = 0.02$}
\centering
\resizebox{\ifdim\width>\linewidth\linewidth\else\width\fi}{!}{
\begin{tabular}[t]{ccccccccccc}
\toprule
\multicolumn{2}{c}{ } & \multicolumn{3}{c}{$\alpha$} & \multicolumn{3}{c}{$\gamma$} & \multicolumn{3}{c}{$\phi$} \\
\cmidrule(l{3pt}r{3pt}){3-5} \cmidrule(l{3pt}r{3pt}){6-8} \cmidrule(l{3pt}r{3pt}){9-11}
N & $\mu$ & RelBias (\%) & EmpSE & RMSE & RelBias (\%) & EmpSE & RMSE & RelBias (\%) & EmpSE & RMSE\\
\midrule
& 0 & -0.002 & 0.021 & 0.021 & 0.054 & 0.001 & 0.001 & -1.251 & 0.003 & 0.003\\

 & 3 & 0.014 & 0.023 & 0.023 & 0.315 & 0.001 & 0.001 & -6.181 & 0.003 & 0.003\\

 & 6 & 0.049 & 0.026 & 0.027 & 0.503 & 0.002 & 0.002 & -9.014 & 0.004 & 0.004\\

\multirow{-4}{*}[1.5\dimexpr\aboverulesep+\belowrulesep+\cmidrulewidth]{\centering\arraybackslash 100} & 9 & 0.043 & 0.030 & 0.030 & -0.263 & 0.002 & 0.002 & -5.447 & 0.005 & 0.005\\
\cmidrule{1-11}
 & 0 & 0.001 & 0.012 & 0.012 & 0.028 & 0.001 & 0.001 & -0.713 & 0.002 & 0.002\\

 & 3 & 0.009 & 0.013 & 0.013 & 0.207 & 0.001 & 0.001 & -3.910 & 0.002 & 0.002\\

 & 6 & 0.023 & 0.015 & 0.015 & 0.268 & 0.001 & 0.001 & -5.158 & 0.002 & 0.003\\

\multirow{-4}{*}[1.5\dimexpr\aboverulesep+\belowrulesep+\cmidrulewidth]{\centering\arraybackslash 300} & 9 & 0.017 & 0.017 & 0.017 & -0.245 & 0.001 & 0.001 & -4.277 & 0.003 & 0.003\\
\cmidrule{1-11}
 & 0 & -0.002 & 0.009 & 0.009 & 0.024 & 0.001 & 0.001 & -0.363 & 0.001 & 0.001\\

 & 3 & 0.006 & 0.010 & 0.010 & 0.165 & 0.001 & 0.001 & -2.828 & 0.002 & 0.002\\

 & 6 & 0.015 & 0.011 & 0.011 & 0.200 & 0.001 & 0.001 & -3.731 & 0.002 & 0.002\\

\multirow{-4}{*}[1.5\dimexpr\aboverulesep+\belowrulesep+\cmidrulewidth]{\centering\arraybackslash 500} & 9 & 0.011 & 0.013 & 0.013 & -0.146 & 0.001 & 0.001 & -3.282 & 0.002 & 0.002\\
\bottomrule
\end{tabular}}
\end{table}

\begin{table}[!ht]
\caption{Gamma regression with $\alpha=10$, $\gamma = 0.02$ and $\phi = 0.02$}
\centering
\resizebox{\ifdim\width>\linewidth\linewidth\else\width\fi}{!}{
\begin{tabular}[t]{ccccccccccc}
\toprule
\multicolumn{2}{c}{ } & \multicolumn{3}{c}{$\alpha$} & \multicolumn{3}{c}{$\gamma$} & \multicolumn{3}{c}{$\phi$} \\
\cmidrule(l{3pt}r{3pt}){3-5} \cmidrule(l{3pt}r{3pt}){6-8} \cmidrule(l{3pt}r{3pt}){9-11}
N & $\mu$ & RelBias (\%) & EmpSE & RMSE & RelBias (\%) & EmpSE & RMSE & RelBias (\%) & EmpSE & RMSE\\
\midrule
& 0 & -0.003 & 0.021 & 0.021 & -0.119 & 0.001 & 0.001 & -1.277 & 0.003 & 0.003\\

 & 3 & -0.012 & 0.021 & 0.022 & 0.211 & 0.001 & 0.001 & -2.128 & 0.003 & 0.003\\

 & 6 & -0.032 & 0.022 & 0.023 & 0.918 & 0.001 & 0.001 & -3.855 & 0.003 & 0.003\\

\multirow{-4}{*}[1.5\dimexpr\aboverulesep+\belowrulesep+\cmidrulewidth]{\centering\arraybackslash 100} & 9 & -0.038 & 0.024 & 0.024 & 1.070 & 0.001 & 0.002 & -4.473 & 0.003 & 0.003\\
\cmidrule{1-11}
 & 0 & 0.001 & 0.012 & 0.012 & -0.069 & 0.001 & 0.001 & -0.713 & 0.002 & 0.002\\

 & 3 & -0.004 & 0.013 & 0.013 & 0.139 & 0.001 & 0.001 & -1.276 & 0.002 & 0.002\\

 & 6 & -0.013 & 0.013 & 0.013 & 0.445 & 0.001 & 0.001 & -1.968 & 0.002 & 0.002\\

\multirow{-4}{*}[1.5\dimexpr\aboverulesep+\belowrulesep+\cmidrulewidth]{\centering\arraybackslash 300} & 9 & -0.015 & 0.013 & 0.013 & 0.489 & 0.001 & 0.001 & -2.437 & 0.002 & 0.002\\
\cmidrule{1-11}
 & 0 & -0.002 & 0.009 & 0.009 & -0.061 & 0.001 & 0.001 & -0.364 & 0.001 & 0.001\\

 & 3 & -0.006 & 0.010 & 0.010 & 0.100 & 0.001 & 0.001 & -0.871 & 0.001 & 0.001\\

 & 6 & -0.011 & 0.010 & 0.010 & 0.300 & 0.001 & 0.001 & -1.334 & 0.001 & 0.001\\

\multirow{-4}{*}[1.5\dimexpr\aboverulesep+\belowrulesep+\cmidrulewidth]{\centering\arraybackslash 500} & 9 & -0.013 & 0.010 & 0.010 & 0.352 & 0.001 & 0.001 & -1.598 & 0.001 & 0.001\\
\bottomrule
\end{tabular}}
\end{table}

\begin{table}[!ht]
\centering
\caption{Gamma regression with $\alpha=10$, $\gamma = -0.05$ and $\phi = 1$}
\centering
\resizebox{\ifdim\width>\linewidth\linewidth\else\width\fi}{!}{
\begin{tabular}[t]{ccccccccccc}
\toprule
\multicolumn{2}{c}{ } & \multicolumn{3}{c}{$\alpha$} & \multicolumn{3}{c}{$\gamma$} & \multicolumn{3}{c}{$\phi$} \\
\cmidrule(l{3pt}r{3pt}){3-5} \cmidrule(l{3pt}r{3pt}){6-8} \cmidrule(l{3pt}r{3pt}){9-11}
N & $\mu$ & RelBias (\%) & EmpSE & RMSE & RelBias (\%) & EmpSE & RMSE & RelBias (\%) & EmpSE & RMSE\\
\midrule
& 0 & 0.148 & 0.142 & 0.143 & 2.450 & 0.009 & 0.009 & -2.794 & 0.118 & 0.121\\

 & 3 & 0.170 & 0.144 & 0.144 & 2.965 & 0.009 & 0.009 & -2.961 & 0.119 & 0.122\\

 & 6 & 0.218 & 0.144 & 0.146 & 3.970 & 0.009 & 0.009 & -3.136 & 0.120 & 0.124\\

\multirow{-4}{*}[1.5\dimexpr\aboverulesep+\belowrulesep+\cmidrulewidth]{\centering\arraybackslash 100} & 9 & 0.267 & 0.148 & 0.150 & 5.604 & 0.010 & 0.011 & -3.542 & 0.120 & 0.125\\
\cmidrule{1-11}
 & 0 & 0.090 & 0.081 & 0.081 & 1.268 & 0.005 & 0.005 & -0.876 & 0.072 & 0.073\\

 & 3 & 0.103 & 0.082 & 0.082 & 1.583 & 0.005 & 0.005 & -0.951 & 0.073 & 0.073\\

 & 6 & 0.126 & 0.081 & 0.082 & 2.184 & 0.005 & 0.005 & -1.095 & 0.074 & 0.074\\

\multirow{-4}{*}[1.5\dimexpr\aboverulesep+\belowrulesep+\cmidrulewidth]{\centering\arraybackslash 300} & 9 & 0.146 & 0.085 & 0.086 & 2.635 & 0.006 & 0.006 & -1.098 & 0.073 & 0.074\\
\cmidrule{1-11}
 & 0 & 0.081 & 0.063 & 0.063 & 0.769 & 0.004 & 0.004 & -0.391 & 0.056 & 0.056\\

 & 3 & 0.095 & 0.063 & 0.064 & 1.054 & 0.004 & 0.004 & -0.458 & 0.056 & 0.056\\

 & 6 & 0.105 & 0.063 & 0.064 & 1.348 & 0.004 & 0.004 & -0.521 & 0.057 & 0.057\\

\multirow{-4}{*}[1.5\dimexpr\aboverulesep+\belowrulesep+\cmidrulewidth]{\centering\arraybackslash 500} & 9 & 0.121 & 0.064 & 0.065 & 1.724 & 0.004 & 0.004 & -0.527 & 0.057 & 0.058\\
\bottomrule
\end{tabular}}
\end{table}

\begin{table}[!ht]
\centering
\caption{Gamma regression with $\alpha=10$, $\gamma = 0.02$ and $\phi = 1$}
\centering
\resizebox{\ifdim\width>\linewidth\linewidth\else\width\fi}{!}{
\begin{tabular}[t]{ccccccccccc}
\toprule
\multicolumn{2}{c}{ } & \multicolumn{3}{c}{$\alpha$} & \multicolumn{3}{c}{$\gamma$} & \multicolumn{3}{c}{$\phi$} \\
\cmidrule(l{3pt}r{3pt}){3-5} \cmidrule(l{3pt}r{3pt}){6-8} \cmidrule(l{3pt}r{3pt}){9-11}
N & $\mu$ & RelBias (\%) & EmpSE & RMSE & RelBias (\%) & EmpSE & RMSE & RelBias (\%) & EmpSE & RMSE\\
\midrule
 & 0 & 0.148 & 0.142 & 0.143 & -6.126 & 0.009 & 0.009 & -2.794 & 0.118 & 0.121\\

 & 3 & 0.140 & 0.143 & 0.144 & -5.809 & 0.009 & 0.009 & -2.825 & 0.118 & 0.121\\

 & 6 & 0.115 & 0.145 & 0.146 & -4.894 & 0.009 & 0.009 & -2.869 & 0.118 & 0.122\\

\multirow{-4}{*}[1.5\dimexpr\aboverulesep+\belowrulesep+\cmidrulewidth]{\centering\arraybackslash 100} & 9 & 0.132 & 0.153 & 0.153 & -5.520 & 0.010 & 0.010 & -2.917 & 0.119 & 0.122\\
\cmidrule{1-11}
 & 0 & 0.090 & 0.081 & 0.081 & -3.170 & 0.005 & 0.005 & -0.876 & 0.072 & 0.073\\

 & 3 & 0.086 & 0.080 & 0.081 & -2.983 & 0.005 & 0.005 & -0.889 & 0.072 & 0.073\\

 & 6 & 0.079 & 0.083 & 0.083 & -2.700 & 0.005 & 0.005 & -0.904 & 0.073 & 0.073\\

\multirow{-4}{*}[1.5\dimexpr\aboverulesep+\belowrulesep+\cmidrulewidth]{\centering\arraybackslash 300} & 9 & 0.065 & 0.085 & 0.085 & -2.118 & 0.005 & 0.005 & -0.963 & 0.073 & 0.073\\
\cmidrule{1-11}
 & 0 & 0.082 & 0.063 & 0.063 & -1.920 & 0.004 & 0.004 & -0.386 & 0.056 & 0.056\\

 & 3 & 0.078 & 0.063 & 0.063 & -1.784 & 0.004 & 0.004 & -0.404 & 0.056 & 0.056\\

 & 6 & 0.072 & 0.064 & 0.064 & -1.567 & 0.004 & 0.004 & -0.419 & 0.056 & 0.056\\

\multirow{-4}{*}[1.5\dimexpr\aboverulesep+\belowrulesep+\cmidrulewidth]{\centering\arraybackslash 500} & 9 & 0.060 & 0.065 & 0.065 & -1.048 & 0.004 & 0.004 & -0.454 & 0.056 & 0.056\\
\bottomrule
\end{tabular}}
\end{table}

\begin{table}[!ht]
\centering
\caption{Logistic regression with $\alpha=0$ and $\gamma = -0.1$}
\centering
\resizebox{0.8\linewidth}{!}{
\begin{tabular}[t]{cccccccc}
\toprule
\multicolumn{2}{c}{ } & \multicolumn{3}{c}{$\alpha$} & \multicolumn{3}{c}{$\gamma$} \\
\cmidrule(l{3pt}r{3pt}){3-5} \cmidrule(l{3pt}r{3pt}){6-8}
N & $\mu$ & Bias (\%) & EmpSE & RMSE & RelBias (\%) & EmpSE & RMSE\\
\midrule
 & 0 & 0.054 & 0.379 & 0.382 & 10.717 & 0.041 & 0.043\\

 & 3 & 0.063 & 0.394 & 0.398 & 12.032 & 0.044 & 0.045\\

 & 6 & 0.095 & 0.399 & 0.410 & 16.465 & 0.047 & 0.049\\

\multirow{-4}{*}[1.5\dimexpr\aboverulesep+\belowrulesep+\cmidrulewidth]{\centering\arraybackslash 100} & 9 & 0.106 & 0.434 & 0.446 & 18.248 & 0.051 & 0.054\\
\cmidrule{1-8}
 & 0 & 0.019 & 0.220 & 0.221 & 2.761 & 0.023 & 0.023\\

 & 3 & 0.022 & 0.226 & 0.227 & 3.210 & 0.024 & 0.024\\

 & 6 & 0.030 & 0.228 & 0.230 & 4.166 & 0.025 & 0.025\\

\multirow{-4}{*}[1.5\dimexpr\aboverulesep+\belowrulesep+\cmidrulewidth]{\centering\arraybackslash 300} & 9 & 0.039 & 0.245 & 0.248 & 5.598 & 0.027 & 0.028\\
\cmidrule{1-8}
 & 0 & 0.013 & 0.170 & 0.170 & 2.130 & 0.017 & 0.017\\

 & 3 & 0.015 & 0.173 & 0.173 & 2.379 & 0.017 & 0.018\\

 & 6 & 0.019 & 0.176 & 0.177 & 2.903 & 0.018 & 0.018\\

\multirow{-4}{*}[1.5\dimexpr\aboverulesep+\belowrulesep+\cmidrulewidth]{\centering\arraybackslash 500} & 9 & 0.023 & 0.184 & 0.185 & 3.539 & 0.019 & 0.020\\
\bottomrule
\end{tabular}}
\end{table}

\begin{table}[!ht]
\centering
\caption{Logistic regression with $\alpha=0$ and $\gamma = 0.1$}
\centering
\resizebox{0.8\linewidth}{!}{
\begin{tabular}[t]{cccccccc}
\toprule
\multicolumn{2}{c}{ } & \multicolumn{3}{c}{$\alpha$} & \multicolumn{3}{c}{$\gamma$} \\
\cmidrule(l{3pt}r{3pt}){3-5} \cmidrule(l{3pt}r{3pt}){6-8}
N & $\mu$ & Bias (\%) & EmpSE & RMSE & RelBias (\%) & EmpSE & RMSE\\
\midrule
 & 0 & -0.055 & 0.378 & 0.381 & 10.612 & 0.041 & 0.043\\

 & 3 & -0.063 & 0.394 & 0.398 & 12.032 & 0.044 & 0.045\\

 & 6 & -0.095 & 0.399 & 0.410 & 16.465 & 0.047 & 0.049\\

\multirow{-4}{*}[1.5\dimexpr\aboverulesep+\belowrulesep+\cmidrulewidth]{\centering\arraybackslash 100} & 9 & -0.106 & 0.434 & 0.446 & 18.248 & 0.051 & 0.054\\
\cmidrule{1-8}
 & 0 & -0.019 & 0.220 & 0.221 & 2.785 & 0.023 & 0.023\\

 & 3 & -0.022 & 0.226 & 0.227 & 3.210 & 0.024 & 0.024\\

 & 6 & -0.030 & 0.228 & 0.230 & 4.166 & 0.025 & 0.025\\

\multirow{-4}{*}[1.5\dimexpr\aboverulesep+\belowrulesep+\cmidrulewidth]{\centering\arraybackslash 300} & 9 & -0.039 & 0.245 & 0.248 & 5.598 & 0.027 & 0.028\\
\cmidrule{1-8}
 & 0 & -0.013 & 0.170 & 0.170 & 2.130 & 0.017 & 0.017\\

 & 3 & -0.015 & 0.173 & 0.173 & 2.379 & 0.017 & 0.018\\

 & 6 & -0.019 & 0.176 & 0.177 & 2.903 & 0.018 & 0.018\\

\multirow{-4}{*}[1.5\dimexpr\aboverulesep+\belowrulesep+\cmidrulewidth]{\centering\arraybackslash 500} & 9 & -0.023 & 0.184 & 0.185 & 3.539 & 0.019 & 0.020\\
\bottomrule
\end{tabular}}
\end{table}

\clearpage
\subsubsection{Tables of MCSE for performance metrics}
\label{AppB23}

\begin{table}[!ht]
\centering
\caption{Gamma regression with $\alpha=10$, $\gamma = -0.05$ and $\phi = 0.02$}
\centering
\resizebox{\ifdim\width>\linewidth\linewidth\else\width\fi}{!}{
\begin{tabular}[t]{ccccccccccc}
\toprule
\multicolumn{2}{c}{ } & \multicolumn{3}{c}{$\alpha$} & \multicolumn{3}{c}{$\gamma$} & \multicolumn{3}{c}{$\phi$} \\
\cmidrule(l{3pt}r{3pt}){3-5} \cmidrule(l{3pt}r{3pt}){6-8} \cmidrule(l{3pt}r{3pt}){9-11}
N & $\mu$ & RelBias(\%) & EmpSE & RMSE & RelBias(\%) & EmpSE & RMSE & RelBias(\%) & EmpSE & RMSE\\
\midrule
 & 0 & 0.009 & 0.0007 & 0.0007 & 0.108 & $<0.0001$ & $<0.0001$ & 0.611 & $<0.0001$ & $<0.0001$\\

 & 3 & 0.010 & 0.0007 & 0.0007 & 0.115 & $<0.0001$ & $<0.0001$ & 0.706 & $<0.0001$ & $<0.0001$\\

 & 6 & 0.012 & 0.0008 & 0.0009 & 0.147 & $<0.0001$ & $<0.0001$ & 0.905 & 0.0001 & 0.0001\\

\multirow{-4}{*}[1.5\dimexpr\aboverulesep+\belowrulesep+\cmidrulewidth]{\centering\arraybackslash 100} & 9 & 0.013 & 0.0010 & 0.0010 & 0.174 & $<0.0001$ & $<0.0001$ & 1.017 & 0.0001 & 0.0001\\
\cmidrule{1-11}
 & 0 & 0.005 & 0.0004 & 0.0004 & 0.064 & $<0.0001$ & $<0.0001$ & 0.365 & $<0.0001$ & $<0.0001$\\

 & 3 & 0.006 & 0.0004 & 0.0004 & 0.070 & $<0.0001$ & $<0.0001$ & 0.440 & $<0.0001$ & $<0.0001$\\

 & 6 & 0.007 & 0.0005 & 0.0005 & 0.083 & $<0.0001$ & $<0.0001$ & 0.535 & $<0.0001$ & $<0.0001$\\

\multirow{-4}{*}[1.5\dimexpr\aboverulesep+\belowrulesep+\cmidrulewidth]{\centering\arraybackslash 300} & 9 & 0.008 & 0.0005 & 0.0005 & 0.101 & $<0.0001$ & $<0.0001$ & 0.578 & $<0.0001$ & $<0.0001$\\
\cmidrule{1-11}
 & 0 & 0.004 & 0.0003 & 0.0003 & 0.050 & $<0.0001$ & $<0.0001$ & 0.281 & $<0.0001$ & $<0.0001$\\

 & 3 & 0.005 & 0.0003 & 0.0003 & 0.055 & $<0.0001$ & $<0.0001$ & 0.338 & $<0.0001$ & $<0.0001$\\

 & 6 & 0.005 & 0.0004 & 0.0003 & 0.063 & $<0.0001$ & $<0.0001$ & 0.395 & $<0.0001$ & $<0.0001$\\

\multirow{-4}{*}[1.5\dimexpr\aboverulesep+\belowrulesep+\cmidrulewidth]{\centering\arraybackslash 500} & 9 & 0.006 & 0.0004 & 0.0004 & 0.076 & $<0.0001$ & $<0.0001$ & 0.457 & $<0.0001$ & $<0.0001$\\
\bottomrule
\end{tabular}}
\end{table}

\begin{table}[!ht]
\centering
\caption{Gamma regression with $\alpha=10$, $\gamma = 0.02$ and $\phi = 0.02$}
\centering
\resizebox{\ifdim\width>\linewidth\linewidth\else\width\fi}{!}{
\begin{tabular}[t]{ccccccccccc}
\toprule
\multicolumn{2}{c}{ } & \multicolumn{3}{c}{$\alpha$} & \multicolumn{3}{c}{$\gamma$} & \multicolumn{3}{c}{$\phi$} \\
\cmidrule(l{3pt}r{3pt}){3-5} \cmidrule(l{3pt}r{3pt}){6-8} \cmidrule(l{3pt}r{3pt}){9-11}
N & $\mu$ & RelBias(\%) & EmpSE & RMSE & RelBias(\%) & EmpSE & RMSE & RelBias(\%) & EmpSE & RMSE\\
\midrule
 & 0 & 0.009 & 0.0007 & 0.0007 & 0.270 & $<0.0001$ & $<0.0001$ & 0.611 & $<0.0001$ & $<0.0001$\\

 & 3 & 0.010 & 0.0007 & 0.0007 & 0.284 & $<0.0001$ & $<0.0001$ & 0.648 & $<0.0001$ & $<0.0001$\\

 & 6 & 0.010 & 0.0007 & 0.0008 & 0.301 & $<0.0001$ & $<0.0001$ & 0.668 & $<0.0001$ & $<0.0001$\\

\multirow{-4}{*}[1.5\dimexpr\aboverulesep+\belowrulesep+\cmidrulewidth]{\centering\arraybackslash 100} & 9 & 0.011 & 0.0008 & 0.0008 & 0.333 & $<0.0001$ & $<0.0001$ & 0.726 & 0.0001 & 0.0001\\
\cmidrule{1-11}
 & 0 & 0.005 & 0.0004 & 0.0004 & 0.159 & $<0.0001$ & $<0.0001$ & 0.365 & $<0.0001$ & $<0.0001$\\

 & 3 & 0.006 & 0.0004 & 0.0004 & 0.165 & $<0.0001$ & $<0.0001$ & 0.375 & $<0.0001$ & $<0.0001$\\

 & 6 & 0.006 & 0.0004 & 0.0004 & 0.174 & $<0.0001$ & $<0.0001$ & 0.389 & $<0.0001$ & $<0.0001$\\

\multirow{-4}{*}[1.5\dimexpr\aboverulesep+\belowrulesep+\cmidrulewidth]{\centering\arraybackslash 300} & 9 & 0.006 & 0.0004 & 0.0004 & 0.182 & $<0.0001$ & $<0.0001$ & 0.424 & $<0.0001$ & $<0.0001$\\
\cmidrule{1-11}
 & 0 & 0.004 & 0.0003 & 0.0003 & 0.125 & $<0.0001$ & $<0.0001$ & 0.281 & $<0.0001$ & $<0.0001$\\

 & 3 & 0.004 & 0.0003 & 0.0003 & 0.128 & $<0.0001$ & $<0.0001$ & 0.288 & $<0.0001$ & $<0.0001$\\

 & 6 & 0.005 & 0.0003 & 0.0003 & 0.137 & $<0.0001$ & $<0.0001$ & 0.297 & $<0.0001$ & $<0.0001$\\

\multirow{-4}{*}[1.5\dimexpr\aboverulesep+\belowrulesep+\cmidrulewidth]{\centering\arraybackslash 500} & 9 & 0.005 & 0.0003 & 0.0003 & 0.144 & $<0.0001$ & $<0.0001$ & 0.325 & $<0.0001$ & $<0.0001$\\
\bottomrule
\end{tabular}}
\end{table}

\begin{table}[!ht]
\centering
\caption{Gamma regression with $\alpha=10$, $\gamma = -0.05$ and $\phi = 1$}
\centering
\resizebox{\ifdim\width>\linewidth\linewidth\else\width\fi}{!}{
\begin{tabular}[t]{ccccccccccc}
\toprule
\multicolumn{2}{c}{ } & \multicolumn{3}{c}{$\alpha$} & \multicolumn{3}{c}{$\gamma$} & \multicolumn{3}{c}{$\phi$} \\
\cmidrule(l{3pt}r{3pt}){3-5} \cmidrule(l{3pt}r{3pt}){6-8} \cmidrule(l{3pt}r{3pt}){9-11}
N & $\mu$ & RelBias(\%) & EmpSE & RMSE & RelBias(\%) & EmpSE & RMSE & RelBias(\%) & EmpSE & RMSE\\
\midrule
 & 0 & 0.064 & 0.0045 & 0.0050 & 0.789 & 0.0003 & 0.0003 & 0.528 & 0.0037 & 0.0034\\

 & 3 & 0.064 & 0.0045 & 0.0052 & 0.810 & 0.0003 & 0.0003 & 0.530 & 0.0038 & 0.0035\\

 & 6 & 0.064 & 0.0046 & 0.0050 & 0.817 & 0.0003 & 0.0003 & 0.537 & 0.0038 & 0.0036\\

\multirow{-4}{*}[1.5\dimexpr\aboverulesep+\belowrulesep+\cmidrulewidth]{\centering\arraybackslash 100} & 9 & 0.066 & 0.0047 & 0.0053 & 0.933 & 0.0003 & 0.0004 & 0.538 & 0.0038 & 0.0035\\
\cmidrule{1-11}
 & 0 & 0.036 & 0.0026 & 0.0025 & 0.438 & 0.0002 & 0.0002 & 0.324 & 0.0023 & 0.0020\\

 & 3 & 0.037 & 0.0026 & 0.0025 & 0.447 & 0.0002 & 0.0002 & 0.325 & 0.0023 & 0.0020\\

 & 6 & 0.036 & 0.0026 & 0.0025 & 0.466 & 0.0002 & 0.0002 & 0.329 & 0.0023 & 0.0020\\

\multirow{-4}{*}[1.5\dimexpr\aboverulesep+\belowrulesep+\cmidrulewidth]{\centering\arraybackslash 300} & 9 & 0.038 & 0.0027 & 0.0028 & 0.497 & 0.0002 & 0.0002 & 0.328 & 0.0023 & 0.0021\\
\cmidrule{1-11}
 & 0 & 0.028 & 0.0020 & 0.0020 & 0.323 & 0.0001 & 0.0001 & 0.251 & 0.0018 & 0.0018\\

 & 3 & 0.028 & 0.0020 & 0.0020 & 0.330 & 0.0001 & 0.0001 & 0.251 & 0.0018 & 0.0018\\

 & 6 & 0.028 & 0.0020 & 0.0020 & 0.345 & 0.0001 & 0.0001 & 0.256 & 0.0018 & 0.0018\\

\multirow{-4}{*}[1.5\dimexpr\aboverulesep+\belowrulesep+\cmidrulewidth]{\centering\arraybackslash 500} & 9 & 0.029 & 0.0020 & 0.0020 & 0.366 & 0.0001 & 0.0001 & 0.257 & 0.0018 & 0.0018\\
\bottomrule
\end{tabular}}
\end{table}

\begin{table}[!ht]
\centering
\caption{Gamma regression with $\alpha=10$, $\gamma = 0.02$ and $\phi = 1$}
\centering
\resizebox{\ifdim\width>\linewidth\linewidth\else\width\fi}{!}{
\begin{tabular}[t]{ccccccccccc}
\toprule
\multicolumn{2}{c}{ } & \multicolumn{3}{c}{$\alpha$} & \multicolumn{3}{c}{$\gamma$} & \multicolumn{3}{c}{$\phi$} \\
\cmidrule(l{3pt}r{3pt}){3-5} \cmidrule(l{3pt}r{3pt}){6-8} \cmidrule(l{3pt}r{3pt}){9-11}
N & $\mu$ & RelBias(\%) & EmpSE & RMSE & RelBias(\%) & EmpSE & RMSE & RelBias(\%) & EmpSE & RMSE\\
\midrule
 & 0 & 0.064 & 0.0045 & 0.0050 & 1.972 & 0.0003 & 0.0003 & 0.528 & 0.0037 & 0.0034\\

 & 3 & 0.064 & 0.0045 & 0.0051 & 2.009 & 0.0003 & 0.0003 & 0.529 & 0.0037 & 0.0034\\

 & 6 & 0.065 & 0.0046 & 0.0050 & 2.039 & 0.0003 & 0.0003 & 0.529 & 0.0037 & 0.0034\\

\multirow{-4}{*}[1.5\dimexpr\aboverulesep+\belowrulesep+\cmidrulewidth]{\centering\arraybackslash 100} & 9 & 0.068 & 0.0048 & 0.0051 & 2.211 & 0.0003 & 0.0004 & 0.531 & 0.0038 & 0.0034\\
\cmidrule{1-11}
 & 0 & 0.036 & 0.0026 & 0.0025 & 1.095 & 0.0002 & 0.0002 & 0.324 & 0.0023 & 0.0020\\

 & 3 & 0.036 & 0.0025 & 0.0025 & 1.093 & 0.0002 & 0.0002 & 0.324 & 0.0023 & 0.0020\\

 & 6 & 0.037 & 0.0026 & 0.0025 & 1.137 & 0.0002 & 0.0002 & 0.325 & 0.0023 & 0.0020\\

\multirow{-4}{*}[1.5\dimexpr\aboverulesep+\belowrulesep+\cmidrulewidth]{\centering\arraybackslash 300} & 9 & 0.038 & 0.0027 & 0.0026 & 1.157 & 0.0002 & 0.0002 & 0.326 & 0.0023 & 0.0020\\
\cmidrule{1-11}
 & 0 & 0.028 & 0.0020 & 0.0020 & 0.809 & 0.0001 & 0.0001 & 0.252 & 0.0018 & 0.0018\\

 & 3 & 0.028 & 0.0020 & 0.0019 & 0.813 & 0.0001 & 0.0001 & 0.252 & 0.0018 & 0.0018\\

 & 6 & 0.029 & 0.0020 & 0.0020 & 0.836 & 0.0001 & 0.0001 & 0.251 & 0.0018 & 0.0018\\

\multirow{-4}{*}[1.5\dimexpr\aboverulesep+\belowrulesep+\cmidrulewidth]{\centering\arraybackslash 500} & 9 & 0.029 & 0.0021 & 0.0021 & 0.852 & 0.0001 & 0.0001 & 0.252 & 0.0018 & 0.0018\\
\bottomrule
\end{tabular}}
\end{table}

\begin{table}[!ht]
\centering
\caption{Logistic regression with $\alpha=0$ and $\gamma = -0.1$}
\centering
\resizebox{.8\linewidth}{!}{
\begin{tabular}[t]{cccccccc}
\toprule
\multicolumn{2}{c}{ } & \multicolumn{3}{c}{$\alpha$} & \multicolumn{3}{c}{$\gamma$} \\
\cmidrule(l{3pt}r{3pt}){3-5} \cmidrule(l{3pt}r{3pt}){6-8}
N & $\mu$ & Bias(\%) & EmpSE & RMSE & RelBias(\%) & EmpSE & RMSE\\
\midrule
 & 0 & 0.017 & 0.0120 & 0.0134 & 1.852 & 0.0013 & 0.0017\\

 & 3 & 0.018 & 0.0125 & 0.0142 & 1.954 & 0.0014 & 0.0018\\

 & 6 & 0.018 & 0.0126 & 0.0143 & 2.088 & 0.0015 & 0.0020\\

\multirow{-4}{*}[1.5\dimexpr\aboverulesep+\belowrulesep+\cmidrulewidth]{\centering\arraybackslash 100} & 9 & 0.019 & 0.0137 & 0.0171 & 2.291 & 0.0016 & 0.0026\\
\cmidrule{1-8}
 & 0 & 0.010 & 0.0070 & 0.0069 & 1.029 & 0.0007 & 0.0008\\

 & 3 & 0.010 & 0.0072 & 0.0073 & 1.063 & 0.0008 & 0.0009\\

 & 6 & 0.010 & 0.0072 & 0.0072 & 1.098 & 0.0008 & 0.0009\\

\multirow{-4}{*}[1.5\dimexpr\aboverulesep+\belowrulesep+\cmidrulewidth]{\centering\arraybackslash 300} & 9 & 0.011 & 0.0077 & 0.0070 & 1.221 & 0.0009 & 0.0010\\
\cmidrule{1-8}
 & 0 & 0.008 & 0.0054 & 0.0052 & 0.761 & 0.0005 & 0.0006\\

 & 3 & 0.008 & 0.0055 & 0.0054 & 0.781 & 0.0006 & 0.0006\\

 & 6 & 0.008 & 0.0056 & 0.0054 & 0.813 & 0.0006 & 0.0006\\

\multirow{-4}{*}[1.5\dimexpr\aboverulesep+\belowrulesep+\cmidrulewidth]{\centering\arraybackslash 500} & 9 & 0.008 & 0.0058 & 0.0055 & 0.868 & 0.0006 & 0.0007\\
\bottomrule
\end{tabular}}
\end{table}

\begin{table}[!ht]
\centering
\caption{Logistic regression with $\alpha=0$ and $\gamma = 0.1$}
\centering
\resizebox{.8\linewidth}{!}{
\begin{tabular}[t]{cccccccc}
\toprule
\multicolumn{2}{c}{ } & \multicolumn{3}{c}{$\alpha$} & \multicolumn{3}{c}{$\gamma$} \\
\cmidrule(l{3pt}r{3pt}){3-5} \cmidrule(l{3pt}r{3pt}){6-8}
N & $\mu$ & Bias(\%) & EmpSE & RMSE & RelBias(\%) & EmpSE & RMSE\\
\midrule
 & 0 & 0.017 & 0.0120 & 0.0134 & 1.858 & 0.0013 & 0.0017\\

 & 3 & 0.018 & 0.0125 & 0.0142 & 1.954 & 0.0014 & 0.0018\\

 & 6 & 0.018 & 0.0126 & 0.0143 & 2.088 & 0.0015 & 0.0020\\

\multirow{-4}{*}[1.5\dimexpr\aboverulesep+\belowrulesep+\cmidrulewidth]{\centering\arraybackslash 100} & 9 & 0.019 & 0.0137 & 0.0171 & 2.291 & 0.0016 & 0.0026\\
\cmidrule{1-8}
 & 0 & 0.010 & 0.0070 & 0.0069 & 1.033 & 0.0007 & 0.0008\\

 & 3 & 0.010 & 0.0072 & 0.0073 & 1.063 & 0.0008 & 0.0009\\

 & 6 & 0.010 & 0.0072 & 0.0072 & 1.098 & 0.0008 & 0.0009\\

\multirow{-4}{*}[1.5\dimexpr\aboverulesep+\belowrulesep+\cmidrulewidth]{\centering\arraybackslash 300} & 9 & 0.011 & 0.0077 & 0.0070 & 1.221 & 0.0009 & 0.0010\\
\cmidrule{1-8}
 & 0 & 0.008 & 0.0054 & 0.0052 & 0.761 & 0.0005 & 0.0006\\

 & 3 & 0.008 & 0.0055 & 0.0054 & 0.781 & 0.0006 & 0.0006\\

 & 6 & 0.008 & 0.0056 & 0.0054 & 0.813 & 0.0006 & 0.0006\\

\multirow{-4}{*}[1.5\dimexpr\aboverulesep+\belowrulesep+\cmidrulewidth]{\centering\arraybackslash 500} & 9 & 0.008 & 0.0058 & 0.0055 & 0.868 & 0.0006 & 0.0007\\
\bottomrule
\end{tabular}}
\end{table}

\clearpage
\subsection{\texorpdfstring{Coverage probability of $95\%$ confidence intervals}{Coverage probability of 95\% confidence intervals}}
\label{AppB3}

\vspace{-1em}

\begin{table}[!ht]
\centering
\caption{Coverage probability for parameter $\gamma$}
\centering
\resizebox{0.74\linewidth}{!}{
\begin{tabular}[t]{cccccccc}
\toprule
\multicolumn{2}{c}{ } & \multicolumn{4}{c}{Gamma model} & \multicolumn{2}{c}{Logistic regression}  \\
\cmidrule(l{3pt}r{3pt}){3-6} \cmidrule(l{3pt}r{3pt}){7-8} 
\multicolumn{2}{c}{ } & \multicolumn{2}{c}{$\phi=0.02$} & \multicolumn{2}{c}{$\phi=1$} & \multicolumn{2}{c}{ } \\
\cmidrule(l{3pt}r{3pt}){3-4}
\cmidrule(l{3pt}r{3pt}){5-6} 
N & $\mu$ & $\gamma=-0.05$ & $\gamma=0.02$ &  $\gamma=-0.05$ & $\gamma=0.02$ & $\gamma=-0.1$ & $\gamma=0.1$\\
\midrule
 & 0 & 0.96 & 0.96 & 0.96 & 0.96 & 0.95 & 0.95\\

 & 3 & 0.95 & 0.94 & 0.95 & 0.95 & 0.94 & 0.94\\

 & 6 & 0.89 & 0.93 & 0.94 & 0.95 & 0.94 & 0.94\\

\multirow{-4}{*}[1.5\dimexpr\aboverulesep+\belowrulesep+\cmidrulewidth]{\centering\arraybackslash 100} & 9 & 0.85 & 0.93 & 0.92 & 0.93 & 0.94 & 0.94\\
\cmidrule{1-8}
 & 0 & 0.94 & 0.94 & 0.95 & 0.95 & 0.93 & 0.93\\

 & 3 & 0.93 & 0.94 & 0.95 & 0.95 & 0.93 & 0.93\\

 & 6 & 0.91 & 0.94 & 0.94 & 0.94 & 0.93 & 0.93\\

\multirow{-4}{*}[1.5\dimexpr\aboverulesep+\belowrulesep+\cmidrulewidth]{\centering\arraybackslash 300} & 9 & 0.88 & 0.94 & 0.93 & 0.95 & 0.92 & 0.92\\
\cmidrule{1-8}
 & 0 & 0.93 & 0.93 & 0.95 & 0.95 & 0.94 & 0.94\\

 & 3 & 0.92 & 0.93 & 0.95 & 0.95 & 0.95 & 0.95\\

 & 6 & 0.91 & 0.93 & 0.94 & 0.94 & 0.92 & 0.92\\

\multirow{-4}{*}[1.5\dimexpr\aboverulesep+\belowrulesep+\cmidrulewidth]{\centering\arraybackslash 500} & 9 & 0.88 & 0.94 & 0.94 & 0.95 & 0.94 & 0.94\\
\bottomrule
\end{tabular}}
\end{table}

\begin{table}[!ht]
\centering
\caption{Coverage probability for parameter $\alpha$}
\centering
\resizebox{0.74\linewidth}{!}{
\begin{tabular}[t]{cccccccc}
\toprule
\multicolumn{2}{c}{ } & \multicolumn{4}{c}{Gamma model} & \multicolumn{2}{c}{Logistic regression}  \\
\cmidrule(l{3pt}r{3pt}){3-6} \cmidrule(l{3pt}r{3pt}){7-8} 
\multicolumn{2}{c}{ } & \multicolumn{2}{c}{$\phi=0.02$} & \multicolumn{2}{c}{$\phi=1$} & \multicolumn{2}{c}{ } \\
\cmidrule(l{3pt}r{3pt}){3-4}
\cmidrule(l{3pt}r{3pt}){5-6} 
N & $\mu$ & $\gamma=-0.05$ & $\gamma=0.02$ &  $\gamma=-0.05$ & $\gamma=0.02$ & $\gamma=-0.1$ & $\gamma=0.1$\\
\midrule
 & 0 & 0.95 & 0.95 & 0.94 & 0.94 & 0.93 & 0.93\\

 & 3 & 0.92 & 0.94 & 0.95 & 0.95 & 0.93 & 0.93\\

 & 6 & 0.90 & 0.94 & 0.95 & 0.94 & 0.94 & 0.94\\

\multirow{-4}{*}[1.5\dimexpr\aboverulesep+\belowrulesep+\cmidrulewidth]{\centering\arraybackslash 100} & 9 & 0.87 & 0.93 & 0.94 & 0.94 & 0.93 & 0.93\\
\cmidrule{1-8}
 & 0 & 0.95 & 0.95 & 0.95 & 0.95 & 0.94 & 0.94\\

 & 3 & 0.94 & 0.93 & 0.95 & 0.96 & 0.93 & 0.93\\

 & 6 & 0.91 & 0.94 & 0.95 & 0.95 & 0.94 & 0.94\\

\multirow{-4}{*}[1.5\dimexpr\aboverulesep+\belowrulesep+\cmidrulewidth]{\centering\arraybackslash 300} & 9 & 0.89 & 0.94 & 0.94 & 0.94 & 0.92 & 0.92\\
\cmidrule{1-8}
 & 0 & 0.94 & 0.94 & 0.96 & 0.96 & 0.93 & 0.93\\

 & 3 & 0.92 & 0.94 & 0.96 & 0.96 & 0.93 & 0.93\\

 & 6 & 0.93 & 0.91 & 0.96 & 0.97 & 0.93 & 0.93\\

\multirow{-4}{*}[1.5\dimexpr\aboverulesep+\belowrulesep+\cmidrulewidth]{\centering\arraybackslash 500} & 9 & 0.91 & 0.93 & 0.95 & 0.96 & 0.94 & 0.94\\
\bottomrule
\end{tabular}}
\end{table}

\clearpage

\begin{table}[!ht]
\centering
\caption{MCSE of coverage probability for parameter $\gamma$}
\centering
\resizebox{0.75\linewidth}{!}{
\begin{tabular}[t]{cccccccc}
\toprule
\multicolumn{2}{c}{ } & \multicolumn{4}{c}{Gamma model} & \multicolumn{2}{c}{Logistic regression}  \\
\cmidrule(l{3pt}r{3pt}){3-6} \cmidrule(l{3pt}r{3pt}){7-8} 
\multicolumn{2}{c}{ } & \multicolumn{2}{c}{$\phi=0.02$} & \multicolumn{2}{c}{$\phi=1$} & \multicolumn{2}{c}{ } \\
\cmidrule(l{3pt}r{3pt}){3-4}
\cmidrule(l{3pt}r{3pt}){5-6} 
N & $\mu$ & $\gamma=-0.05$ & $\gamma=0.02$ &  $\gamma=-0.05$ & $\gamma=0.02$ & $\gamma=-0.1$ & $\gamma=0.1$\\
\midrule
 & 0 & 0.009 & 0.009 & 0.009 & 0.009 & 0.010 & 0.010\\

 & 3 & 0.010 & 0.010 & 0.010 & 0.010 & 0.010 & 0.010\\

 & 6 & 0.014 & 0.011 & 0.010 & 0.010 & 0.010 & 0.010\\

\multirow{-4}{*}[1.5\dimexpr\aboverulesep+\belowrulesep+\cmidrulewidth]{\centering\arraybackslash 100} & 9 & 0.016 & 0.012 & 0.012 & 0.011 & 0.010 & 0.010\\
\cmidrule{1-8}
 & 0 & 0.010 & 0.010 & 0.010 & 0.010 & 0.011 & 0.011\\

 & 3 & 0.011 & 0.010 & 0.010 & 0.010 & 0.011 & 0.011\\

 & 6 & 0.013 & 0.011 & 0.011 & 0.011 & 0.011 & 0.011\\

\multirow{-4}{*}[1.5\dimexpr\aboverulesep+\belowrulesep+\cmidrulewidth]{\centering\arraybackslash 300} & 9 & 0.015 & 0.010 & 0.011 & 0.010 & 0.012 & 0.012\\
\cmidrule{1-8}
 & 0 & 0.011 & 0.011 & 0.009 & 0.009 & 0.010 & 0.010\\

 & 3 & 0.012 & 0.011 & 0.009 & 0.010 & 0.010 & 0.010\\

 & 6 & 0.013 & 0.012 & 0.011 & 0.011 & 0.012 & 0.012\\

\multirow{-4}{*}[1.5\dimexpr\aboverulesep+\belowrulesep+\cmidrulewidth]{\centering\arraybackslash 500} & 9 & 0.015 & 0.011 & 0.011 & 0.010 & 0.011 & 0.011\\
\bottomrule
\end{tabular}}
\end{table}

\begin{table}[!ht]
\centering
\caption{MCSE of coverage probability for parameter $\alpha$}
\centering
\resizebox{0.75\linewidth}{!}{
\begin{tabular}[t]{cccccccc}
\toprule
\multicolumn{2}{c}{ } & \multicolumn{4}{c}{Gamma model} & \multicolumn{2}{c}{Logistic regression}  \\
\cmidrule(l{3pt}r{3pt}){3-6} \cmidrule(l{3pt}r{3pt}){7-8} 
\multicolumn{2}{c}{ } & \multicolumn{2}{c}{$\phi=0.02$} & \multicolumn{2}{c}{$\phi=1$} & \multicolumn{2}{c}{ } \\
\cmidrule(l{3pt}r{3pt}){3-4}
\cmidrule(l{3pt}r{3pt}){5-6} 
N & $\mu$ & $\gamma=-0.05$ & $\gamma=0.02$ &  $\gamma=-0.05$ & $\gamma=0.02$ & $\gamma=-0.1$ & $\gamma=0.1$\\
\midrule
 & 0 & 0.010 & 0.010 & 0.010 & 0.010 & 0.012 & 0.011\\

 & 3 & 0.012 & 0.011 & 0.010 & 0.010 & 0.011 & 0.011\\

 & 6 & 0.013 & 0.011 & 0.010 & 0.010 & 0.011 & 0.011\\

\multirow{-4}{*}[1.5\dimexpr\aboverulesep+\belowrulesep+\cmidrulewidth]{\centering\arraybackslash 100} & 9 & 0.015 & 0.011 & 0.011 & 0.011 & 0.012 & 0.012\\
\cmidrule{1-8}
 & 0 & 0.010 & 0.010 & 0.010 & 0.010 & 0.011 & 0.011\\

 & 3 & 0.011 & 0.011 & 0.010 & 0.009 & 0.011 & 0.011\\

 & 6 & 0.013 & 0.010 & 0.009 & 0.010 & 0.011 & 0.011\\

\multirow{-4}{*}[1.5\dimexpr\aboverulesep+\belowrulesep+\cmidrulewidth]{\centering\arraybackslash 300} & 9 & 0.014 & 0.011 & 0.011 & 0.010 & 0.012 & 0.012\\
\cmidrule{1-8}
 & 0 & 0.011 & 0.011 & 0.009 & 0.009 & 0.012 & 0.012\\

 & 3 & 0.012 & 0.010 & 0.009 & 0.009 & 0.012 & 0.012\\

 & 6 & 0.011 & 0.013 & 0.008 & 0.008 & 0.011 & 0.011\\

\multirow{-4}{*}[1.5\dimexpr\aboverulesep+\belowrulesep+\cmidrulewidth]{\centering\arraybackslash 500} & 9 & 0.013 & 0.011 & 0.009 & 0.009 & 0.011 & 0.011\\
\bottomrule
\end{tabular}}
\end{table}

\clearpage
\subsection{Computation time}
\label{AppB4}

\begin{table}[!ht]
\centering
\caption{Summary grouped by sample size and mean width $\mu$ of the observed intervals of the computation times (h:min:sec) and the number of augmented Turnbull intervals $m$. Each row comprises 6 scenarios, with $500$ repetitions per scenario. \vspace{3ex}}
\centering
\resizebox{\ifdim\width>\linewidth\linewidth\else\width\fi}{!}{
\begin{tabular}[t]{ccccccc}
\toprule
\multicolumn{2}{c}{ } & \multicolumn{2}{c}{Time per model} & \multicolumn{2}{c}{m} & \multicolumn{1}{c}{ } \\
\cmidrule(l{3pt}r{3pt}){3-4} \cmidrule(l{3pt}r{3pt}){5-6}
N & mu & mean & sd & mean & sd & non-convergences\\
\midrule
 & 0 & 00:00:04 & 00:00:01 & 100 & 0.0 & 6\\

 & 3 & 00:03:53 & 00:00:47 & 186 & 3.5 & 0\\

 & 6 & 00:06:23 & 00:01:17 & 178 & 4.4 & 0\\

\multirow{-4}{*}[1.5\dimexpr\aboverulesep+\belowrulesep+\cmidrulewidth]{\centering\arraybackslash 100} & 9 & 00:08:04 & 00:01:38 & 170 & 4.7 & 0\\
\cmidrule{1-7}
 & 0 & 00:00:19 & 00:00:07 & 300 & 0.0 & 5\\

 & 3 & 00:32:10 & 00:05:55 & 563 & 5.6 & 0\\

 & 6 & 00:52:00 & 00:09:54 & 535 & 7.4 & 0\\

\multirow{-4}{*}[1.5\dimexpr\aboverulesep+\belowrulesep+\cmidrulewidth]{\centering\arraybackslash 300} & 9 & 01:05:15 & 00:12:39 & 511 & 8.3 & 0\\
\cmidrule{1-7}
 & 0 & 00:00:45 & 00:00:15 & 500 & 0.0 & 1\\

 & 3 & 01:26:58 & 00:15:31 & 940 & 6.9 & 0\\

 & 6 & 02:25:01 & 00:27:46 & 893 & 8.9 & 0\\

\multirow{-4}{*}[1.5\dimexpr\aboverulesep+\belowrulesep+\cmidrulewidth]{\centering\arraybackslash 500} & 9 & 03:03:03 & 00:35:31 & 852 & 10.1 & 0\\
\bottomrule
\end{tabular}}
\end{table}

% \end{document}

\end{document}